%% file: sample.tex
\documentclass[11pt, journal, twocolumn]{IEEEtran}

\usepackage[margin=2.2cm]{geometry}
\ifCLASSINFOpdf
\else
\fi

\usepackage{amssymb}
\usepackage{amsthm}
\usepackage{float}
\usepackage[cmex10]{amsmath} 
\usepackage{graphicx}
\usepackage[ruled,vlined,linesnumbered]{algorithm2e}
\usepackage[caption=false]{subfig}
\newtheorem{theorem}{\bf{Theorem}}
\newtheorem{definition}{\bf{Definition}}
\newtheorem{lemma}{\bf{Lemma}}
\newtheorem{corollary}{\bf{Corollary}}
\newtheorem{proposition}{\bf{Proposition}}
\newtheorem{remark}{\bf{Remark}}
\newtheorem{example}{\bf{Example}}
\newtheorem{assumption}{\bf{Assumption}}
\usepackage{xcolor}
\usepackage{hyperref}
\usepackage{float}
\usepackage[numbers]{natbib}

\hypersetup{
    colorlinks=true,
    linkcolor=blue,
    filecolor=magenta,      
    urlcolor=cyan,
    pdftitle={Overleaf Example},
    pdfpagemode=FullScreen,
    }
\usepackage{xurl}
\usepackage{xcolor}

\newcommand{\ERM}{{\sf{ERM}}}

\newcommand{\E}[1]{\mathbb E\left[#1\right]}
\newcommand{\Esub}[2]{\mathbb E_{#1}\left[#2\right]}
\newcommand{\gen}{\textup{gen}}

\DeclareMathOperator*{\argmin}{argmin}

\allowdisplaybreaks

\begin{document}
%
\title{On the Generalization for Transfer Learning: An Information-Theoretic Analysis}
%


%
%

\author{Xuetong Wu,\thanks{This work is an extended version of the preliminary work \cite{wu2020information} presented in ISIT2020 conference. 
}, \thanks{Xuetong Wu, Jonathan H. Manton and Jingge Zhu are with the Department of Electrical and Electronic Engineering, University of Melbourne, VIC, 3010, Australia. E-mail: wfyitf@gmail.com; \{jmanton, jingle.zhu\}@unimelb.edu.au;}~
Jonathan H. Manton,~\IEEEmembership{Fellow,~IEEE,}~ Uwe Aickelin,\thanks{Uwe Aickelin is with the School of Computing and Information Systems, University of Melbourne, VIC, 3010, Australia. E-mail: uwe.aickelin@unimelb.edu.au.}
Jingge Zhu,~\IEEEmembership{Member,~IEEE}
}
\maketitle

\begin{abstract}
Transfer learning, or domain adaptation, is concerned with machine learning problems in which training and testing data come from possibly different probability distributions. In this work, we give an information-theoretic analysis of the generalization error and excess risk of transfer learning algorithms. Our results suggest, perhaps as expected, that the Kullback-Leibler (KL) divergence $D(\mu\|\mu')$ plays an important role in the characterizations where $\mu$ and $\mu'$ denote the distribution of the training data and the testing data, respectively. Specifically, we provide generalization error and excess risk upper bounds for learning algorithms where data from both distributions are available in the training phase. Recognizing that the bounds could be sub-optimal in general, we provide improved excess risk upper bounds for a certain class of algorithms, including the empirical risk minimization (ERM) algorithm, by making stronger assumptions through the \textit{central condition}.  To demonstrate the usefulness of the bounds, we further extend the analysis to the Gibbs algorithm and the noisy stochastic gradient descent method. We then generalize the mutual information bound with other divergences such as $\phi$-divergence and Wasserstein distance, which may lead to tighter bounds and can handle the case when $\mu$ is not absolutely continuous with respect to $\mu'$. Several numerical results are provided to demonstrate our theoretical findings. Lastly, to address the problem that the bounds are often not directly applicable in practice due to the absence of the distributional knowledge of the data, we develop an algorithm (called InfoBoost) that dynamically adjusts the importance weights for both source and target data based on certain information measures. The empirical results show the effectiveness of the proposed algorithm.

\end{abstract}



%
\input{1.introduction}
\input{2.literature}

\input{3.mainresult}
\input{4.extension}

\input{5.divergence}
\input{6.example}
\input{7.conclusion}
\input{8.appendix}
\bibliographystyle{IEEEtranN} 
\bibliography{sample}



\ifCLASSOPTIONcaptionsoff
  \newpage
\fi

\begin{IEEEbiographynophoto}{Xuetong Wu}
Xuetong Wu received a B.S. degree in 2016 from Fudan University, Shanghai, China, and completed his M.S. and Ph.D. degrees in 2018 and 2023, respectively, in the Department of Electrical and Electronic Engineering at the University of Melbourne. His research interests include machine learning, transfer learning, and information theory.
\end{IEEEbiographynophoto}
\begin{IEEEbiographynophoto}{Jonathan H. Manton}
Professor Jonathan Manton holds a Distinguished Chair at the University of Melbourne with the title Future Generation Professor. He is also an adjunct professor in the Mathematical Sciences Institute at the Australian National University, a Fellow of the Australian Mathematical Society (FAustMS) and a Fellow of the Institute of Electrical and Electronics Engineers (FIEEE). He received his Bachelor of Science (mathematics) and Bachelor of Engineering (electrical) degrees in 1995 and his Ph.D. degree in 1998, all from The University of Melbourne, Australia. In 2005 he became a full Professor in the Research School of Information Sciences and Engineering (RSISE) at the Australian National University. From mid-2006 till mid-2008, he was on secondment to the Australian Research Council as Executive Director, Mathematics, Information and Communication Sciences. He has served as an Associate Editor for the IEEE Transactions on Signal Processing and a Lead Guest Editor for the IEEE Transactions on Selected Topics in Signal Processing. He has been a Committee Member of the IEEE Signal Processing for Communications (SPCOM) Technical Committee, and a Committee Member on the Mathematics Panel for the ACT Board of Senior Secondary Studies in Australia. Currently he is a Committee Member of the IEEE Machine Learning for Signal Processing (MLSP) Technical Committee, and is the Signal Processing Chapter Chair for the IEEE Victorian Section. Awards include a prestigious Queen Elizabeth II Fellowship and a Future Summit Australian Leadership Award.

His principal fields of interest are Mathematical Systems Theory (including Signal Processing and Optimisation), Geometry and Topology (Differential and Algebraic), and Learning and Computation (including Systems Biology, Systems Neuroscience and Machine Learning).
\end{IEEEbiographynophoto}
\begin{IEEEbiographynophoto}{Uwe Aickelin}
Prof Uwe Aickelin received the Ph.D. degree from the University of Wales, U.K. He is currently a Professor and the Head of the School of Computing and Information Systems, University of Melbourne. His current research interests include artificial intelligence (modelling and simulation), data mining and machine learning (robustness and uncertainty), decision support and optimization (medicine and digital economy), and health informatics (electronic healthcare records).
\end{IEEEbiographynophoto}
\begin{IEEEbiographynophoto}{Jingge Zhu}
Jingge Zhu received the B.S. degree and M.S. degree in electrical engineering from Shanghai Jiao Tong University, Shanghai, China, the Dipl.-Ing. degree in Technische Informatik from Technische Universit\"{a}t Berlin, Berlin, Germany and the Doctorat \`{e}s Sciences degree from the Ecole Polytechnique F\'{e}d\'{e}rale (EPFL), Lausanne, Switzerland. He was a post-doctoral researcher at the University of California, Berkeley.  He is now a Senior Lecturer at the University of Melbourne, Australia. His research interests include information theory with applications in communication systems and machine learning. 

Dr. Zhu received the Discovery Early Career Research Award (DECRA) from the Australian Research Council in 2021, the IEEE Heinrich Hertz Award for Best Communications Letters in 2013, and the Early Postdoc. Mobility Fellowship from Swiss National Science Foundation in 2015 and the Chinese Government Award for Outstanding Students Abroad in 2016.
\end{IEEEbiographynophoto}



\end{document}

%% file: 1.introduction.tex
\section{Introduction}
A learning algorithm is viewed as a stochastic mapping, which takes training data as its input and produces a hypothesis as the output. The output hypothesis will then be used on data not seen before (testing data). Most machine learning methods focus on the setup where the training and testing data are drawn from the same distribution. Transfer learning, or domain adaptation, is concerned with machine learning problems where training and testing data come from possibly different distributions. This setup is of particular interest in real-world applications, as in many cases we often have easy access to a substantial amount of labelled (or unlabelled) data from a distribution $\mu$, namely the source domain, on which our learning algorithm trains, but wish to use the trained hypothesis for data coming from a different distribution $\mu'$, namely the target domain, from which we have limited data for training. Generalization error, a crucial measure of learning performance, is defined as the difference between the empirical loss and the population loss for a given hypothesis and indicates if the hypothesis suffers from overfitting or underfitting for the target domain of interest. Conventionally, many bounding techniques are proposed under different conditions and assumptions for traditional machine learning methods. For example, \citet{vapnik2015uniform} proposed VC-dimension which describes the richness of a hypothesis class for generalization ability.  The notion of ``algorithmic stability" was introduced in \cite{kearns1997algorithmic} and \cite{devroye1979distributiona} for bounding the generalization error, by examining if a single training sample has a significant effect on the expected loss. PAC-Bayes bounds are a class of algorithm-dependent bounds first introduced by \citet{mcallester1999some}. \citet{xu2012robustness} develop another notion, namely the robustness for the generalization. Unlike the stability, the robustness conveys geometric intuition and it can be extended to non-standard setups such as Markov chain or quantile loss, facilitating new bounds on generalization. 

While upper bounds on generalization error are classical results in statistical learning theory, only a relatively small number of papers are devoted to this problem for transfer learning algorithms. To mention a few, \citet{ben-david_theory_2010} defined the generalization error for transfer learning problems and gave its VC dimension-style bounds for classification problems with the proposed $\mathcal{H}$-divergence. \citet{blitzer_learning_2008} studied transfer learning problems with a similar setup and obtained upper bounds in terms of Rademacher complexity.  \citet{long_adaptation_2013} developed a more general framework for transfer learning and the error is bounded by the distribution difference and adaptability of the hypothesis output. \citet{dai_boosting_2007} and \citet{eaton2011selective} proposed two boosting-based transfer learning algorithms that emphasize the significance of source instances in transfer boosting, and it can be shown that the error bounds are increasingly smaller with iterations increasing. \citet{zhang_bridging_2019} proposed an extension of theories in \cite{ben-david_theory_2010} to multiclass classification in transfer learning with a novel domain divergence called margin disparity discrepancy.  Compared with traditional learning problems, the generalization error of transfer learning additionally takes the distribution divergence between the source and target into account and how to evaluate this "domain shift" is non-trivial. Traditional bounds used in transfer learning theory, such as VC dimension or Rademacher complexity, typically do not take the learning algorithm into account. They primarily focus on characterizing the complexity of the whole hypothesis class or making assumptions about the algorithms and loss functions. Consequently, the results often provide an overly pessimistic view of the learning problem, especially when the data or the algorithm has some underlying structure that could be exploited to simplify the learning task. On the other hand, some of the works focus specifically on certain transfer learning algorithms and loss functions \cite{ben-david_theory_2010,dai_boosting_2007}, and the resulting bounds on the generalization error are not universally applicable. Moreover, most bounds mentioned above are only concerned with the hypothesis or the algorithm solely, which fails to account for the complex interplay between data distribution, model hypothesis, and learning algorithm in determining the generalization error. 
    
To characterize the intrinsic nature of a learning algorithm, some recent works have shown that the generalization error can be upper bounded using information-theoretic quantities. In particular, \citet{russo_controlling_2016} study the connection between mutual information and generalization error in the context of adaptive learning. The authors show that the mutual information between the training data and the output hypothesis can be used to upper bound the generalization error. One nice property of this framework is that the mutual information bound explicitly explores the dependence between training data and the output hypothesis, in contrast to the bounds obtained by traditional methods with VC dimension and Rademacher complexity \cite{shalev-shwartz_understanding_2014}. This paper exploits the information-theoretic framework in the transfer learning settings and derives the upper bounds for the generalization error.  
To summarize our main contributions, we highlight the following points.

\begin{itemize}
\item[1.] We give an information-theoretic upper bound on the generalization error and the excess risk of transfer learning algorithms where training and testing data come from different distributions. Specifically, the upper bound involves the mutual information $I(W;Z)$ where $W$ denotes the output hypothesis and $Z$ denotes the training instance, and an additional term $D(\mu||\mu')$ that captures the effect of domain adaptation where $\mu, \mu'$ denotes the distributions of the source and the target domains, respectively. Such a result can be easily extended to multi-source transfer learning problems. To show the usefulness of the bounds, we further specialize the upper bounds on three specific algorithms: empirical risk minimization (ERM), the noisy stochastic gradient descent, and the Gibbs algorithm.

\item[2.] We observe that the mutual information upper bounds derived from existing methods are in general sub-optimal in terms of the convergence rate. To arrive at the correct learning rate, we further tighten the bounds for specific algorithms such as ERM and regularized ERM using the proposed $(\eta, c)$-central condition. We demonstrate through a few examples that, compared with the mutual information bounds previously derived, the new bound improves the convergence rate of the generalization error from $O(\sqrt{1/n})$ to $O(1/n)$ up to the domain divergence. We could further arrive at intermediate rates under the relaxed $(v,c)$-central conditions. 

\item[3.] We extend the results by using other types of divergences such as $\phi$-divergence and Wasserstein distance, which can be tighter than the mutual information bound under mild conditions. Such an extension also allows us to handle more general learning scenarios where the mutual information bound may be vacuous, e.g. when $\mu$ is not absolutely continuous with respect to $\mu'$.


\item[4.] Finally, we give a few examples of some simple transfer learning problems to validate our proposed bounds. However, in practical scenarios, these bounds are often not directly applicable due to the lack of knowledge of the data distributions. To address this, we propose a boosting-type algorithm called \texttt{InfoBoost} in which the importance weights for source and target data are adjusted adaptively in accordance with information measures. We also conduct several experiments on real datasets, and the empirical results show that, in most cases, our algorithm outperforms the state-of-the-art benchmarks.
\end{itemize}

The outline of this paper is structured as follows. We provided an in-depth review of the literature on the information-theoretic analysis for machine learning and transfer learning in Section~\ref{sec:literature}. Then, we formally formulate the transfer learning problem and give the main results in Section~\ref{sec:main}. We then specialize the bounds on the noisy iterative algorithms and the Gibbs algorithm following a similar intuition in Section~\ref{sec:applications}. In Section~\ref{sec:other_metric}, for scenarios where the mutual information-based bounds are vacuous, we extend the results to other divergences such as the $\phi$-divergence and the Wasserstein distance. Some examples are illustrated in Section~\ref{sec:example} to show the effectiveness of the bounds. Additionally, from a practical perspective, we propose an intuitive algorithm for transfer learning problems in Section~\ref{sec:example} that potentially works in real-world scenarios, inspired by the mutual information bounds. Section~\ref{sec:conclusion} concludes the paper with some remarks. 

%% file: 2.literature.tex
\section{Literature Review}\label{sec:literature}
In this section, we take a closer look at the related works that have shaped our understanding of machine learning and transfer learning. We will begin with an exploration of the information-theoretic analysis for traditional machine learning. Considering that the current information-theoretic analysis has a couple of primary limitations, such as the bounds being ineffective for deterministic algorithms and typically showing a slower rate of generalization error w.r.t. the data size, we will also delve into related works that improve the learning rate using various novel approaches.

\subsection{Information-theoretic analysis for machine learning}
Generally speaking, the \emph{generalization error} measures how well a learned model performs on previously unseen data, which is usually characterized by the gap between the training loss and testing loss. The generalization error of a learning algorithm lies in the core analysis of the statistical learning theory, the estimation of which becomes remarkably crucial in machine learning problems. Conventionally, many bounding techniques are proposed with different notions as aforementioned \cite{vapnik1999nature, bousquet_stability_2002,mcallester1999some,xu2012robustness,shalev-shwartz_understanding_2014,dwork2015preserving}. However, most bounds mentioned above are only concerned with the hypothesis or the algorithm solely. For example, VC-dimension methods care about the worst-case bound, which depends only on the hypothesis space. The stability methods only specify the properties of learning algorithms but do not require additional assumptions on hypothesis space. 

Information theory has been demonstrated to not only offer theoretical insights into the generalization error but also steer the intuition behind specific learning algorithms as it is a useful tool that extracts the statistical characterizations of the input data. Recently, \citet{russo_controlling_2016} and \citet{xu_information-theoretic_2017} studied a general statistical learning problem, and the authors give theoretical bounds for striking the balance between the data fit and generalization by controlling the mutual information between the output hypothesis and input data. In contrast to conventional generalization error bounds, such a result hinges on the data distribution, algorithms, and the learned hypothesis. Moreover, the dependence between the input data and the output hypothesis can be regarded as a measure that prompts the hypothesis class and algorithmic stability, hence recovering many other existing results such as VC dimension, algorithmic stability, and differential privacy. Similarly, \citet{Bassily2017} studied learning algorithms that only use a small amount of information from input samples by focusing on the capacity of the algorithm space. Furthermore, the mutual information bounds are improved in \cite{bu_tightening_2019}, and the authors proposed a tighter version with the mutual information between the output hypothesis and single data instance for the generalization error. In \cite{raginsky2016information}, several information-theoretic measures from an algorithmic stability perspective are used to upper bound the generalization error. In contrast to the mutual information, the generalization error bounds based on the Wasserstein distance are proposed in \cite{Lopez2019} and have been extended in \cite{wang2019information} with the total variation distance by exploiting the geometric nature of the Kantorovich-Rubinstein duality theorem. 

\subsection{Improvements on information-theoretic bounds}
However, there are several drawbacks to the information-theoretic bounds for typical machine learning problems. The first is that, for some deterministic algorithms, the mutual information quantities may be infinite and the bound will become vacuous. To tackle the infinity issue, instead of measuring information with the whole dataset, \citet{bu_tightening_2019} propose the bound based on the mutual information between the single instance and the hypothesis, which is finite even for deterministic algorithms. \citet{Negrea2019} propose generalization error bounds based on a subset of the whole dataset chosen uniformly at random. \citet{steinke2020reasoning} and \citet{Haghifam2020} improve the results with the conditional mutual information that is always finite by introducing a set of discrete (binary) random variables. \citet{Rodriguez-Galvez2021} have further tightened the bound using the random subset techniques under the Wasserstein distance. Another drawback is that these bounds are usually hard to estimate if the hypothesis and the data are of high dimensions. To make the quantity estimable, \citet{harutyunyan2021information} derive a novel generalization error bound that measures information with the predictions instead of the hypothesis produced by the training samples, which is significantly easier to estimate. The third drawback is that the bounds are usually sub-optimal, in terms of the convergence rate w.r.t. the sample size. In most of the relevant works, the convergence rate of the expected generalization error in traditional statistical learning problems is in the form of $O(\sqrt{\lambda/{n}})$ where $\lambda$ is some information-theoretic quantities such as the mutual information between the data sample and the learned hypothesis. However, such a learning rate is typically considered to be ``slow", compared to a ``fast rate" of $O(1/n)$ in many learning scenarios. Fast rate conditions are less investigated under the information-theoretic framework. Only a few works are dedicated to it, e.g., \citet{Grunwald2021pac} applies the conditional mutual information \cite{steinke2020reasoning} for the fast rate characterization under the PAC-Bayes framework and the results rely on prior knowledge of the hypothesis space. \citet{wu2020information} proposed the $(\eta, c)$-central condition for fast rate generalization error with the mutual information between the hypothesis and single data instance. \citet{bu2022characterizing} characterize the exact generalization error of the transfer learning under the Gibbs algorithm using the symmetric KL divergence and arrive at the fast rate under mild conditions. 

\subsection{Transfer learning bounds and comparisons}
The transfer learning problem is of particular interest in real-world applications, as in many cases we often have easy access to a substantial amount of labelled (or unlabelled) data from one distribution, on which our learning algorithm trains, but wish to use the learned hypothesis for data coming from a different distribution, from which we only have limited data for training. In practice, there will be perturbations or shifts in the distributions of the training and testing data, or obtaining the training data for some tasks can be very expensive and difficult such as robotics \cite{cote2017transfer, van2019sim}, medical images \cite{erickson2017machine, cheplygina2017transfer, morid2021scoping} and rare language translation \cite{zoph2016transfer, kocmi2018trivial}. Popular empirical risk minimization (ERM)--based methods usually minimize a convex combination of source and target data. The performance of the ERM has been initiated and investigated in works such as \cite{blitzer_learning_2008, mansour2009domain,ben-david_theory_2010,long_adaptation_2013}. These studies generally offer different bounds on the generalization error, contingent on specific domain divergences between source and target distributions. For example, the high-probability bounds for generalization error based on the error distance are presented in works such as Theorem 2 in \cite{blitzer_learning_2008}, Theorem 3 in \cite{ben-david_theory_2010}, and Theorem 2 in \cite{long_adaptation_2013}. Other works, such as Theorem 8 in \cite{mansour2009domain} and Theorem 3.7 in \cite{zhang_bridging_2019}, place their results within the context of the $\mathcal{A}$-discrepancy and disparity discrepancy with the Rademacher complexity of the hypothesis space. However, such discrepancies are typically linked to the hypothesis space only as they assess the largest discrepancy between two domains among all possible hypotheses given certain loss functions. It is also worth noting that these bounds are usually worst-case bounds as they work for any hypothesis in the hypothesis space. In specific learning regimes, such as in \cite{dai_boosting_2007}, the researchers investigate the boosting type transfer learning algorithm, introducing a learning bound in Theorem 3 that hinges on the iteration number and target sample error. However, this bound is closely tied to the chosen hypothesis space and does not capture the value of the source data. \citet{kuzborskij2013stability} consider the scenario in that only the source hypothesis induced from the source data is available. The authors conduct an algorithmic stability analysis on a class of hypothesis transfer learning problems with the regularized least squares algorithm, and the result suggests the relatedness of the source and target domains (e.g., how the source hypothesis performs on the target domains) determines the effectiveness of the transfer. \citet{germain2013pac} take the first attempt at the transfer learning problem under the PAC-Bayesian framework and provide a novel pseudo-distance on domain distributions, which leverages the idea from \cite{ben-david_theory_2010,mansour2009domain} by changing the pointwise disagreement to an averaging disagreement to fit into the PAC-Bayesian analysis. 

In our work, we introduce an information-theoretic framework for the ERM setup that presents multiple advantages over previous results. Specifically, our suggested bounds delve deeply into the dependence among data distribution, output hypothesis, and the algorithm itself. Such a bound may, in fact, be tighter than traditional bounds because the relationship between the dataset and the hypothesis can be seen as a metric that encourages algorithmic stability.  Our bound also contains the KL divergence between the source and target domains, effectively capturing the domain divergence. Moreover, whereas the majority of previous bounds exhibit sublinear convergence w.r.t. the sample size up to the domain divergence term, our bound can provide the correct linear convergence rate for some algorithms.

%% file: 3.mainresult.tex
\section{Problem formulation and main results}\label{sec:main}
We consider an instance space $\mathcal Z$, a hypothesis space $\mathcal W$, and a non-negative loss function $\ell: \mathcal W\times\mathcal Z\mapsto \mathbb R^+$. Let $\mu$  and $\mu'$ be two probability distributions defined on $\mathcal Z$, and assume that $\mu$ is absolutely continuous with respect to $\mu'$. In the sequel, the distribution $\mu$ is referred to as the \textit{source distribution}, and $\mu'$ as the \textit{target distribution}. We are given a set of training data with size $n$. More precisely, for a fixed number $\beta\in[0,1)$, we assume that $\beta n$ is an integer and the samples $S'=\{Z'_1,\ldots, Z'_{\beta n}\}$ are drawn IID from the target distribution $\mu'$, and the samples $S=\{Z_{\beta n+1},\ldots, Z_{n}\}$ are drawn IID from the source distribution $\mu$.  

In the setup of transfer learning, a learning algorithm is a (randomized) mapping from the training data $S, S'$ to a hypothesis $w\in\mathcal W$, characterized by a conditional distribution $P_{W|SS'}$, with the goal to find a hypothesis $w$ that minimizes the population risk with respect to the \textit{target distribution}
\begin{align}
L_{\mu'}(w):=\Esub{Z'\sim\mu'}{\ell(w, Z')},
\end{align}
where $Z'$ is distributed according to $\mu'$.  Notice that $\beta=0$ corresponds to the important case when we do not have any samples from the target distribution. Obviously, $\beta=1$ takes us back to the classical setup where training data comes from the same distribution as test data, which is not the focus of this paper. We call $(\mu', \mu, \ell, \mathcal{W}, \mathcal{A})$ as a transfer learning tuple.

\subsection{Empirical risk minimization}\label{subsec:erm}
In this section, we focus on one particular  \textit{empirical risk minimization} (ERM) algorithm. For a hypothesis $w\in\mathcal W$, the empirical risk of $w$ on the a training sequence $\tilde S:=\{Z_1,\ldots, Z_m\}$ is defined as
\begin{align}
\hat L(w,\tilde S):=\frac{1}{m}\sum_{i=1}^m\ell(w,Z_i).
\end{align}
Given samples $S$ and $S'$ from both distributions, it is natural to form an empirical risk function as a convex combination of the empirical risk induced by $S$ and $S'$  defined as
\begin{align}
\hat L_{\alpha}(w,S,S')&:=\frac{\alpha}{\beta n}\sum_{i=1}^{\beta n}\ell(w,Z'_i) \nonumber \\
&+ \frac{1-\alpha}{(1-\beta)n}\sum_{i=\beta n+1}^{n}\ell(w,Z_i) 
\label{eq:er}
\end{align}
for some weight parameter $\alpha\in[0,1]$ to be determined. We will use $\hat{L}_{\alpha}(w)$ interchangeably for $\hat L_{\alpha}(w,S,S')$ to simplify the notation, and then we define the ERM solution by
\begin{align}
w_{\ERM} :=\text{argmin}_w \hat L_{\alpha}(w,S,S').
\end{align}
Accordingly, we define the optimal hypothesis with respect to the distribution $\mu'$ as
\begin{align}
w^*=\text{argmin}_{w\in\mathcal W} L_{\mu'}(w).
\end{align}
We are interested in two quantities for any general machine learning algorithms, including the ERM algorithm. The first one is the  \textit{generalization error} defined as
\begin{align}
\gen(w, S, S'):= L_{\mu'}(w)-\hat L_{\alpha}(w,S,S'),
\label{eq:generalization_error}
\end{align}
namely the difference between the minimized empirical risk and the population risk of some hypothesis under the target distribution. We are also interested in the \textit{excess risk} defined as 
\begin{align}
R_{\mu'}(w)=L_{\mu'}(w)-L_{\mu'}(w^*), \label{eq:excess_expression}
\end{align}
which is the difference between the population risk of $w$ compared to that of the optimal hypothesis $w^*$. Notice that the excess risk of the ERM solution is related to the generalization error via the following upper bound:
\begin{align}
&L_{\mu'}(w_{\ERM})-L_{\mu'}(w^*)\nonumber\\
&=L_{\mu'}(w_{\ERM})-\hat L_{\alpha}(w_{\ERM},S,S')+\hat L_{\alpha}(w_{\ERM},S,S')\nonumber\\
&-\hat L_{\alpha}(w^*,S,S')+\hat L_{\alpha}(w^*,S,S')\nonumber\\
&-L_{\alpha}(w^*) +L_{\alpha}(w^*)-L_{\mu'}(w^*)\nonumber\\
&\leq \gen(w_{\ERM}, S, S')+\hat L_{\alpha}(w^*,S,S')-L_{\alpha}(w^*)\nonumber\\
& \quad +(1-\alpha)(L_{\mu}(w^*)-L_{\mu'}(w^*)),
\label{eq:excess_decompose}
\end{align}
where we have used the fact $\hat L_{\alpha}(w_{\ERM})-\hat L_{\alpha}(w^*)\leq 0$ by the definition of $w_{\ERM}$. For any $w\in\mathcal W$, the quantity $L_{\alpha}(w)$ in the above expression is defined as
\begin{align}
&L_{\alpha}(w):=(1-\alpha) L_{\mu}(w)+\alpha L_{\mu'}(w) \nonumber \\
&=(1-\alpha)\Esub{Z\sim\mu}{\ell(w,Z)}+\alpha \Esub{Z'\sim\mu'}{\ell(w,Z')}.
\end{align}

\subsection{Upper bound on the generalization error}
We consider the hypothesis $W$ as a random variable induced by the random samples $S, S'$ with some algorithm $\mathcal{A}$, characterized by a conditional distribution $P_{W|SS'}$. We will first study the expectation of the generalization error
\begin{align}
&\Esub{WSS'}{\gen(W, S, S')} \nonumber\\
&=\Esub{WSS'}{L_{\mu'}(W)-\hat L_{\alpha}(W,S, S')},
\label{eq:generalization_error_expect}
\end{align}
where the expectation is taken with respect to the distribution $P_{WSS'}$ defined as
\begin{align}
&P_{WSS'}(w,S,S') \nonumber\\
&=P_{W|SS'}(w|S,S')\prod_{i=1}^{\beta n} \mu'(z'_i)\prod_{i=\beta n+1}^{n}\mu(z_i).
\label{eq:excess}
\end{align}
Furthermore, we use $P_W$ to denote the marginal distribution of $W$ induced by the joint distribution $P_{WSS'}$. Following the characterization used in \cite{bu_tightening_2019}, the following theorem provides an upper bound on the expectation of the generalization error in terms of the mutual information between individual samples $Z_i$ and the hypothesis $W$ induced by a certain algorithm $P_{W|SS'}$, as well as the KL-divergence between the source and target distributions. As pointed out in \cite{bu_tightening_2019},  using mutual information between the hypothesis and individual samples $I(W;Z_i)$, in general, gives a tighter upper bound than using  $I(W;S)$.
\begin{theorem}[Generalization error of generic algorithms]
Assume that the hypothesis $W$ is distributed over $P_{W}$ induced by some algorithm, and the cumulant generating function of the random variable $\ell(W, Z)-\E{\ell(W,Z)}$ is upper bounded by $\psi(\lambda)$ in the interval $(b_{-},b_{+})$ under the product distribution $P_W\otimes\mu'$ for some $b_{-}<0$ and $b_{+}>0$.  Then for any $\beta>0$, the expectation of the  generalization error  in (\ref{eq:generalization_error_expect}) is upper bounded as
\begin{align*}
& \Esub{WSS'}{\gen(W, S, S')}\leq  \frac{\alpha}{\beta n}\sum_{i=1}^{\beta n}\psi^{*-1}_{-}(I(W;Z'_i)) \\
& +\frac{(1-\alpha)}{(1-\beta)n}\sum_{i=\beta n+1}^{n}\psi^{*-1}_{-}(I(W;Z_i)+D(\mu||\mu')), \\
& -\Esub{WSS'}{\gen(W, S, S')}\leq \frac{\alpha}{\beta n}\sum_{i=1}^{\beta n}\psi^{*-1}_{+}(I(W;Z'_i)) \\
& +\frac{(1-\alpha)}{(1-\beta)n}\sum_{i=\beta n+1}^{n}\psi^{*-1}_{+}(I(W;Z_i)+D(\mu||\mu')),
\end{align*}
where we define
\begin{align*}
\psi^{*-1}_{-}(x) &:=\inf_{\lambda\in[0,-b_{-})}\frac{x+\psi(-\lambda)}{\lambda}, \\
\psi^{*-1}_{+}(x) &:=\inf_{\lambda\in[0,b_{+})}\frac{x+\psi(\lambda)}{\lambda}.
\end{align*}
\label{thm:exp_gen}
\end{theorem}

The proof can be found in Appendix~\ref{apd:main_theorem}. Notice that the bound above is not specific to the ERM algorithm but applicable to any hypothesis generated by a learning algorithm that has a bounded cumulant generating function.  Comparing the derived results with other transfer learning bounds such as Theorem 2 in \cite{blitzer_learning_2008}, and Theorem 3 in \cite{ben-david_theory_2010} that are applied in worst-case scenarios for any hypothesis, our proposed bound suggests that the generalization error inherently depends on the mutual information $I(W;Z_i)$ data distribution, the output hypothesis, and the algorithm itself, which may lead to a tighter characterization as it also takes the algorithm $P_{W|SS'}$ into account. From a stability point of view, good algorithms (ERM, for example) should ensure that the mutual information $I(W;Z_i)$ vanishes as $n\rightarrow\infty$, which has similar insights as \cite{bousquet_stability_2002} from the point view that a single instance should not affect the output hypothesis much. On the other hand, the domain shift is reflected in the KL-divergence $D(\mu\| \mu')$, as this term does not vanish when $n$ goes to infinity. The KL divergence is only dependent on the data distributions and irrelevant to the loss function and hypothesis. This is in contrast to other metrics that may depend on the hypothesis space and the prediction functions such as discrepancy distance \cite{mansour2009domain} and $\mathcal{H}\Delta \mathcal{H}$-divergence\cite{ben-david_theory_2010}.

\begin{remark}
It is natural to consider the problem of minimizing the upper bound with respect to the parameter $\alpha$ as it mediates the balance between performance on the source and target domains. This is, however, a non-trivial problem as the output hypothesis $W$ implicitly involves $\alpha$, and optimization of bound w.r.t. $\alpha$ is challenging. In principle, $\alpha$ should depend on the source and target data sizes and distribution differences between the two domains. Specifically, when the target domain data is abundant and the source and target domains are significantly different, $\alpha$ can be biased towards the target domain (i.e., $\alpha$ should be closer to 1). This is because the model can rely more on the target domain data for better learning performance. Conversely, when the target domain data is limited and the source and target domains are closely related, $\alpha$ should be biased towards the source domain (i.e., $\alpha$ should be closer to 0). This is because the model can benefit more from leveraging knowledge from the source domain to avoid overfitting the limited target data. 

The optimal value of $\alpha$ often requires empirical validation. Notice that if we care about the generalization error with respect to the population risk under the target distribution for $n\rightarrow\infty$ (the number of samples $S'$ from the target distribution also goes to infinity), the intuition says that we should choose $\alpha=1$, i.e. only using $S'$  from the target domain in the training process. On the other hand, if we only have limited data samples, $\alpha$ can be set to be approximate as $\beta$ as suggested in \cite{zhang2012generalization,ben-david_theory_2010} that this choice is shown to achieve a tighter bound empirically. Overall, we suggest that non-asymptotically, $\alpha$ should approach $1$ with the target sample size $\beta n$ increasing, say, $\alpha = 1 - O(\frac{1}{\beta n})$. In real practice, techniques such as cross-validation or grid search could also be used to tune $\alpha$ by evaluating model performance across a range of values. 
\end{remark}

The result in Theorem \ref{thm:exp_gen} does not cover the case $\beta=0$ (no samples from the target distribution).  However, it is easy to see that in this case, we should choose $\alpha=0$ in our ERM algorithm, and a corresponding upper bound is given as in the following corollary under the generic hypothesis. 

\begin{corollary}[Generalization error with source only]
Let $\beta=0$ so that we only have samples $S$ from the source distribution $\mu$. Let $P_{ W|S}$ be the conditional distribution characterizing the learning algorithm, which maps samples $S$ to a hypothesis $ W$. Under the same assumption as in Theorem \ref{thm:exp_gen},  the expected generalization error of $ W$ is upper bounded as
\begin{small}
\begin{align*}
\Esub{ WS}{\gen(W, S)}\leq \frac{1}{n} \sum_{i=1}^{n}\psi^{*-1}_{-}(I( W;Z_i)+D(\mu||\mu')),\\
-\Esub{ WS}{\gen(W, S)}\leq \frac{1}{n} \sum_{i=1}^{n}\psi^{*-1}_{+}(I( W;Z_i)+D(\mu||\mu')).
\end{align*}
\end{small}
\label{coro:gen_beta0}
\end{corollary}
The proof of this result is given in Appendix \ref{proof:corollary_beta0}. If the loss function $\ell(W, Z)$ is $r^2$-subgaussian, namely 
\begin{align*}
\log \E{ e^{\lambda(\ell(W,Z)-\E{\ell(W,Z))}}}\leq \frac{r^2\lambda^2}{2}
\end{align*}
for any $\lambda\in\mathbb R$ under the distribution $P_W\otimes \mu'$, the bound in Theorem \ref{thm:exp_gen} can be further simplified with $\psi^{*-1}(y) =\sqrt{2r^2y}$. In particular, if the loss function takes value in $[a,b]$, then $\ell(W,Z)$ is $\frac{(b-a)^2}{4}$-subgaussian. We give the following corollary for the subgaussian loss function.

\begin{corollary}[Generalization error for subgaussian loss functions]
Let $P_W$ be the marginal distribution induced by $S,S'$ and $P_{W|SS'}$ for some algorithm. If $\ell(W,Z)$ is $r^2$-subgaussian under the distribution $P_W\otimes \mu'$, then the expectation of the generalization error is upper bounded as
\begin{align}
&|\Esub{WSS'}{\gen(W, S, S')}|\leq  \frac{\alpha\sqrt{2r^2}}{\beta n}\sum_{i=1}^{\beta n}\sqrt{I(W;Z'_i)} \nonumber  \\
&+\frac{(1-\alpha)\sqrt{2r^2}}{(1-\beta)n}\sum_{i=\beta n+1}^{n}\sqrt{ (I(W;Z_i)+D(\mu||\mu'))}.
\end{align}
If $\beta=0$, for any hypothesis $W$ induced by $S$ and a learning algorithm $P_{W|S}$, we have the upper bound
\begin{small}
\begin{align}
|\Esub{WS}{\gen(W, S)}|\leq \frac{\sqrt{2r^2}}{n} \sum_{i=1}^{n}\sqrt{I(W;Z_i)+D(\mu||\mu')}.
\label{eq:subgaussian_beta0}
\end{align}
\end{small}
\label{coro:general_bound}
\end{corollary}
The above result follows directly from Theorem \ref{thm:exp_gen} and Corollary  \ref{coro:gen_beta0} by noticing that we can set $\psi(\lambda)=\frac{r^2\lambda^2}{2}, b_{-}=-\infty, b_{+}=\infty$ with the assumption that $\ell(W,Z)$ is $r^2$-subgaussian.

\begin{remark}
Using the chain rule of mutual information and the fact that $Z_i$'s are IID, we can relax the upper bound in (\ref{eq:subgaussian_beta0}) as 
\begin{align}
\Esub{WS}{\gen(W, S)}\leq \sqrt{2r^2\left(\frac{I( W;S)}{n} +D(\mu||\mu')\right)},
\end{align}
which recovers the result in \cite{xu_information-theoretic_2017} if $\mu=\mu'$. Moreover, we see that the effect of the ``domain shift"  is simply captured by the KL divergence between the source and the target distribution. 
\end{remark}
\begin{remark}
Even though we focus on the supervised learning setups for transfer learning with the convexly combined empirical loss. Such a framework can be easily extended to various different transfer learning setups such as multi-source transfer learning problems, pre-trained hypothesis setups, and unsupervised setups. To maintain the paper focus and prevent excessive use of notations, we have placed all detailed results and discussions in Appendix~\ref{apd:MSTLP},~\ref{apd:pre-train} and~\ref{apd:unlabel}, respectively.
\end{remark}

\subsection{Upper bound on the excess risk of ERM}
In this section, we focus on the case  $\beta>0$ and give a data-dependent upper bound on the excess risk defined in (\ref{eq:excess_expression}). To do this, we first define a distance quantity between the two divergent distributions as
\begin{align}
d_{\mathcal W}(\mu, \mu')=\sup_{w\in\mathcal W} |L_{\mu}(w)- L_{\mu'}(w)|.
\label{eq:empirical_distance}
\end{align}
The following theorem gives an upper bound on the excess risk.

\begin{theorem}[Excess risk of ERM]
Let $P_W$ be the marginal distribution induced by $S,S'$ and $P_{W|SS'}$ for the ERM algorithm, assume the loss function $\ell(W, Z)$ is $r^2$-subgaussian under the distribution $P_{W} \otimes \mu'$. Then the following inequality holds.
\begin{align}
&\mathbb{E}_{W}[R_{\mu'}(W_{\ERM})]\leq  \frac{\alpha\sqrt{2r^2}}{\beta n}\sum_{i=1}^{\beta n}\sqrt{I(W_{\ERM};Z'_i)}\nonumber \\
&+\frac{(1-\alpha)\sqrt{2r^2}}{(1-\beta)n}\sum_{i=\beta n+1}^{n}\sqrt{ (I(W_{\ERM};Z_i)+D(\mu||\mu'))}\nonumber \\
&+ (1-\alpha)d_{\mathcal W}(\mu, \mu').
\label{eq:excess_ERM}
\end{align}
Furthermore in the case when $\alpha = \beta=0$ (no samples from the target distribution $\mu'$),  the inequality becomes
\begin{small}
\begin{align}
\mathbb{E}_{W}[R_{\mu'}(W_{\ERM})] \leq  &\frac{\sqrt{2r^2}}{n}\sum_{i=1}^{n}\sqrt{ (I(W_{\ERM};Z_i)+D(\mu||\mu'))} \nonumber \\
&+ d_{\mathcal W}(\mu,\mu'). \label{eq:excess_erm_bound_mi}
\end{align}
\end{small}
\label{thm:excess}
\end{theorem}

The proof of this theorem is given in Appendix \ref{proof:thm_excess}. Note that $d_{\mathcal W}(\mu, \mu')$ is normally known as the integral probability metric, which is challenging to evaluate. \citet{sriperumbudur2012empirical} investigated the data-dependent estimation to compute the quantity using the Kantorovich metric, Dudley metric, and kernel distance, respectively. Another evaluation method is proposed in \cite{ben-david_theory_2010} to resolve the issue for classification problems. We verify our bound for the transfer learning with a toy example studied in \cite{bu_tightening_2019}. In Section~\ref{subsec:logit}, we also verify the bounds on the logistic regression transfer problem where the hypothesis cannot be explicitly calculated.

\begin{example}[Estimating the mean of Gaussian]\label{example:gaussian}
 Assume that  $S$ comes from the source distribution $\mu= \mathcal{N} (m, \sigma^2)$ and $S'$ comes form the target distribution $\mu'=\mathcal{N}(m', \sigma^2)$ where $m \neq m'$.  We define the loss function as
\begin{align*}
\ell(w,z)=(w-z)^2.
\end{align*}
For simplicity, we assume here that $\beta=0$. The empirical risk minimization (ERM) solution is obtained by minimizing $\hat L(w,S):=\frac{1}{ n}\sum_{i=1}^{n}(w-Z_i)^2$, where the solution is given by
\begin{align*}
W_{\ERM}=\frac{1}{n}\sum_{i=1}^nZ_i.
\end{align*}
To obtain the upper bound, we first notice that in this case
\begin{align*}
I(W_{\ERM};Z_i)=\frac{1}{2}\log \frac{n}{n-1},
\end{align*}
for all $i = 1,2,\cdots,n$. It is easy to see that the loss function $\ell(w, z_i)$ is non-central chi-square distribution $\chi^2(1)$ of $1$ degree of freedom with the variance of $\sigma^2_\ell = \frac{n+1}{n}\sigma^2$. 
Furthermore, the cumulant generating function can be bounded as for any $\lambda > 0$:
\begin{align*}
\log \mathbb{E}{e^{\lambda(\ell(w, z_i)-\mathbb{E}{\left[\ell(w, z_i))\right]}}} \leq \sigma^4_\ell\lambda^2 + \frac{2\lambda^2 \sigma^2_{\ell} (m-m')^2}{1 + 2\lambda\sigma^2_{\ell}}.
\end{align*}
Then it can be seen that the loss function is $\sqrt{2\sigma^4_\ell + 4\sigma^2_\ell(m-m')^2}$-subgaussian under the distribution $P_W\otimes \mu'$. Let $\sigma^2_{\ell'} =  2\sigma^4_\ell + 4\sigma^2_\ell(m-m')^2$, we reach at
\begin{small}
\begin{align*}
\mathbb{E}_{WS} \left[ \gen(W_{\ERM},S)\right] &\leq   \sqrt{\sigma^2_{\ell'}\log\frac{n}{n-1} + 2\sigma^2_{\ell'}D(\mu \| \mu')},
\end{align*}
\end{small}
where $D(\mu||\mu') = \frac{(m - m')^2}{2\sigma^2}$. Then the excess risk is upper bounded by,
\begin{align*}
&\mathbb{E}_{W}[L_{\mu'}(W_{\ERM}) - L_{\mu'}(w^*)] \leq\\
&\sqrt{\sigma^2_{\ell'}\log\frac{n}{n-1} + 2\sigma^2_{\ell'}D(\mu \| \mu')} 
+  d_{\mathcal{W}}\left(\mu, \mu^{\prime}\right).
\end{align*}
In this case, the generalization error and the excess risk of $W_{\ERM}$ can be calculated exactly to be
\begin{align*}
&\mathbb{E}_{WS}\left[ \hat L(W_{\ERM}, S)-L_{\mu'}(W_{\ERM}) \right]\\
&= \frac{2\sigma^2}{n} + 2\sigma^2 D(\mu||\mu'), \\
&\mathbb{E}_{W}[L_{\mu'}(W_{\ERM}) - L_{\mu'}(w^*)] = \frac{\sigma^2}{n} + 2\sigma^2 D(\mu||\mu').
\end{align*}
The derived excess risk bound approaches $\sqrt{4\sigma^4D(\mu\|\mu') + 16\sigma^4D(\mu\|\mu')^2} + d_{\mathcal{W}}(\mu, \mu')$ as $n\rightarrow \infty$ with a decay rate of  $O(1/\sqrt{n})$, which does not capture the bound asymptotically well as the true value should be $\sqrt{4\sigma^4D(\mu\|\mu')^2}$. Moreover, the hypothesis space-dependent quantity $d_{\mathcal{W}}(\mu, \mu')$ will be infinite if $w$ is unbounded, resulting in a vacuous bound. To further tighten the bound, we propose various ``easiness" conditions on the excess risk, which is shown to capture the true behavior up to a scaling factor in Section \ref{sec:fast}.
\end{example}

\subsection{Fast rate upper bound on the excess risk of ERM} \label{sec:fast}
As can be seen from previous sections, the convergence rate of the excess risk is in the form of $O\Big(\sqrt{\frac{\lambda}{\beta n}} + \sqrt{\frac{\lambda'}{(1-\beta) n}  + D(\mu\|\mu')} + d_{\mathcal{W}}(\mu,\mu')\Big)$ where $\lambda$ and $\lambda'$ is some information-theoretic related quantities such as the mutual information, $D(\mu\|\mu')$ and $d_{\mathcal{W}}(\mu,\mu')$ are the domain divergences between the source and target domains. However, such a learning rate is in general suboptimal.  

In this section, we give an alternative analysis for the excess risk under various ``easiness" conditions following the idea from \cite{Grunwald2021pac,wu2022fast}. With the new technique, we can show that the excess risk is characterized by the mutual information between the hypothesis and data instances and the rate will be of the form $O\left(\frac{\eta}{\beta n} + \frac{\eta'}{(1-\beta)n} +  D(\mu\|\mu')\right)$ where $\eta$ and $\eta'$ are some information-theoretic related quantities different from $\lambda$ and $\lambda'$ for specific learning algorithms such as empirical risk minimization and the hypothesis dependent term $d_{\mathcal{W}}(\mu,\mu')$ vanishes in the new bound. While the results presented in this section offer a more refined analysis for certain scenarios, it is important to note that they do not entirely supersede the 'slow-rate' results. For example, in the case of the Gaussian problem, we could show that the fast rate bound in Theorem~\ref{thm:central-transfer} does indeed provide a strictly better bound on the excess risk than the slow rate result in Theorem~\ref{thm:excess}. However, in a more general context, the two bounds on excess risk are not directly comparable, as they pertain to different underlying assumptions and scenarios. To simplify the notation, we also define the \emph{empirical} excess risk w.r.t. $w^*$ for some $w \in \mathcal{W}$ given the dataset $S$ as
\begin{align}
\hat{R}(w,S) := \hat L(w,S) - \hat{L}(w^*,S).
\end{align}
The empirical excess risk combined with both source and target is defined as
\begin{align}
\hat{R}_{\alpha}(w,S,S') :=  \hat L_{\alpha}(w,S,S') - \hat{L}_{\alpha}(w^*,S,S').
\end{align}
We also define the unexpected excess risk:
\begin{align}
    r(w,z_i) := \ell(w,z_i) - \ell(w^*,z_i)
\end{align}
for single instance $z_i$ as well. We further define the excess risk as
\begin{align}
    R_{\mu'}(w_\ERM) := L_{\mu'}(w_{\ERM}) - L_{\mu'}(w^*)].
\end{align}
Then the expected excess risk over $W_{\ERM}$ can be bounded by the following inequality:
\begin{align}
\mathbb{E}_{W}[R_{\mu'}(W_\ERM)] \leq \mathbb{E}_{WSS'}[\mathcal{E}(W_{\ERM}, S,S')], 
\end{align}
where the empirical excess risk gap is defined as
\begin{align}
    &\mathcal{E}(w, S,S') = \alpha \left({R_{\mu'}(w) - \hat{R}(w,S')}\right) \nonumber  \\ 
    & + (1-\alpha) \left({R_{\mu'}(w) - \hat{R}(w,S)}\right).
\end{align}
Here we have used the fact $\hat L_{\alpha}(W_{\ERM},S,S')-\hat L_{\alpha}(w^*,S,S')\leq 0$ by the definition of $W_{\ERM}$. Recently, numerous studies have focused on formulating fast information-theoretic bounds \cite{Grunwald2021pac,wu2022fast}, aiming to eliminate the square root present in the upper bounds proposed in the previous sections. This is mainly achieved by changing the assumptions from the loss function $\ell(w,z)$ to the unexpected excess risk $r(w,z)$. In this context, referencing the optimal hypothesis $w^*$ becomes crucial to further narrow down the concentration condition. This essential requirement is recognized as the central condition \cite{van2015fast,Grunwald2021pac,wu2022fast}. In our work, we utilize this condition and make subtle modifications to cater to our goals in transfer learning as follows.
\begin{definition}[$(\eta,c)$-Central Condition]
Let ${\eta}>0$ and $0 < c \leq 1$. We say that a transfer learning problem satisfies the expected $(\eta,c)$-central condition under the target distribution $\mu'$ if the following inequality holds for the optimal hypothesis $w^*$:
\begin{align}
&\log \mathbb{E}_{P_W\otimes \mu'} \left[e^{-{\eta}\left(\ell(W,Z)-\ell(w^*,Z)\right)}\right]  \leq \nonumber  \\ 
& -c\eta  \mathbb{E}_{P_W\otimes \mu'}\left[\ell(W,Z) - \ell(w^*,Z)\right]. \label{eq:eta_c}
\end{align}
where $P_W$ is the marginal distribution of the output hypothesis.
\end{definition}
This condition is similar to the central condition \cite{van2015fast,Grunwald2020}, where some assumptions are made on the small lower tail for the excess risk function with the exponential concavity, implying good concentration properties of the excess risk. Compared to the $\eta$-central condition defined in \cite[Def. 3.1]{van2015fast} which can be retrieved by setting $c = 0$, the R.H.S. of~(\ref{eq:eta_c}) will be negative and has tighter control of the tail behavior for some $c> 0$. We also point out that such a condition is indeed the key assumption for improving the rate by removing the square root, which also coincides with some well-known conditions that lead to a fast rate such as the Bernstein condition \cite{bartlett2006empirical,bartlett2006convexity, hanneke2016refined,mhammedi2019pac} and the central condition with the witness condition \cite{van2015fast,Grunwald2020}. Next, we provide several instances where the $(\eta,c)$-central condition is satisfied. While some of these examples are discussed in \cite{wu2022fast}, we revisit them here for the sake of completeness.
\begin{example}\label{coro:subexponential}
If $r(W, Z)$ is ($\nu^2$, $\alpha$)-sub-exponential under the distribution $P_W \otimes \mu'$, then the learning tuple satisfies $(\min(\frac{1}{\alpha}, \frac{\nu^2}{\mathbb{E}_{P_W \otimes \mu'}[r(W,Z)]}), \frac{1}{2})$-central condition. 
\end{example}
\begin{example}\label{coro:subgamma}
If $r(W, Z)$ is ($\nu^2$, $\alpha$)-sub-Gamma under the distribution $P_W \otimes \mu'$, then the learning tuple satisfies $(\frac{\mathbb{E}_{P_W\otimes\mu'}[r(W,Z)]}{\nu^2+\alpha \mathbb{E}_{P_W\otimes\mu'}[r(W,Z)]}, \frac{1}{2})$-central condition.
\end{example}
\begin{example}
Let $\gamma \in [0,1]$ and $B \geq 1$. Let $P_W$ is induced by $P_{WSS'}$. We assume that the \textbf{Bernstein condition} holds under the target distribution $P_W \otimes \mu'$. Namely, the following inequality holds for the optimal hypothesis $w^*$:
\begin{align*}
      &{\mathbb{E}}_{P_W \otimes \mu'}\left[\left(\ell\left(W, Z^{\prime}\right)-\ell\left(w^{*} , Z^{\prime}\right)\right)^{2}\right] \leq \\
      & B \left({\mathbb{E}}_{P_W \otimes \mu'}\left[\ell\left(W, Z^{\prime}\right)- \ell\left(w^{*} ; Z^{\prime}\right)\right]\right)^{\gamma}.
\end{align*}
In case, and if $\gamma = 1$ and $r(w,z_i)$ is bounded by $-b$ with some $b > 0$ for all $w$ and $z_i$, then the learning tuple also satisfies $(\min(\frac{1}{b}, \frac{1}{2B(e-2)}), \frac{1}{2})$-central condition.

The Bernstein condition is commonly identified as a way to describe the 'easiness' of a learning problem. The typical Bernstein condition necessitates that the inequality is satisfied for all $w \in \mathcal{W}$. However, in our case, we only require that the inequality is satisfied in expectation over $P_W$. In particular, consider the source-only case, $\gamma = 1$ corresponds to the easiest and the learning rate will be $O(\frac{1}{n} + cD(\mu\|\mu'))$ if $I(W;Z_i)$ is converging with the rate of $O(\frac{1}{n})$ for some leading constant $c$. For the bounded loss, the Bernstein condition will automatically hold with $\gamma = 0$ and it will recover the results in Corollary~\ref{coro:gen_beta0} with the rate of $O(\sqrt{\frac{1}{n} + cD(\mu\|\mu')})$.
\end{example}
\begin{example}
The second condition is the central condition with the witness condition \cite{van2015fast,Grunwald2020}, which also implies the $(\eta,c)$-central condition. We say the learning tuple $(\mu, \mu', \ell, \mathcal{W}, \mathcal{A})$ satisfies the $\eta$-central condition \cite{van2015fast,Grunwald2020} if for the optimal hypothesis $w^*$, the following inequality holds,
\begin{align*}
\mathbb{E}_{P_W\otimes \mu'}\left[e^{-{\eta}\left(\ell(W,Z)-\ell(w^*,Z)\right)}\right] \leq 1. 
\end{align*}
We also say the learning tuple  $(\mu, \mu', \ell, \mathcal{W}, \mathcal{A})$ satisfies the  $(u, c)$-witness condition \cite{Grunwald2020} if for constants $u > 0$ and $c \in (0,1]$, the following inequality holds.
\begin{align*}
    &\mathbb{E}_{P_W\otimes \mu'} [\left(\ell(W,Z)-\ell({w^{*}},Z) \right) \cdot   \mathbf{1}_{\left\{\ell(W,Z)-\ell({w^{*}},Z) \leq u \right\}}] \\ 
    &\geq c \mathbb{E}_{P_W\otimes \mu'}\left[\ell(W,Z) -\ell({w^{*}},Z) \right],
\end{align*}
where $ \mathbf{1}_{\{\cdot\}}$ denotes the indicator function. Then we have the following statement: If the learning tuple satisfies both $\eta$-central condition and $(u,c)$-witness condition, then the learning tuple also satisfies the $(\eta', \frac{c-\frac{c\eta'}{\eta}}{\eta' u +1})$-central condition for any $0 < \eta' < \eta$. 

The standard $\eta$-central condition \cite{van2015fast,mehta2017fast,Grunwald2020} is essential for deriving fast-rate bounds of generalization error. Some typical examples include exponential concave loss functions (including log-loss) with $\eta = 1$ \cite{mehta2017fast,zhu2020semi} and bounded loss functions with Massart noise condition with various $\eta$ \cite{van2015fast}. The witness condition \cite[Def. 12]{Grunwald2020} is applied to exclude cases where learnability is inherently impossible. This condition ensures that we exclude the poorly performing hypothesis $w$ with almost zero probability (even though it can still affect the expected loss), which we will never observe empirically. One trivial example is, if the excess risk is upper bounded by some constant $b$, we may always take $u=b$ and $c=1$ so that a witness condition is satisfied. 
\end{example}

With the definitions in place, we derive the fast rate bounds for the excess risk in transfer learning under the $(\eta, c)$-central condition as follows.
\begin{theorem}[Fast Rate with ($\eta,c$)-Central Condition]\label{thm:central-transfer}
Assume the learning problem with the ERM algorithm satisfies the expected ($\eta,c$)-central condition under the target distribution $\mu'$. Then for any $0 < \eta' \leq \eta$, the expected excess risk can be upper bounded by,
\begin{align}
&\mathbb{E}_{W}[R_{\mu'}(W_\ERM)] \leq  \frac{1}{c\eta'} \frac{\alpha}{\beta n} \sum_{i=1}^{\beta n} I(W_\ERM;Z'_i) \nonumber \\ 
&+ \frac{1}{c\eta'} \frac{1-\alpha}{(1-\beta)n} \sum_{i=\beta n+ 1}^{n} \left(I(W_\ERM;Z_i) + D(\mu\|\mu')\right). \label{eq:fast-erm}
\end{align}
More generally, for any algorithm $\mathcal{A}$ and any $W$ induced by the algorithm, if the expected ($\eta,c$)-central condition holds, we have that for any $0 < \eta' \leq \eta$,
\begin{align}
&\mathbb{E}_{W}[R_{\mu'}(W_\ERM)] \leq  \frac{1}{c\eta'} \frac{\alpha}{\beta n} \sum_{i=1}^{\beta n} I(W;Z'_i) \nonumber \\ 
&+ \frac{1}{c\eta'} \frac{1-\alpha}{(1-\beta)n} \sum_{i=\beta n +1}^{n} \left(I(W;Z_i) + D(\mu\|\mu')\right) \nonumber \\
&+ \frac{1}{c}\left[\alpha\mathbb{E}_{WS'}[\hat{R}(W,S')] + (1-\alpha)\mathbb{E}_{WS}[\hat{R}(W,S) \right]. \label{eq:tightened}
\end{align}
\end{theorem}
The proof can be found in Appendix~\ref{proof:central-transfer}. Now we compare (\ref{eq:fast-erm}) with the bound in Theorem~\ref{thm:excess} which we reproduce below,
\begin{align*}
&\mathbb{E}_{W}[R_{\mu'}(W_\ERM)] \leq \frac{\alpha \sqrt{2 r^{2}}}{\beta n} \sum_{i=1}^{n} \sqrt{I\left(W_{\ERM} ; Z'_{i}\right)} \nonumber \\ 
&+\frac{(1-\alpha) \sqrt{2 r^{2}}}{(1-\beta)n} \sum_{i=\beta n +1}^{n} \sqrt{\left(I\left(W_{\ERM} ; Z_{i}\right)+D\left(\mu \| \mu^{\prime}\right)\right)} \nonumber  \\
& + (1-\alpha) d_{\mathcal{W}}\left(\mu, \mu^{\prime}\right).
\end{align*}
In the new bound, the square root term is removed and we may achieve a faster rate for converging to the domain divergence $D(\mu\|\mu')$. Furthermore, the new bound does not contain the hypothesis space-dependent divergence term $d_{\mathcal{W}}\left(\mu, \mu^{\prime}\right)$, which might be very large or unbounded for certain distributions and hypothesis space. In the following Gaussian mean estimation example, we verify that the new bound is tighter than the previous bound and captures the true behaviour of the excess risk.

\begin{example}[Continuing from Example~\ref{example:gaussian}]
We continue to examine the bound in~Theorem~\ref{thm:central-transfer} that achieves the correct rate of convergence in the Gaussian mean estimation, which satisfies the $(\eta, c)$-central condition for certain $\eta$ and $c$. To this end, for a large sample size $n$, we check,
\begin{align*}
 &\log \mathbb{E}_{P_W \otimes \mu'}\left[e^{-\eta r(W,Z)}\right] = \log\sqrt{\frac{n}{n+2\eta\sigma^2(1-2\eta\sigma^2)}} \\ 
 &+ (2\eta^2\sigma^2 -\eta)(m-m')^2\leq -c\eta \left(\frac{\sigma^2}{n} + (m-m')^2\right).   
\end{align*}
From the above inequality, this learning problem satisfy the $(\eta, c)$-central condition for any $0 < \eta < \frac{1}{2\sigma^2}$ and any 
\begin{small}
\begin{align*}
    c &\leq -\frac{1}{\eta} \lim_{n\rightarrow \infty} \frac{\frac{1}{2}{\log\frac{n}{n+2\eta\sigma^2(1-2\eta\sigma^2)}} + (2\eta^2\sigma^2 -\eta)(m-m')^2}{\frac{\sigma^2}{n} + (m-m')^2} \\ 
    &= 1- 2\eta \sigma^2,
\end{align*}
\end{small}
by the quotient law of limits, where the choice of $c$ is independent of the sample size and thus does not affect the convergence rate. Therefore, take $\eta = \frac{1}{4\sigma^2}$ and $c = \frac{1}{2}$, the excess risk bound in (\ref{eq:tightened}) for ERM under the source only case becomes
\begin{align}
     &\mathbb{E}_{W}[R_{\mu'}(W_\ERM)] \leq \frac{1}{c\eta' n} \sum_{i=1}^{n} \left(I(W;Z_i) + D(\mu\|\mu')\right) \nonumber \\ 
     & + \frac{1}{c}\mathbb{E}_{P_{WS}}[\hat{R}(W,S)] \nonumber \\
     & = 4\sigma^2\log \frac{n}{n-1} + 4(m-m')^2 - \frac{2\sigma^2}{n} - 2(m-m')^2 \nonumber \\
     &= 4\sigma^2\log \frac{n}{n-1} - \frac{2\sigma^2}{n}  + 2(m-m')^2 \nonumber \\
     &\asymp \frac{2\sigma^2}{n} + 2(m-m')^2, \label{eq:fast_gaussian}
\end{align}
for large $n$. While the true excess risk can be calculated by
\begin{align*}
\mathbb{E}_{W}[R_{\mu'}(W_\ERM)] &= (m-m')^2 + \frac{\sigma^2}{m} + \sigma^2 - \sigma^2 \\
&= (m-m')^2 + \frac{\sigma^2}{n}.
\end{align*}
The new bound is tight in the sense that it captures the true excess risk up to a scaling factor. However, if we apply (\ref{eq:excess_erm_bound_mi}) and the bound becomes,
\begin{align*}
     \mathbb{E}_{W}[R_{\mu'}(W_\ERM)] \leq  & \sqrt{\sigma^2_{\ell'}\log\frac{n}{n-1} + 2\sigma^2_{\ell'}D(\mu \| \mu')} \\ 
     &+ d_{\mathcal{W}}(\mu,\mu'),
\end{align*}
where $\sigma^2_{\ell'} =  2\sigma^4_\ell + 4\sigma^2_\ell(m-m')^2$ and $d_{\mathcal{W}}(\mu,\mu') = \sup_{w\in\mathcal{W}}\left| (w-m)^2 - (w-m') \right|$. Then this bound approaches
\begin{align}
  \sqrt{4(m-m')^4 + 2\sigma^2(m-m')^2} + d_{\mathcal{W}}(\mu,\mu')
\end{align}
with the rate of $\sqrt{1/n}$, which is apparently worse than (\ref{eq:fast_gaussian}).
\end{example}

\begin{remark}[Justification of the tightness]
In the following, we examine the tightness of the bound and show why $r(w,z)$ is a more sensible choice than $\ell(w,z)$. For simplicity, we consider the source-only case with the Gaussian mean estimation problem where in the proof we used the Donsker-Varadhan representation for the KL divergence between the $D(P_{WZ_i}\|P_W \otimes \mu')$ for each $Z_i$:
\begin{align}
      &D(P_{WZ_i} \| P_W \otimes \mu')=\sup _{f: \mathcal{W} \otimes \mathcal{Z} \rightarrow \mathbb{R}} \mathbb{E}_{WZ_i}[f(W,Z_i)]\nonumber \\ 
      &-\log \left(\mathbb{E}_{P_W\otimes \mu'}\left[e^{f(W,Z_i)}\right]\right). \label{eq:donsker}
\end{align}
It is known that under mild conditions \cite{donsker1975asymptotic}, the optimal function where the equality is achieved for Eq.~(\ref{eq:donsker}) is chosen by $f'(dP_{WZ_i}/(d(P_W\otimes\mu')$ where $f(t) = t\log t$. We will calculate this optimizer explicitly and show that the choice of $r(w,z_i)$ is actually tight. To this end, we firstly calculate the densities of $P_W$ and $P_{W|Z_i}$ as: $dP_W = \frac{\sqrt{n}}{\sqrt{2\pi\sigma^2}} \exp(\frac{-(W - \mu)^2 n }{2\sigma^2})$, $dP_{W|Z_i} = \frac{n}{\sqrt{2\pi \sigma^2(n-1)}} \exp(-\frac{(W - \frac{n-1}{n}\mu - \frac{1}{n}Z_i)^2n^2}{2\sigma^2(n-1)}).$ Then we can calculate the optimizer as:
\begin{align*}
&f'(dP_{WZ_i}/d(P_W\otimes \mu')) = \log \frac{dP_{W|Z_i}}{dP_W} + \log \frac{d\mu}{d\mu'} + 1 \\
&= \underbrace{\frac{1}{2}\log\frac{n}{n-1} - \frac{(w-z_i)^2 - (\mu - z_i)^2}{2\sigma^2} - \frac{(w-z_i)^2}{2\sigma^2(n-1)}}_{\log \frac{dP_{W|Z_i}}{dP_W}} \\ 
&+ \underbrace{\frac{(\mu' - z_i)^2 - (\mu - z_i)^2}{2\sigma^2}}_{\log \frac{d\mu}{d\mu'}} + 1   \\
&= \frac{1}{2}\log\frac{n}{n-1} - \frac{(w-z_i)^2 
 - (\mu' - z_i)^2}{2\sigma^2}\\
&- \frac{(w-z_i)^2}{2\sigma^2(n-1)} + 1
\end{align*}
for fixed $w$ and $z_i$. The above function can be written as:
\begin{align*}
    f'(dP_{W|Z_i}/dP_W) = &-\frac{r(w,z_i)}{2\sigma^2} - \frac{\ell(w,z_i)}{2\sigma^2(n-1)} \\ 
    &+ \frac{1}{2}\log\frac{n}{n-1} + 1.
\end{align*}
The unexpected excess risk $r(w,z_i)$ clearly appears in the optimizer with some scaling factor and shifting constant (up to a $O(\frac{1}{n})$ difference), which, however, will not affect the convergence. To rigorously show this, we state the following result.
\begin{lemma}
The choice of the function $-\frac{r(w,z_i)}{2\sigma^2}$ satisfies the following inequality:
\begin{align*}
    & \frac{n-1}{n} \left(I(W;Z_i) + D(\mu\|\mu')\right) \\
    & \leq \mathbb{E}_{WZ_i}[-\frac{r(W,Z_i)}{2\sigma^2}]  - \log \mathbb{E}_{P_W\otimes \mu'}[e^{-\frac{r(W,Z_i)}{2\sigma^2}}] \\ 
    & \leq I(W;Z_i) + D(\mu\|\mu'),
\end{align*}
\end{lemma}
The provided lemma confirms that for the Gaussian example, the variational representation is actually tight. However, when using the loss function $\ell(w, z_i)$ without referencing $w^*$, the mutual information bound might not be tight as the equality in the variational representations may not be achieved (up to $O(\frac{1}{n} + D(\mu\|\mu'))$). Moreover, we posit that selecting $-r(w,z_i)$ also results in a tight upper bound on the generalization error. This is further supported by our demonstration of a corresponding lower bound for the generalization error in the following.
\begin{lemma}[Matching Lower Bound]\label{thm:lower_bounds}
Consider the Gaussian mean estimation problem with $\beta = 0$. With a large $n$, the following inequality holds for ERM:
\begin{align*}
    &\mathbb{E}_{WS}[\gen(W,S)] \geq \\ 
    & 2\sigma^2 \frac{n-1}{n^2}\sum_{i=1}^{n} I(W;Z_i) + D(\mu\|\mu').
\end{align*}
\end{lemma}
From our result, it could be seen that the sample-wise mutual information is present in both the upper and lower bounds. When considering the generalization error, the rates of convergence for both bounds align, albeit with varying leading constants. In the context of the Gaussian mean example, the excess risk upper bound is precise, given that both the empirical excess risk and generalization error are of order $O\left(\frac{1}{n} + D(\mu\|\mu')\right)$. Yet, for a broader range of learning problems, the lower bound for the excess risk primarily hinges on the data distributions and may not be calculated easily and explicitly.
\end{remark}

\noindent Moreover, the learning bound in Theorem~\ref{thm:central-transfer} can be applied to the regularized ERM algorithm as:
\begin{align*}
    w_{\sf{RERM}} = \argmin_{w \in \mathcal{W}} \hat{L}_{\alpha}(w,S,S') + \frac{\lambda}{n}g(w),
\end{align*}
where $g : \mathcal{W} \rightarrow \mathbb{R}$ denotes the regularizer function and $\lambda$ is some penalizing coefficient. We define $\hat{{R}}_{\textup{reg}}(w,S,S') = \hat{{R}}_{\alpha}(w,S,S') + \frac{\lambda}{n}(g(w) - g(w^*))$, then we have the following lemma.
\begin{lemma}\label{lemma:rerm}
We assume conditions in Theorem~\ref{thm:central-transfer} hold for the regularized ERM and also assume $|g(w_1) - g(w_2)| \leq B$ for any $w_1$ and $w_2$ in $\mathcal{W}$ with some $B >0$. Then for $W_{\sf{RERM}}$:
\begin{align*}
     &\mathbb{E}_{W}\left[L_{\mu'}(W_{\sf{RERM}}) - L_{\mu'}(w^*) \right]  \leq  \\
     & \frac{1}{c} \mathbb{E}_{P_{WSS'}}\left[\hat{{R}}_{\textup{reg}}\left(W_{\sf{RERM}}, S, S'\right)\right] + \frac{\lambda B}{cn}   \\
      &+ \frac{1}{c\eta'} \frac{\alpha}{\beta n} \sum_{i=1}^{\beta n} I(W_{\sf{RERM}}; Z'_i) \\
      & + \frac{1}{c\eta'} \frac{1-\alpha}{(1-\beta)n} \sum_{i=\beta n+ 1}^{n} \left(I(W_{\sf{RERM}};Z_i) + D(\mu\|\mu')\right).
\end{align*}
\end{lemma}
\noindent The proof of this result is given in Appendix \ref{proof:lemma_rerm}. As $\hat{{R}}_{\textup{reg}}(w,S,S')$ will be negative for $w_{\sf{RERM}}$, the regularized ERM algorithm can lead to the fast rate up to the domain divergence (by ignoring the multiplicative constant) if both $I(W_{\sf{RERM}};Z'_i)$ and $I(W_{\sf{RERM}};Z_i)$ are of $O(1/n)$.

From Theorem~\ref{thm:central-transfer}, we can achieve the fast rate if the mutual information between the hypothesis and data example is converging up to the domain divergence with the rate $O(1/n)$, and one may ask whether we could arrive at an intermediate rate between $O(1/\sqrt{n})$ and $O(1/n)$. To further relax the $(\eta,c)$-central condition, we can also derive the intermediate rate with the order of $O(n^{-\alpha})$ for $\alpha \in [\frac{1}{2}, 1]$. Similar to the $v$-central condition, which is a weaker condition of the $\eta$-central condition \cite{van2015fast,Grunwald2020}, we propose the $(v,c)$-central condition first and derive the intermediate rate results in Theorem~\ref{lemma:intermediate}.

\begin{definition}[$(v,c)$-Central Condition]\label{def:weaker-eta-c}
Let $v:[0, \infty) \rightarrow[0, \infty)$ is a bounded and non-decreasing function satisfying $v(\epsilon)>0$ for all $\epsilon > 0$. We say that $(\mu, \mu', \ell, \mathcal{W}, \mathcal{A})$  satisfies the $(v,c)$-central condition if for all $\epsilon \geq 0$, it holds that
\begin{align}
& \log \mathbb{E}_{P_W\otimes \mu'}  \left[e^{-{v(\epsilon)}\left(\ell(W,Z)-\ell(w^*,Z)\right)}\right]  \leq \nonumber \\ 
& -cv(\epsilon)  \mathbb{E}_{P_W\otimes \mu'}\left[\ell(W,Z) - \ell(w^*,Z)\right] + v(\epsilon) \epsilon. \label{eq:v-central}
\end{align}
\end{definition}


\begin{example}[Sub-Gaussian Condition]
We assume that the $\sigma^2$-subgaussian holds under the target distribution $P_W \otimes \mu'$ for a certain $\sigma^2$. Then the learning tuple also satisfies $(\min(v, c)$-central condition where $v(\epsilon) = \frac{2\epsilon}{\sigma^2}$ and $c = 1$.
\end{example}

\begin{example}[Bernstein Condition]
Let $\gamma \in (0,1]$ and $B \geq 1$. We assume that the \textbf{Bernstein condition} holds under the target distribution $P_W \otimes \mu'$ for given $\gamma$ and $B$. Additionally, if  $r(w,z_i)$ is bounded by $-b$ with some $b > 0$ for all $w$ and $z_i$, the learning tuple also satisfies $(\min(v, c)$-central condition where $v(\epsilon) = \frac{\epsilon^{1-\gamma}}{2Bc(1-\gamma)^{1-\gamma}}$ and $c = \min\{\frac{1}{2}, \gamma\}$. 
\end{example}
Then we can derive the following results for the intermediate rate result.
\begin{theorem}\label{lemma:intermediate}
Assume the learning tuple $(\mu, \mu', \ell, \mathcal{W}, \mathcal{A})$ satisfies the  $\left(v, c\right)$-central condition up to $\epsilon$ for some function $v$ as defined in~Def. \ref{def:weaker-eta-c} and $0 < c < 1$. Then it holds that for any $\epsilon \geq 0$ and any $0< \eta' \leq v(\epsilon)$,
\begin{align}
     &\mathbb{E}_{W}[R_{\mu'}(W_\ERM)] 
    \leq  \frac{1}{c} \mathbb{E}_{{WSS'}}[\hat{{R}}_{\alpha}\left(W, S,S' \right)] \nonumber  \\ 
     & + \frac{\alpha}{c \beta n} \sum_{i=1}^{\beta n} \left( \frac{I(W;Z'_i)}{\eta'} + \epsilon \right) \nonumber \\
      & + \frac{1-\alpha}{c(1-\beta)n} \sum_{i=\beta n+ 1}^{n} \left(\frac{I(W;Z_i) +  D(\mu\|\mu')}{\eta'} + \epsilon \right).
 \end{align}
\end{theorem}
\noindent In particular, if $v(\epsilon) = \epsilon^{1-\gamma}$ for some $\gamma \in (0,1]$, then the generalization error is bounded by,
\begin{align*}
     &\mathbb{E}_{W}[R_{\mu'}(W_\ERM)] \leq \\ 
     &  \frac{1}{c} \mathbb{E}_{{WSS'}}[\hat{{R}}_{\alpha}\left(W, S,S' \right)] + \frac{2\alpha}{c \beta n} \sum_{i=1}^{\beta n}  I(W;Z'_i)^{\frac{1}{2-\gamma}} \\
     & + \frac{2(1-\alpha)}{c(1-\beta)n} \sum_{i=\beta n+ 1}^{n} \left(I(W;Z_i) +  D(\mu\|\mu')\right)^{\frac{1}{2-\gamma}} .
 \end{align*}
The proof can be found in Appendix~\ref{apd:lemma_inter}. Thus, the expected generalization is found to have an order of $I(W;Z_i)^{\frac{1}{2-\gamma}}$, which corresponds to the typical results under Bernstein's condition \cite{hanneke2016refined,mhammedi2019pac,Grunwald2021pac}. 

%% file: 4.extension.tex
\section{Applications and Extensions}\label{sec:applications}


\subsection{Generalization error of stochastic noisy iterative algorithms}\label{sec:noisy}
The upper bound obtained in the previous section cannot be evaluated directly as it depends on the distribution of the data, which is, in general, assumed unknown in learning problems. Furthermore, in most cases, $W_{\ERM}$ does not have a closed-form solution but is obtained using an optimization algorithm. In this section, we study the class of optimization algorithms that iteratively update its optimization variable based on both source $S$ and target dataset $S'$. The upper bound derived in this section is useful in the sense that the bound can be easily calculated if the relative learning parameters are given. Specifically, the hypothesis $W$ is represented by the optimization variable of the optimization algorithm, and we use $W(t)$ to denote the variable at iteration $t$. In particular, we consider the following noisy iterative algorithm:
\begin{equation}
W(t) = W(t-1) - \eta_t\nabla\hat{L}_\alpha(W(t-1),S,S^\prime) + n(t),
\label{eq:iterative}
\end{equation}
where $W(t)$ is initialized to be $W(0) \in \mathcal{W}$ arbitrarily, $\nabla\hat{L}_{\alpha}$ denotes the gradient of $\hat{L}_{\alpha}$ with respect to $W$, and $n(t)$ can be any noises with the mean value of $0$ and variance of $\sigma^2_tI_d \in \mathbb{R}^d$. A typical example is $n(t) \sim \mathcal{N}(0,\sigma^2_tI_d)$. 

To obtain a generalization error bound for the above algorithm aligning with the theorem derived, we make the following assumptions.
\begin{assumption}  We assume the loss function $\ell(w,z)$ is $r^2\text{-subgaussian}$ under the distribution $\mu'$ for any $w\in \mathcal{W}$.
\label{asp:r-gaussian}
\end{assumption}

\begin{assumption}
The gradient is bounded, e.g., $\left\| \nabla\ell(w(t), z_i) \right\|_2 \leq K_S$, for all $z_i \in S$, and $\left\| \nabla \ell(w(t),z_i) \right\|_2  \leq K_T $, for all $z_i \in S'$ with $K_S,K_T > 0$, $\forall t\geq 1$. Then it follows that $\left \| \nabla(\hat{L}_\alpha(w(t),S,S^{\prime}))\right\|_2  \leq (1-\alpha)K_S + \alpha K_T \triangleq K_{ST}$. 
\label{asp:bounded}
\end{assumption}
\begin{remark}
    The bounded gradient assumption is a common assumption made in many analyses of machine learning and optimization algorithms, particularly in the context of gradient descent and its variants \cite{moulines2011non,malherbe2017global}.  In simpler terms, it ensures that the function does not have any abrupt or infinitely steep changes, and this assumption can be easily satisfied in many learning setups. For instance, functions like the absolute loss $f(x) = |x|$ and the linear loss $f(x) = mx + c$ are inherently Lipschitz continuous, thus adhering to the bounded gradient criteria. Similarly, polynomials of limited degrees, such as the quadratic function $f(x) = x^2$, satisfy this condition within a closed interval $x \in [a,b]$. Furthermore, the sigmoid and hyperbolic tangent functions, two popular activation functions used in neural networks, also comply with this assumption, having bounded constants of $K = 0.25$ and $K = 1$, respectively.
\end{remark}
Now we will apply the bound in Corollary~\ref{coro:general_bound} by further characterizing the mutual information $I(W;Z_i)$ with the relevant optimization parameters.
\begin{theorem}[Generalization error of stochastic noisy iterative algorithm]
Suppose that Assumptions~\ref{asp:r-gaussian} and \ref{asp:bounded} hold and $W(T)$ is obtained from (\ref{eq:iterative}) at $T$th iteration. Then the generalization error is upper bounded by 
\begin{align}
&\mathbb{E}_{W S S^{\prime}}  \left[\operatorname{gen}\left(W(T), S, S^{\prime}\right)\right]  \leq \alpha \sqrt{\frac{2 r^{2}}{\beta n} \hat{I}(S)} \nonumber \\ 
& + (1-\alpha)\sqrt{2 r^{2}\left(\frac{\hat{I}(S)}{(1-\beta)n}+D\left(\mu \| \mu^{\prime}\right)\right)},
\label{eq:T_iteration_gen}
\end{align}
where we define
\begin{equation}
\hat{I}(S) := \frac{d}{2} \sum_{t=1}^{T}\log \left(2 \pi e \frac{\eta_{t}^{2} K_{ST}^{2}+d \sigma_{t}^{2}}{d}\right)- \sum_{t=1}^{T}h(n_t).
\label{eq:SGDMItarget}
\end{equation}
\label{thm:Tbound}
\end{theorem}

The proof is given in Appendix~\ref{proof:thm_Tbound}.  In this bound, we observe that if the optimization parameters (such as $\alpha, \beta, n(t), w(0), T, d$) and loss function are fixed, the generalization error bound is easy to calculate by using the parameters given above. Also note that our assumptions do not require that the noise is Gaussian distributed or the loss function $\ell(w,z)$ is convex; this generality provides a possibility to tackle a wider range of optimization problems. 

\begin{remark}
For fixed learning parameters $T, \sigma^2_t, K_{ST}, \eta_t$ not depending on the sample size $n$, if we increase the sample size $n$, the bound decays up to $(1-\alpha)\sqrt{2r^2D(\mu\|\mu')}$. This shows that as the sample size increases, the mutual information diminishes as each individual instance has less influence on the gradient descent updates. Consequently, the decreased $\hat{I}(S)$ will result in a reduced generalization error. 
\end{remark}

\begin{remark}
If we fix the sample size instead but increase the iteration number $T$, the bounds would go to infinity with a fast-growing rate if $\eta_t$ and $\sigma^2_t$ are not wisely chosen. To determine proper choices of the optimization parameters to achieve a low generalization error, we consider the case $d = 1$ and $h(n_t) = d\log 2\pi e \sigma^2_t$ where $n_t \sim \mathcal{N}(0,\sigma^2_tI_d) $ for simplicity, and we can then further upper bound $\hat{I}(S)$ as:
\begin{align*}
    \hat{I}(S) = \frac{1}{2}\sum_{t=1}^{T}\log\left(1 + \frac{\eta^2_tK^2_{ST}}{\sigma^2_t} \right)\leq \frac{K^2_{ST}}{2}\sum_{t=1}^{T}\frac{\eta^2_t}{\sigma_t^2}
\end{align*}
with the inequality that $\log(1+x) \leq x$ for $x \geq 0$. To have better control of the mutual information, we may need to control the rate of the summation of $\frac{\eta^2_t}{\sigma^2_t}$. One typical choice of the noisy stochastic gradient descent is $\eta_t = \frac{1}{t}$ and the noise is set as $\sigma_t = \sqrt{\eta_t}$, then $\hat{I}(S) \leq O(\log(T))$, which scales logarithmically and coincides with the results in \cite{welling_bayesian_2011,pensia_generalization_2018}. Recently, in \cite{haghifam2023limitations}, it is shown that the mutual information bounds fail to give vanishing bounds w.r.t. $T$ for both generalization error and excess risk (see (2) and (3) for example). Regarding this, while the mutual information bounds may not provide optimal bounds, we still include these results to showcase the practicality of mutual information bounds in analyzing the noisy gradient descent by knowing the optimization parameters. However, it should be noted that more refined techniques with surrogate algorithms may be required for a tighter characterization. From the ratio, it could also be seen that the mutual information is controlled by the step size $\eta$ and the noise variance $\sigma^2_t$. Given a larger step size $\eta$, the data sample will have a greater influence on the hypothesis. On the contrary, if the noise variance is dominating, altering the step size might have less impact on the hypothesis.
\end{remark}

However, in many cases, the generalization error does not fully reflect the effectiveness of the hypothesis if $W(T) \neq W_{\ERM}$. One can further provide an excess risk upper bound by utilizing Proposition 3 in \cite{schmidt_convergence_2011} with the assumption of a strongly convex loss function, which guarantees the convergence of the hypothesis. We now provide an excess risk upper bound when the loss function is strongly convex. Recall that the excess risk of $W(T)$ is defined as
\begin{equation}
R_{\mu'}\left( W(T) \right)=L_{\mu^{\prime}}\left(W(T)\right)-L_{\mu^{\prime}}\left(w^{*}\right).
\end{equation}
Following the result in Theorem~\ref{thm:excess}, we present the upper bound for the excess risk if the following two assumptions hold.
\begin{assumption}
$\ell(w,z)$ is $\nu$-strongly convex, namely
\begin{equation}
\ell(w_1,z) \geq \ell(w_2,z) + \nabla\ell(w_2)(w_1-w_2) + \frac{\nu}{2}\|w_1 - w_2 \|^2
\end{equation}
for some $\nu >0$ and any $w_1, w_2 \in \mathcal{W}$. \label{asp:strong_convex}
\end{assumption}

\begin{remark}
    The Assumption~\ref{asp:strong_convex} is a fundamental condition in the context of optimization. For algorithms like gradient descent, strong convexity can guarantee faster convergence rates and can be introduced through regularization techniques like $L_2$ regularization, which can prevent overfitting in machine learning models \cite{boyd2004convex, rakhlin2011making}. For example, in the context of linear regression, the loss function is given by the mean squared error. When we add an $L_2$ regularization term (i.e., ridge regression), the loss function satisfies the strongly convex properties. Similarly, the cross-entropy loss in the logistic regression with an added $L_2$ regularization could also satisfy the strong convex condition\cite{salehi2019impact}. Another type of loss - the exponential loss used in the boosting algorithm - is also strongly convex, especially in the context of AdaBoost \cite{freund1997decision}. 
\end{remark}

\begin{assumption}
The loss function $\ell(w,z)$ has $\mathcal{L}$-Lipschitz-continuous gradient such that 
\begin{align}
  |\nabla \ell(w_1,z) - \nabla \ell(w_2,z)| \leq \mathcal{L}|w_1 - w_2|
\end{align}
for any $w_1, w_2 \in \mathcal{W}$ with respect to any $z \in \mathcal{Z}$. 
\label{asp:gradient_lipschitz}
\end{assumption}
\begin{remark}
    The Lipschitz-continuous gradient condition (sometimes referring to $\mathcal{L}$-smooth condition) ensures that the gradient does not change too abruptly and has been widely applied in many optimization problems, particularly in the context of stochastic gradient descent (SGD) and its variants \cite{recht2011hogwild,nguyen2017sarah}. This condition can be used to bound the updates, ensuring that the gradient change does not become too large and destabilize the learning process. Some typical examples include the mean squared loss, the Huber loss \cite{huber1992robust}, and the log-cosh loss \cite{saleh2022statistical}. However, it is important to highlight that not all loss functions in machine learning exhibit smoothness. For instance, the least absolute loss and hinge loss are examples of non-smooth loss functions.
\end{remark}

\begin{corollary}[Excess risk of strongly convex loss function]
Suppose Theorem \ref{thm:Tbound} holds and the loss function~$\ell(w,z)$ satisfies Assumptions~\ref{asp:strong_convex} and \ref{asp:gradient_lipschitz}. Define $\kappa = \frac{\nu}{\mathcal{L}}$, setting $\eta = \frac{1}{\mathcal{L}}$, and $W$ is arbitrarily initialized with $W(0)$. Then the excess risk can be bounded as follows:
\begin{align} 
& \mathbb{E}_{W}\left[R_{\mu'}\left( W(T) \right)\right]  \leq   (1-\alpha)d_{\mathcal W}(\mu, \mu') + \alpha \sqrt{\frac{2 r^{2}}{\beta n} \hat{I}(S)} \nonumber \\ 
& + (1-\alpha)\sqrt{2 r^{2}\left(\frac{\hat{I}(S)}{(1-\beta)n}+D\left(\mu \| \mu^{\prime}\right)\right)} \nonumber \\ 
& + K_{ST}(1-\kappa)^T\mathbb{E}[\left\| W(0) - W_\ERM \right \|] \nonumber  \\ 
& + K_{ST}\sum_{t=1}^T(1-\kappa)^{T-t}\mathbb{E}\left[\left\| n(t) \right\|\right].
\end{align}
\label{coro:convex_er_bound}
\end{corollary}
\begin{remark}
    From the above bound, we could also optimize $\kappa$ and $\eta$ for a tighter bound. Let us consider the same setup that $n(t)$ is Gaussian distributed and $d = 1$. First, we set $\eta_t = \frac{1}{t}$ and $\sigma_t =\sqrt{\eta_t}$ to make the generalization error tight. Now we assume that the hypothesis space is bounded, i.e., $\left\| w_1 - w_2 \right \| \leq W_{B}$ for any $w_1, w_2 \in \mathcal{W}$. Then the third term in R.H.S. will decay exponentially when $T$ goes to infinity, which means the initialization does not affect the generalization error with a large number of iterations. While the fourth term in the R.H.S. can be upper bounded by:
    \begin{align*}
        & K_{ST}\sum_{t=1}^T(1-\kappa)^{T-t}\mathbb{E}\left[\left\| n(t) \right\|\right] = \frac{K_{ST}}{\sqrt{2\pi}}\sum_{t=1}^{T}\frac{(1-\kappa)^{T-t}}{\sqrt{t}} \\ 
        & \leq \frac{K_{ST}}{\sqrt{2\pi}} \frac{1}{1-(1-\kappa)}.
    \end{align*}
    Then the final bound would have the form of 
    \begin{align*}
    &\mathbb{E}_{W}\left[R_{\mu'}(W(T))\right] \leq (1-\alpha)d_{\mathcal W}(\mu, \mu') \\
    & + O\Bigg(\alpha \sqrt{\frac{\log T}{\beta n}} + (1-\alpha) \sqrt{\frac{\log T}{(1-\beta)n} + cD(\mu\|\mu')} \\ 
    & + (1-\kappa)^T + \frac{K_{ST}}{\sqrt{2\pi \kappa}} \Bigg),
    \end{align*}
    where $c$ is some leading constant. From the bound, we can see that selecting $\kappa$ as a constant would not affect the rate of the bound, but the choice of $\eta$ could be crucial as it controls the variance level of the noise.    
\end{remark}
The proof is provided in Appendix \ref{proof:convex_er_bound}. Notice that the summation of the terms $\| n(t)\|$ needs to be finite. Hence this upper bound is effective when $n(t)$ is sampled from a bounded random variable (for example, truncated Gaussian or uniform random variable), but not for the case $n(t)$ is Gaussian distributed, where $\| n(t) \|$ is not bounded. We give a toy example in Section~\ref{sec:bernoulli_sgd} for the Bernoulli transfer to show the effectiveness of the bounds.

\subsection{Generalization error on Gibbs algorithm}\label{sec:gibbs}
Theorem \ref{thm:excess} shows that the generalization error can be upper-bounded in terms of the mutual information between the input data and output hypothesis and the KL divergence between the source and target domains. Since the KL divergence $D(\mu\|\mu')$ is usually uncontrollable as $\mu$ and $\mu'$ are unknown in real settings, it is natural to consider an algorithm that minimizes the empirical risk regularized by $I(W,Z_i)$ as
\begin{align}
    &P^*_{W|S,S'} = \nonumber  \\
    &\argmin_{P_{W|S,S'}}\left(\mathbb{E}_{WSS'}[\hat{L}_{\alpha}(W,S,S')] + \frac{1}{k}\sum_{i=1}^{n}I(W;Z_i) \right).
\end{align}  
With the chain rule, we have
\begin{align}
&\sum_{i=1}^{\beta n}I(W;Z'_i) + \sum^n_{i=\beta n+1}I(W;Z_i)  \leq \sum_{i=1}^{\beta n} I(W;Z'_i|(Z')^{i-1}) \nonumber  \\ 
& + \sum_{i=\beta n + 1}^{n} I(W;Z_i|S', Z_{\beta n}^{i-1})= I(W;S,S'),
\end{align}
with the definition $(Z')^{i-1} = \{Z'_1,Z'_2,\cdots,Z'_{i-1}\}$ and $Z_{\beta n}^{i-1} = \{Z_{\beta n+1}, \cdots, Z_{i-1} \}$. Then we aim to minimize the relaxed Gibbs algorithm as
\begin{equation}
    P^*_{W|S,S'} = \argmin_{P_{W|S,S'}}\left(\mathbb{E}_{WSS'}[\hat{L}_{\alpha}(W,S,S')] + \frac{1}{k}I(W;S,S') \right). \label{eq:gibbs-1}
\end{equation}
We can also relax the above optimization problem by replacing $I(W;S,S')$ with an upper bound 
\begin{align*}
  D\left(P_{W | S,S'} \| Q | P_{S,S'}\right)=I(W;S,S')+D\left(P_{W} \| Q\right)
\end{align*}
where $Q$ is an arbitrary distribution on $W$. We can also rewrite $Q$ as 
\begin{align*}
  &D\left(P_{W | S,S'} \| Q | P_{S,S'}\right)= \\ 
  & \int_{\mathcal{Z}^{n}} D\left(P_{W | S=s, S'=s'} \| Q\right) d^{\otimes \beta n}(\mu') d^{\otimes (1 - \beta) n}(\mu),  
\end{align*}
which does not depend on source distribution $\mu$ or target distribution $\mu'$. Thus we can relax~(\ref{eq:gibbs-1}) and arrive at the following surrogate solution as
\begin{align}
&P^*_{W|S,S'} = \argmin_{P_{W|S,S'}}\Big(\mathbb{E}_{WSS'}[\hat{L}_{\alpha}(W,S,S')] \nonumber \\ 
& + \frac{1}{k}D\left(P_{W | S,S'} \| Q | P_{S,S'}\right) \Big).
    \label{eq:Gibbs}    
\end{align}
\begin{theorem}
The solution to the optimization problem (\ref{eq:Gibbs}) is the Gibbs algorithm, which satisfies
\begin{equation}
P_{W | S', S}^{*}(\mathrm{d} w)=\frac{e^{-k \hat{L}_{\alpha}(W,S,S')} Q(\mathrm{d} w)}{\mathbb{E}_{Q}\left[e^{-k \hat{L}_{\alpha}(W,S,S')}\right]}
\label{eq:Gibbs-Solution}
\end{equation}
for each $(S,S') \in \mathcal{Z}^{n}$.
\label{thm:gibbs-solution}
\end{theorem} 

The proof can be found in Appendix~\ref{apd:gibbs}. For fixed $\alpha$, we denote the hypothesis that achieves the combined minimum population risk among $\mathcal{W}$ by $w^*_{st}(\alpha)$ such that
\begin{align}
    w^*_{st}(\alpha) = \argmin_{w \in \mathcal{W}} \alpha L_{\mu'}(w) + (1-\alpha) L_{\mu}(w).
\end{align} 
In particular, we define $w^*_s = w^*_{st}(0)$ and we also have that $w^* = w^*_{st}(1)$. We further denote $w_G$ as the output hypothesis of the Gibbs algorithm. We have
\begin{align}
& \mathbb{E}_{W} [L_{\mu'}(W_G)] = \mathbb{E}_{WSS'}[L_{\mu'}(W_G) - \hat{L}_{\alpha}(W_G,S,S')] \nonumber  \\ 
& + \mathbb{E}_{WSS'}[\hat{L}_{\alpha}(W_G,S,S')] \nonumber\\
&\leq \mathbb{E}_{WSS'}[\gen(W_G, S, S')] +  \mathbb{E}_{WSS'}[\hat{L}_{\alpha}(W_G,S,S')] \nonumber \\ 
& + \frac{1}{k} D(P^*_{W_G|S,S'} \| Q|S,S').
\label{eq:excess_gibbs}
\end{align}
As a direct application of Corollary~\ref{coro:general_bound}, we then specialize the upper bound on the population risk for the Gibbs algorithm by further upper bounding the generalization error on the R.H.S. of (\ref{eq:excess_gibbs}). For any $\ell \in [0,1]$, we reach the following corollary for countable hypothesis space.
\begin{corollary}
\label{coro:gibbs-finite}
Suppose $\mathcal{W}$ is countable. Let $W_G$ denote the output of the Gibbs algorithm applied on dataset $S, S'$. For some $\alpha \in [0,1]$ and $\ell(w,z) \in [0,1]$ for any $w$ and $z$, the generalization error is upper bounded by:
\begin{align}
& \left|\Esub{WSS'}{\gen(W_G, S, S')}\right|
\leq \frac{\alpha^2k}{4\beta n}+ \frac{(1-\alpha)^2k}{4(1-\beta)n} \nonumber \\ 
& +  (1-\alpha)\sqrt{\frac{D(\mu \| \mu')}{2}}, \label{eq:Gibbs-gen}
\end{align}
and the population risk of the Gibbs algorithm satisfies:
\begin{align}
& \mathbb{E}_{W}\left[L_{\mu'}(W_G)\right] \leq L_{\alpha}(w^*_{st}(\alpha)) + \frac{1}{k} \log \frac{1}{Q\left(w^*_{st}(\alpha)\right)} \nonumber  \\ 
& +\frac{\alpha^2k}{4\beta n}+ \frac{(1-\alpha)^2k}{4(1-\beta)n} +  (1-\alpha)\sqrt{\frac{D(\mu \| \mu')}{2}}.
\end{align}
\end{corollary}

\begin{remark}
From~(\ref{eq:Gibbs-gen}), we can see that the generalization error bound converges up to the domain divergence with $O(\frac{1}{n})$, which coincides with the rate in \cite{bu2022characterizing} for the exact generalization error characterization of the $\alpha-$weighted ERM algorithm. The only difference lies in the definition of the generalization error: we define the generalization error as the gap between the combination of the empirical risks in both the source and target domains and the population risk in the target domain. While in \cite{bu2022characterizing}, the generalization error is defined as the gap between the empirical risk in the target domain only and its population risk. Thus, there is an additional domain divergence term in our upper bound. Apart from that, in both cases, the Gibbs algorithm is able to achieve a fast convergence rate. 
\end{remark}
If we set $\alpha = 1$, $\beta = 1$ (only use the target data), we will retrieve the results for conventional machine learning, derived in \cite[Corollary 2]{xu_information-theoretic_2017} as follows
\begin{align}
  \mathbb{E}_W\left[L_{\mu'}(W_G)\right] \leq L_{\mu'}(w^*)+ \frac{1}{k} \log \frac{1}{Q\left(w^*\right)}+ \frac{k}{4n}.
\end{align}
We consider the special case $\beta = 0$, $\alpha = 0$ (only using the source data),  then we have
\begin{align}
\mathbb{E}_W\left[L_{\mu'}(W_G)\right] \leq &  L_{\mu}(w^*_s) + \frac{1}{k} \log \frac{1}{Q\left(w^*_{s}\right)}  \nonumber \\ 
& + \frac{k}{4n} + \sqrt{ \frac{D(\mu \| \mu')}{2}}.
\end{align}
In this case, it is observed that $D(\mu \| \mu')$ would not vanish. Even when $n$ increases, the excess risk also depends on the population risk of $w^*_{s}$ w.r.t. the source distribution $\mu$.

When $\mathcal{W}$ is uncountable (e.g., $\mathcal{W} = \mathbb{R}^{d}$), the term $\frac{1}{k} D(P_{W_G|S,S'} \| Q|S,S')$ will be bounded using the approximation by a Gaussian distribution for $W_G$ given $S$ and $S'$. As a result, we could bound the population risk using the following corollary.
\begin{corollary} \label{coro:gibbs-infinite}
Suppose $\mathcal{W}=\mathbb{R}^{d}$ and $\ell(\cdot, z)$ is $\rho$-Lipschitz for all $z \in \mathcal{Z}$. Let $W_G$ denote the output of the Gibbs algorithm induced by datasets $S$ and $S'$. For $\ell \in[0,1]$ and some $\alpha \in [0,1]$, the population risk of $W_G$ satisfies
\begin{align}
& \mathbb{E}_W\left[L_{\mu'}(W_G)\right] \leq  L_{\alpha}(w^*_{st}(\alpha)) + \frac{\alpha^2k}{4\beta n}+ \frac{(1-\alpha)^2k}{4(1-\beta)n} \nonumber  \\ 
& +  (1-\alpha)\sqrt{\frac{D(\mu \| \mu')}{2}} \nonumber  \\
&+ \inf _{a>0}\left(a \rho \sqrt{d}+\frac{1}{k} D\left(\mathcal{N}\left(w^*_{st}(\alpha), a^{2} \mathbf{I}_{d}\right) \| Q\right)\right).
\end{align}
\end{corollary} 

%% file: 5.divergence.tex
\section{Bounding with other divergences} \label{sec:other_metric}
The usage of KL divergence to quantify distribution shifts has been popular in many previous works \cite{gupta2021s, nguyen2021kl, aminian2022information}. However, an important observation made by \citet{hanneke2019value} demonstrates certain limitations of this approach - specifically in the context of transfer learning - where in Theorem \ref{thm:excess}, the result is not effective for a class of supervised machine learning problems if $\mu$ is not absolutely continuous with respect to $\mu'$ as the KL divergence goes to infinity and the bound becomes vacuous. In Example 1 of \cite{hanneke2019value}, the authors also illustrate why KL divergence in the hypothesis space is not the right measure for this purpose, and the issue is wisely bypassed by changing the hypothesis distribution from $P(w)$ to a Bernoulli distribution $P(w = w^*)$. The transfer component - the divergence proposed in their work - can easily handle this case, providing insights into the data values in transfer learning. Now we give two examples.

\begin{example}
Let the sample $Z$ be a pair $(X, Y)$ where $X$ denotes features and $Y$ denotes the corresponding label, and we assume that $Y$ is determined by $X$, i.e., $Y=f(X)$ for some deterministic function $f$. In this case, the distribution $\mu$ of $Z$ can be factored as $\mu (z)=\mu(x,y)=P_{X}(x)P_{Y|X}(y|x)=P_X(x) \mathbf 1_{y=f(x)}$ where $P_X$ is the distribution of $X$ and $\mathbf 1$ denotes the indicator function. Let target distribution $\mu'$ factor as $\mu'(x,y)=P_{X}'(x)\mathbf 1_{y=f'(x)}$ for some distribution $P_X'$ and function $f'$. Notice that in this case, $\mu$ is not absolutely continuous with respect to $\mu'$  unless $f=f'$. Indeed, for some $(x, y)$ we have $\mathbf 1_{y=f'(x)}=0$ while $\mathbf 1_{y=f(x)}=1$ unless $f$ agrees with $f'$ almost everywhere. Therefore, the KL divergence is $D(\mu\|\mu')=\infty$, and the upper bound becomes vacuous unless $f=f'$ (however $P_X$ and $P_X'$ could still be different, hence the problem is not necessarily trivial in this case).
\end{example}

\begin{example}\label{ex:domain_div}
Let the sample $Z$ be a discrete random variable over the set $\{1,2,3\}$. Assume that for the source domain $\mu(Z = 1) = \mu(Z = 2) = \mu(Z = 3) = \frac{1}{3}$ while $\mu'(Z = 1) = \mu'(Z = 2) = \frac{1}{2}$ and $\mu'(Z = 3) = 0$. In this case, the KL divergence $D(\mu\|\mu') = \infty$ as $\mu$ is not absolutely continuous w.r.t. $\mu'$ when $Z = 3$. However, if we consider the total variation \cite{dudley2010distances} between $\mu$ and $\mu'$ (where the detailed definition is given in the later context), it can be calculated that $TV(\mu, \mu') = \sum_{z\in\mathcal{Z}} \frac{1}{2}|\mu(z) - \mu'(z)| = \frac{1}{3}$, which is finite instead.
\end{example}

We mitigate these issues tied to KL divergence by introducing bounds using other divergences, such as the Wasserstein distance and $\phi$-divergence, as we show in Example~\ref{ex:domain_div}. These alternative measures provide us with a more general and robust framework to quantify the distribution shift, which may also give a tighter characterization. 

\subsection{$\phi$-divergence bounds}
To develop an appropriate upper bound to handle the case where the KL divergence is infinite, we may extend the results by using other types of divergence following the work of \cite{jiao_dependence_2017}, which do not impose the absolute continuity restriction. To this end, we first introduce a more general divergence between two distributions that can handle such a case, namely, the $\phi$-divergence.
\begin{definition}[$\phi$-divergence]
Given two measures $\mu$, $\nu$ and a convex functional $\phi$, we define the $\phi$ divergence by:
\begin{align}
    D_{\phi}(\nu\|\mu) = \mathbb{E}_{\mu}\left[\phi(\frac{d\nu}{d\mu}) \right]
\end{align}
where $d\nu/d\mu$ is the Radon-Nikodym derivative. 
\end{definition}
To tackle the absolute continuity issue between $\mu$ and $\mu'$, we shall choose $\phi(x) = \frac{1}{2}|x-1|$ and arrive at the bounds with the total variation distance, which is always bounded by $[0,1]$. Following \cite{jiao_dependence_2017}, we suppose that the loss function $\ell (w,z)$ is $L_{\infty}$-norm upper bounded by $\sigma$ where the $L_{\infty}$-norm of a random variable is defined as
\begin{equation*}
    \| X \|_{\infty} = \inf \{ M: P(X > M) = 0 \}.
\end{equation*}
Then we have the following corollary.
\begin{corollary} (Generalization error bound of ERM using $\phi$-divergence)
Assume that for any $w\in\mathcal W$, the loss function $\ell(w, Z)$ is $L_{\infty}$-norm bounded by $\sigma$ under the distribution $\mu'$. Then the following inequality holds. 
\begin{small}
\begin{align}
&\mathbb{E}_{WSS'}\left[\gen (W_{\ERM}, S,S')\right] \leq \frac{2\alpha\|\sigma\|_{\infty}}{\beta n}\sum_{i=1}^{\beta n}I_{\phi}(W_{\ERM};Z'_i) \nonumber \\ 
&+\frac{2(1-\alpha)\|\sigma\|_{\infty}}{(1-\beta)n}\sum_{i=\beta n+1}^{n} \left( I_{\phi}(W_{\ERM};Z_i) + TV(\mu, \mu')\right),
\label{eq:excess-fdiv-ERM}
\end{align}
\end{small}
\label{coro:excess-fdiv-ERM}
where $I_{\phi}(W_{\ERM}; Z_i) = D_{\phi}(P_{W_{\ERM}Z_i}||P_{W_{\ERM}} \otimes P_{Z_i})$ is the $\phi$-divergence between the distribution $P_{W_{\ERM}Z_i}$ and $P_{W_{\ERM}} \otimes P_{Z_i}$ (similar to $Z'_i$) with $D_{\phi}(P||Q) = \frac{1}{2}\int|dP - dQ|$ and $TV(\mu, \mu') = D_{\phi}(\mu || \mu')$ denotes the total variation distance between the distribution $\mu$ and $\mu'$.
\end{corollary}
The proof can be found in Appendix~\ref{proof:coro-fdiv-ERM}.
Note that with this bound, the divergence $TV(\mu \| \mu')$ is always bounded by $[0,1]$ for any $\mu$ and $\mu'$. More generally, the generalization error can be upper bounded using different $\phi$-divergences, and the key to unifying different divergences is the Legendre-Fenchel duality that we used in Theorem~\ref{thm:exp_gen}. For example, by choosing $\phi(x) = x\log x$ and the function $\psi$ for bounding the moment generating function, we will end up with the KL-divergence-based bound. Likewise, if we choose $\phi(x) = \frac{x^2}{2}$ and the function $\psi$ for bounding the variance, we will arrive at the $\chi^2$-divergence based bound as shown in \cite{esposito2022generalisation}, such an extension allows the result to hold for a wider family of distributions (e.g., the sub-exponential random variables) where the mutual information bound may be invalid.  

\subsection{Generalization error with Wasserstein distance}
Another interesting metric, called the Wasserstein distance, has a very close connection with the KL divergence via the transportation-cost inequality \cite{raginsky2013concentration}. Such a distance has several advantages over the KL divergence. As we show later, the Wasserstein distance can handle the deterministic algorithm where the mutual information is infinite. Furthermore, under mild conditions, the Wasserstein distance-based bound is naturally tighter for transfer learning, compared to the mutual information bound as shown in Corollary~\ref{coro:general_bound}. To show our results, we first give some definitions of the probability measures. Let $(\mathcal{Z}, d)$ be a metric space and $p \in[1,+\infty)$, we define $\mathcal{P}_{p}(\mathcal{Z})$ as the set of probability measures $\mu$ on $\mathcal{Z}$ satisfying $(\mathbb{E}_{Z \sim \mu}[d(Z,z_0)^p])^{\frac{1}{p}} < \infty$ for some $z_0 \in \mathcal{Z}$. Then we define the $p$-Wasserstein distance as follows.
\begin{definition}[Wasserstein distance] 
Assume $\mu, \nu \in \mathcal{P}_{p}(\mathcal{Z})$. The $p$-Wasserstein distance between $\mu$ and $\nu$ is defined as 
\begin{align}
 \mathbb{W}_{p}(\mu, \nu)=\inf_{\pi \in \Pi(\mu, \nu)}\left(\mathbb{E}_{\pi}\left[d(Z, Z')^{p}\right]\right)^{1 / p}.   
\end{align}
where $\Pi(\mu, \nu)$ denotes the coupling of $\mu$ and $\nu$, i.e., the set of all the joint probability measures $\pi \in \mathcal{P}(\mathcal{Z} \times \mathcal{Z})$ with marginals equal to $\mu$ and $\nu$.
\end{definition}
With the definition in place, we give the generalization error bound of the ERM algorithm using the Wasserstein distance in the following theorem.
\begin{theorem} [Generalization error bound of ERM with Wasserstein distance] \label{thm:gen-wd} 
\noindent Let $P_W$ be the marginal distribution induced by $S,S'$ and $P_{W|SS'}$ with the ERM algorithm. Assume that for any $w\in\mathcal W$, the loss function $\ell(w, Z)$ is $\mathcal{L}$-Lipschitz for any $W \in \mathcal W$, $Z \in \mathcal Z$. Then the following inequality holds.
\begin{align}
&\mathbb{E}_{W}\left[\gen(W_{\ERM}, S, S')\right]  \leq \frac{\alpha\mathcal{L}}{\beta n} \sum_{i=1}^{\beta n} \mathbb{E}_{\mu'}[\mathbb{W}_1(P_{W},P_{W|Z'_i})]  \nonumber \\ 
& + \frac{(1-\alpha)\mathcal{L}}{(1-\beta) n} \sum_{i=\beta n + 1}^{n}\left( \mathbb{E}_{\mu}[\mathbb{W}_1(P_{W},P_{W|Z_i})] + \mathbb{W}_{1}(\mu,\mu') \right).
\end{align}
\end{theorem}

\begin{remark} Under the Wasserstein distance, the domain divergence is captured by the first order Wasserstein distance $\mathbb{W}_1(\mu, \mu')$, which also resolves the absolutely continuous issue in mutual information bound. This bound requires that the loss function is $\mathcal{L}$-Lipschitz with respect to any hypothesis $W$ and data instance $Z$ while the mutual information bound is derived under the subgaussian assumption. It is also worth noting that the bound is based on the term $\mathbb{W}_1(P_W,P_{W|Z_i})$. Intuitively, this term measures how distribution diverges with a given single instance $Z_i$. In other words, it measures how one instance can affect the distribution of $W$ given a specific algorithm. This intuition is also related to stability in the algorithmic perspective, similar to $I(W,Z_i)$ in the mutual information bound.
\end{remark}
The proof can be found in Appendix~\ref{proof:gen-wd}. Based on the result of the generalization error, we can derive the excess risk upper bound using the Wasserstein distance.

\begin{theorem} [Excess risk bound of ERM with Wasserstein distance] 
\label{thm:er-wd}
Assume the conditions in Theorem~\ref{thm:gen-wd} hold and assume the loss function $\ell(w,z)$ is bounded by $[0,1]$ for any $w$ and $z$. Then the following inequality holds:
\begin{align}
& \mathbb{E}_{W}\left[R_{\mu'}(W_{\ERM})\right] \leq    (1-\alpha) d_{\mathcal W}(\mu,\mu') \nonumber \\ 
&+ \frac{(1-\alpha)\mathcal{L}}{(1-\beta) n} \sum_{i=\beta n + 1}^{n}\left( \mathbb{E}_{\mu}[\mathbb{W}_1(P_{W},P_{W|z_i})] + \mathbb{W}_{1}(\mu,\mu') \right) \nonumber \\
& + \frac{\alpha\mathcal{L}}{\beta n} \sum_{i=1}^{\beta n} \mathbb{E}_{\mu'}[\mathbb{W}_1(P_{W},P_{W|z'_i})]. 
\end{align}  
\label{thm:ER-WD}
\end{theorem}
We can show that this bound is tighter than the mutual information bound under mild conditions. Specifically, we consider the case that only the source domain is available ($\alpha = \beta = 0$) with some $n$. From Theorem \ref{thm:excess}, we can easily bound the \textbf{expected excess risk} for ERM algorithm and we denote the bound by $\mathbb{B}_{\mathrm{Info}}$, which is defined as:
\begin{align}
\mathbb{B}_{\mathrm{Info}}  := & \frac{\sqrt{2r^2}}{n}\sum_{i=1}^{n}\sqrt{ (I(W_{\ERM};Z_i)+D(\mu||\mu'))} \nonumber \\ 
&+ d_{\mathcal W}(\mu,\mu').
\end{align}
From Theorem~\ref{thm:ER-WD}, we can also derive the bound for the expected excess risk with the ERM algorithm, and we denote the Wasserstein type bound by $\mathbb{B}_{\mathrm{Wass}}$ as
\begin{align}
\mathbb{B}_{\mathrm{Wass}} := &  \frac{\mathcal{L}}{n} \sum_{i=  1}^{n}\left( \mathbb{E}_{{\mu}}[\mathbb{W}_1(P_{W},P_{W|z_i})] + \mathbb{W}_{1}(\mu,\mu') \right) \nonumber \\ 
& + d_{\mathcal W}(\mu,\mu').
\end{align}
Next, we will prove that the $\mathbb{B}_{\mathrm{Wass}}$ is, in general, tighter than $\mathbb{B}_{\mathrm{Info}}$ under the mild assumption. To this end, we first introduce the transportation cost inequality (TCI) following the definition 3.4.2 from \cite{raginsky2013concentration}.
\begin{definition}[Transportation Cost Inequality]
We say that a probability measure $\mu$ on $(\mathcal{X}, d)$ satisfies an $L^{p}$ transportation cost inequality with constant $c>0,$ or a $\mathrm{T}_{p}(c)$ inequality for short, if for every probability measure $\nu \ll \mu$ we have
\begin{align}
    \mathbb{W}_{p}(\nu, \mu) \leq \sqrt{2 c D(\nu \| \mu)}. \label{eq:TCI}
\end{align}
\end{definition}
This is a typical definition in many learning setups and we will give several examples below.
\begin{example}
Under the Hamming distance, it can be proved that the first-order Wasserstein distance is equivalent to the total variation \cite{raginsky2013concentration}, then a well-known application of this inequality is the Pinsker's inequality where $p = 1$ and $c = \frac{1}{4}$. This inequality provides an upper bound on the Wasserstein distance in terms of the divergence.
\end{example}

\begin{example}\label{example:TCI}
It is proved that the inequality~(\ref{eq:TCI}) holds if and only if for some $c$ and every 1-Lipschitz function $f$, it satisfies the subgaussian property under $\mu$, e.g., for any $t \in \mathbb{R}$,
\begin{align*}
    \mathbb{E}_{\mu}[e^{tf}] \leq e^{\frac{ct^2}{2}}.
\end{align*}
Readers can refer to Theorem 1 in \cite{bobkov1999exponential} and Theorem $\diamond$ in \cite{esposito2022generalisation} for more details.
\end{example} 

From example~\ref{example:TCI}, it could be shown that if we set $\nu = P_{WZ_i}$ and $\mu = P_W\otimes P_{Z_i}$ in (\ref{eq:TCI}), the Wasserstein distance between $\nu$ and $\mu$ will be tighter than the mutual information measure. We provide a rigorous argument in the following proposition. 
\begin{proposition}
\label{prop:wd-info-comparison}
Consider the case where $\beta = 0$ for simplicity (e.g., with the source data only), let $P_W$ be the marginal distribution induced by some algorithm and the source sample distribution $\mu^{\otimes n}$. We assume that the induced conditional distribution $P_{W|z_i}$ is absolutely continuous w.r.t. $P_W$ for any $z_i \in \mathcal{Z}$, and both $P_{W}$ and $\mu'$ satisfies the $T_{1}(\frac{r^2}{2\mathcal{L}^2})$ transportation cost inequality. Then the following inequality holds:
\begin{equation*}
    \mathbb{B}_{\mathrm{Wass}}  \leq \mathbb{B}_{\mathrm{Info}}.
\end{equation*}
\end{proposition}
The proof can be found in Appendix~\ref{proof:wd-info-comparison}. Essentially, if both $\mu'$ and the resulting hypothesis $P_W$ exhibit subgaussian characteristics, which include having light tails and preserving strong concentration for any Lipschitz functions, then the Wasserstein distance could be a better metric than the mutual information when evaluating the generalization error. 

In this section, we demonstrate how various distribution metrics, including the total variation, $\phi$-divergence, and the Wasserstein distance, can be employed to derive different information-theoretic bounds. Metrics like the Wasserstein distance are particularly noteworthy as they not only address the absolute continuity concerns but have also proven to provide tighter bounds than those based on mutual information. However, a caveat is that evaluating these bounds requires the knowledge of data and hypothesis distributions, and estimating these distributions can be very challenging in practical scenarios. As an initial approach to this concern, the subsequent section introduces a heuristic algorithm that capitalizes on the properties and insights drawn from the aforementioned bounds.

%% file: 6.example.tex
\section{Examples and Algorithms}\label{sec:example}
In this section, we provide three examples to illustrate the upper bounds we obtained in previous sections. First we provide an example of calculating the learning bounds on the Bernoulli transfer problem given the optimization parameters with the stochastic gradient descent algorithm. Then we present the logistic regression transfer learning problem with the mutual information bounds, and lastly we evaluate the proposed InfoBoost algorithm in several real-world transfer learning scenarios.

\begin{figure*}[!htp]
	\centering
	\subfloat{
		\includegraphics[width=0.5\textwidth]{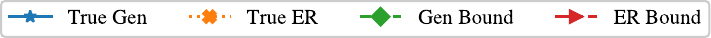}		}	
	\addtocounter{subfigure}{-1}
	\\
	\subfloat[~100 samples]{
		\includegraphics[width=0.2\textwidth]{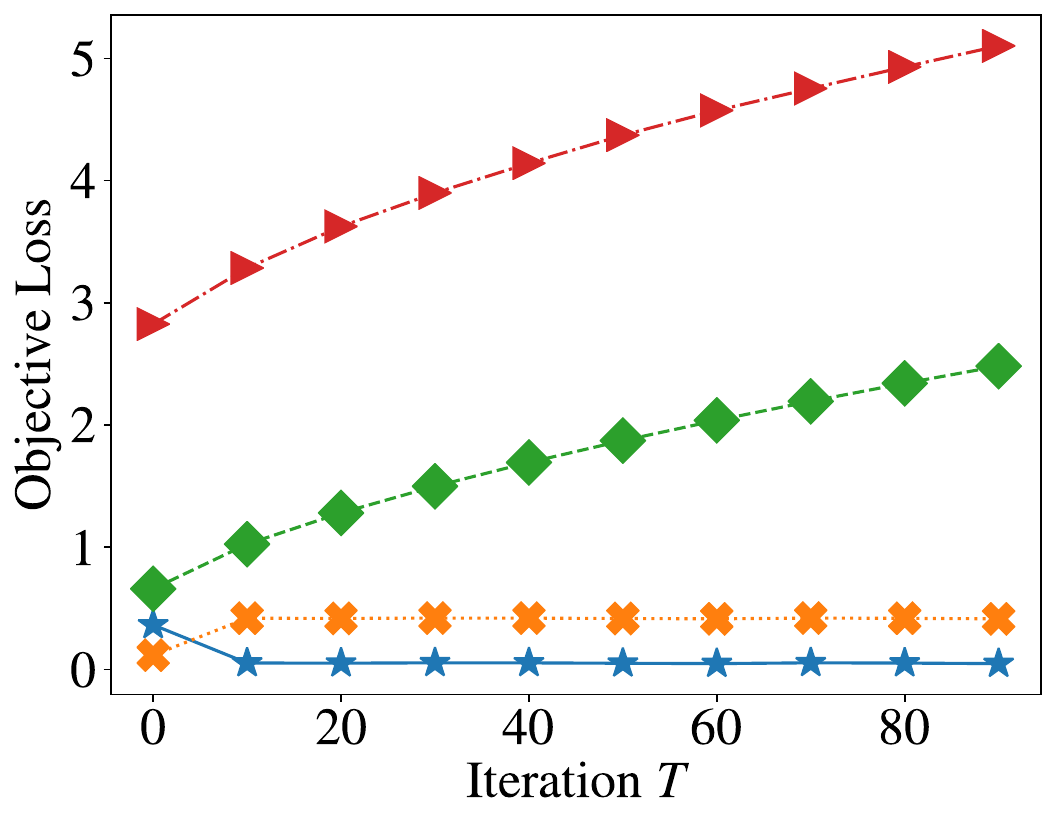}}
	\subfloat[~1000 samples]{
		\includegraphics[width=0.2\textwidth]{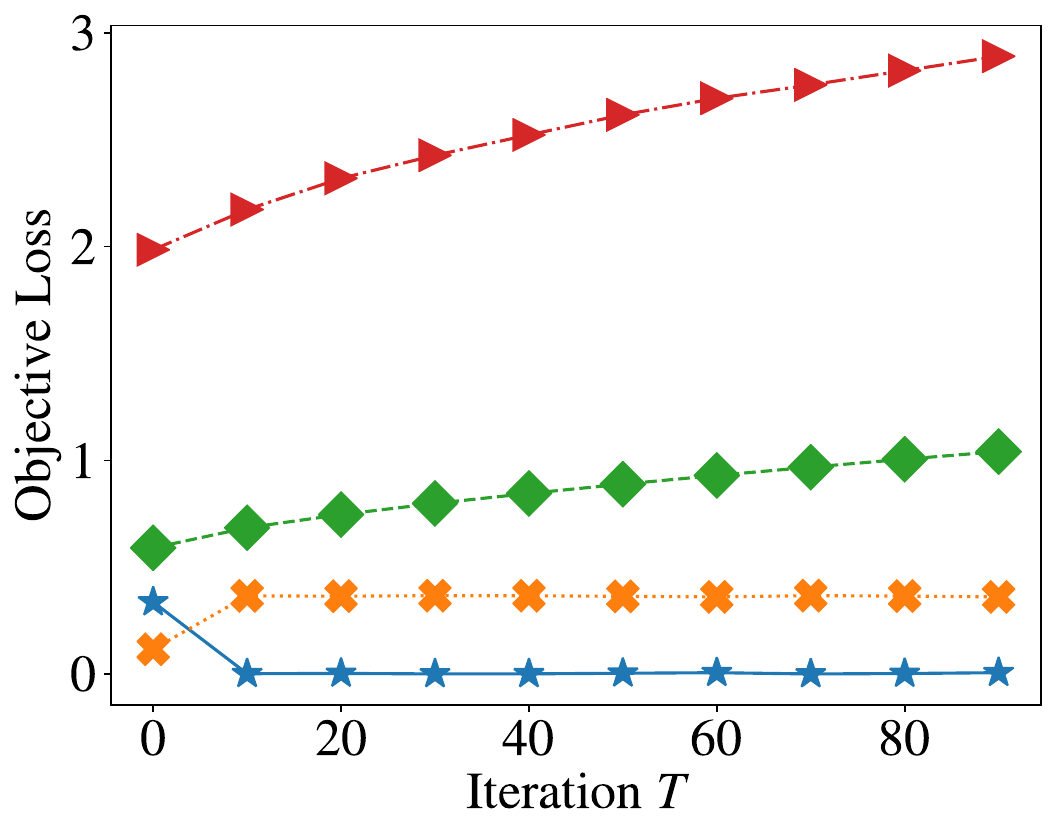}
	}
	\subfloat[~10000 samples]{
		\includegraphics[width=0.2\textwidth]{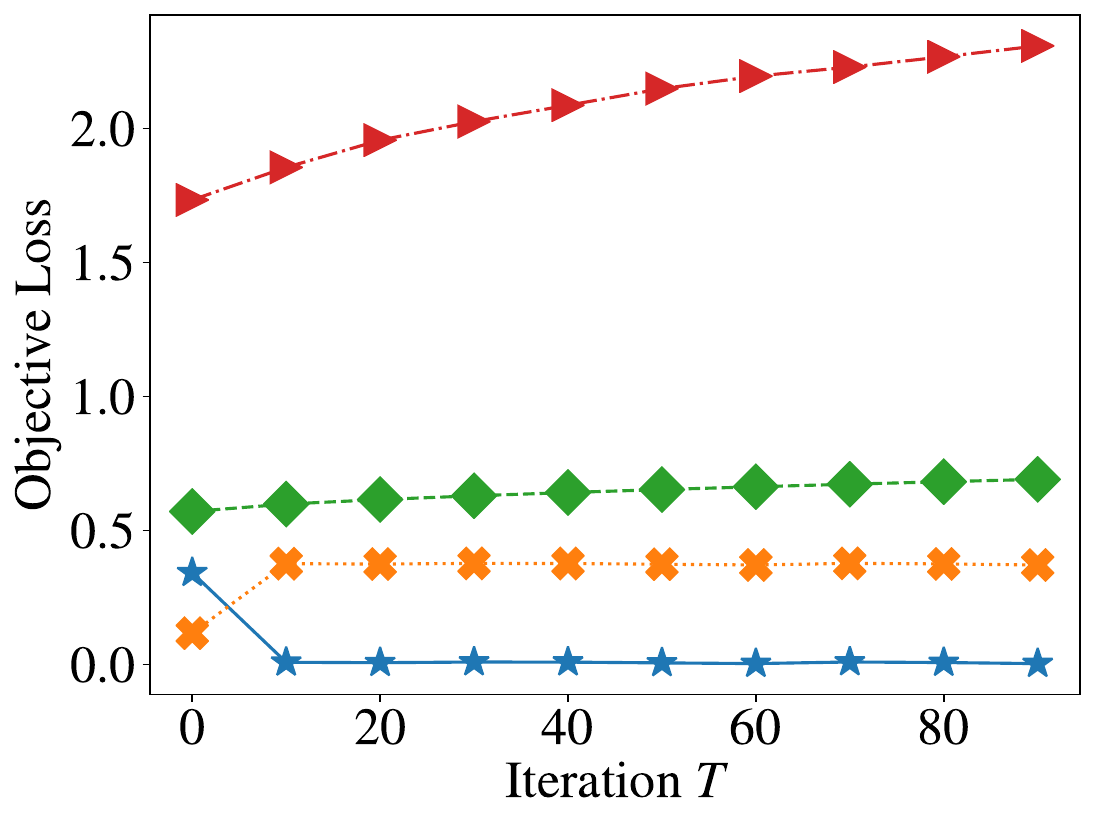}
	}
	\subfloat[~100000 samples]{
		\includegraphics[width=0.2\textwidth]{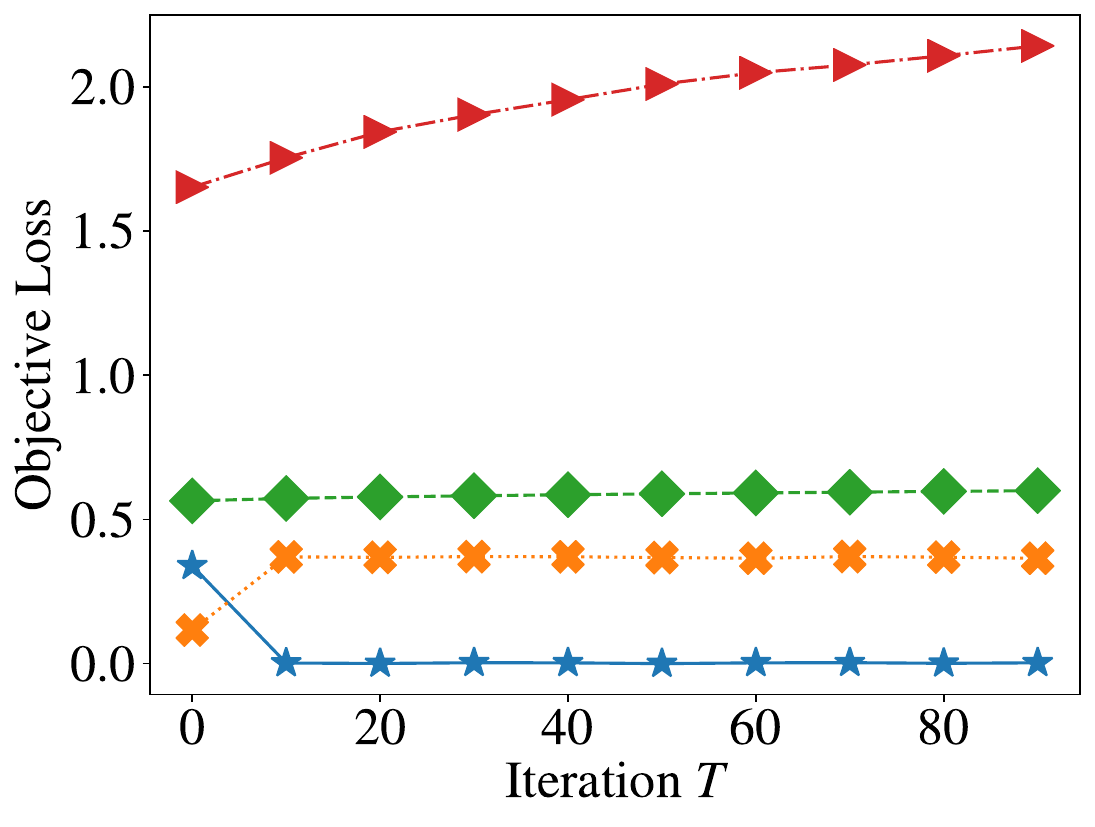}
	}\\
	\subfloat[~$\beta$ = 0.2]{
		\includegraphics[width=0.2\textwidth]{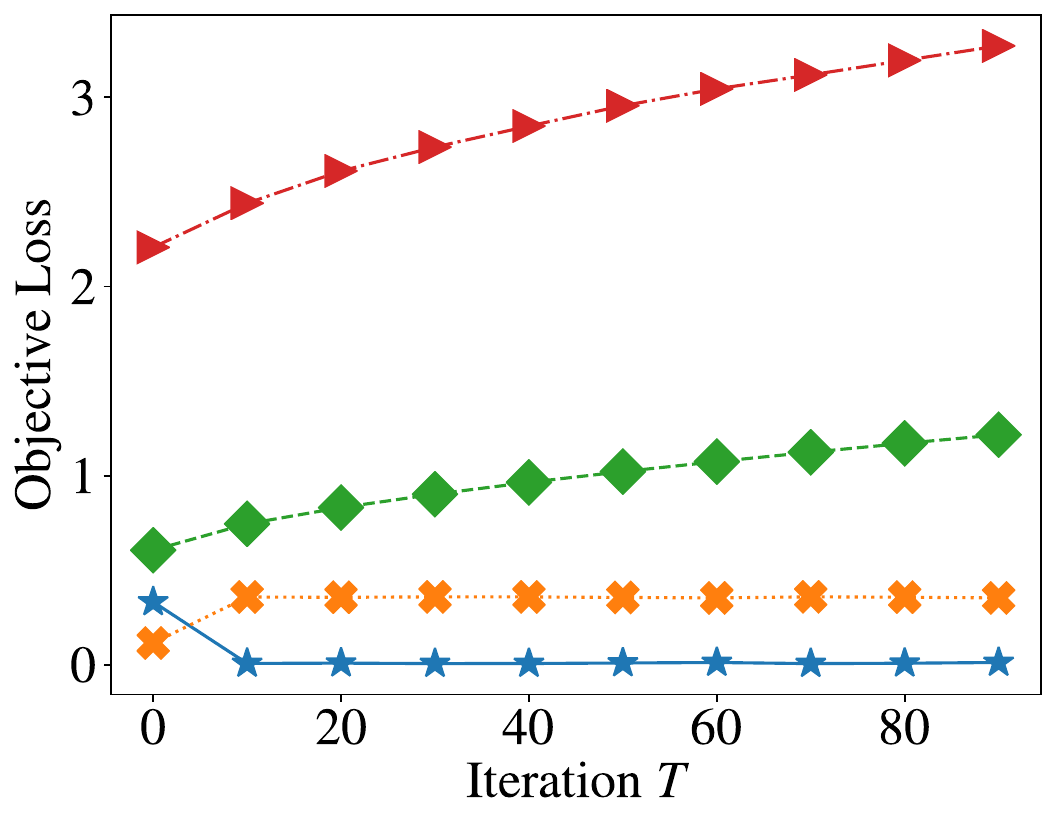}}
	\subfloat[~$\beta$ = 0.5]{
		\includegraphics[width=0.2\textwidth]{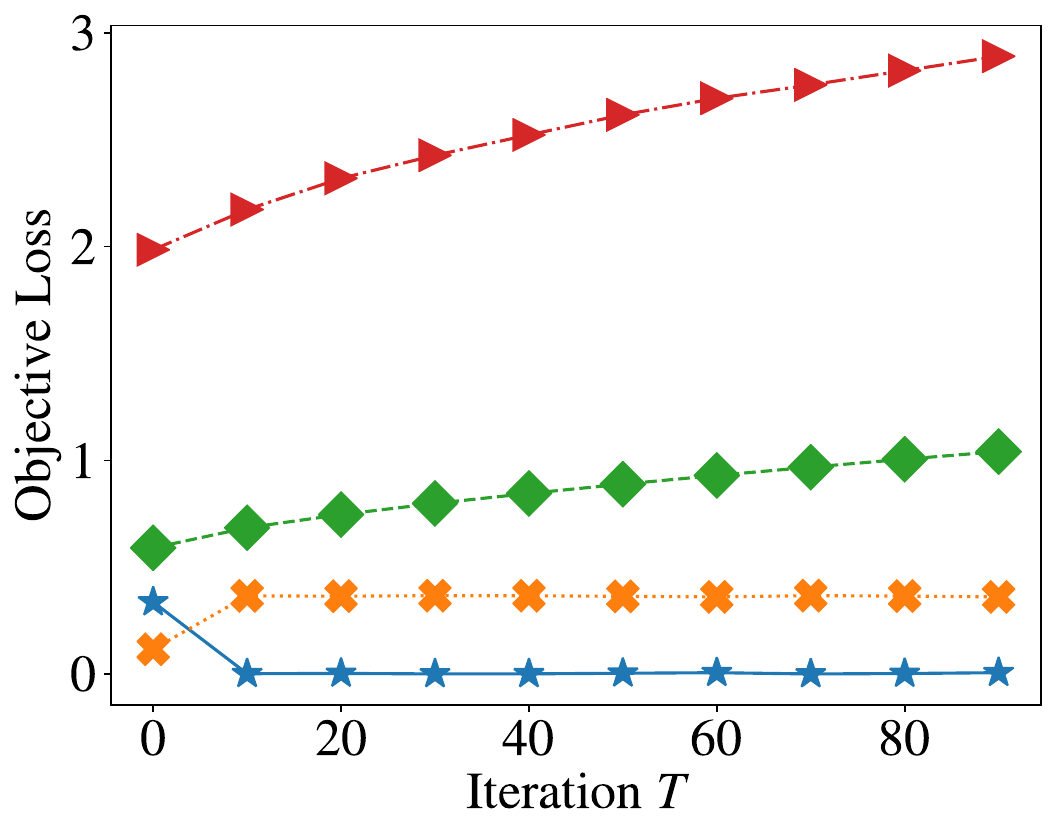}
	}
	\subfloat[~$\beta$ = 0.8]{
		\includegraphics[width=0.2\textwidth]{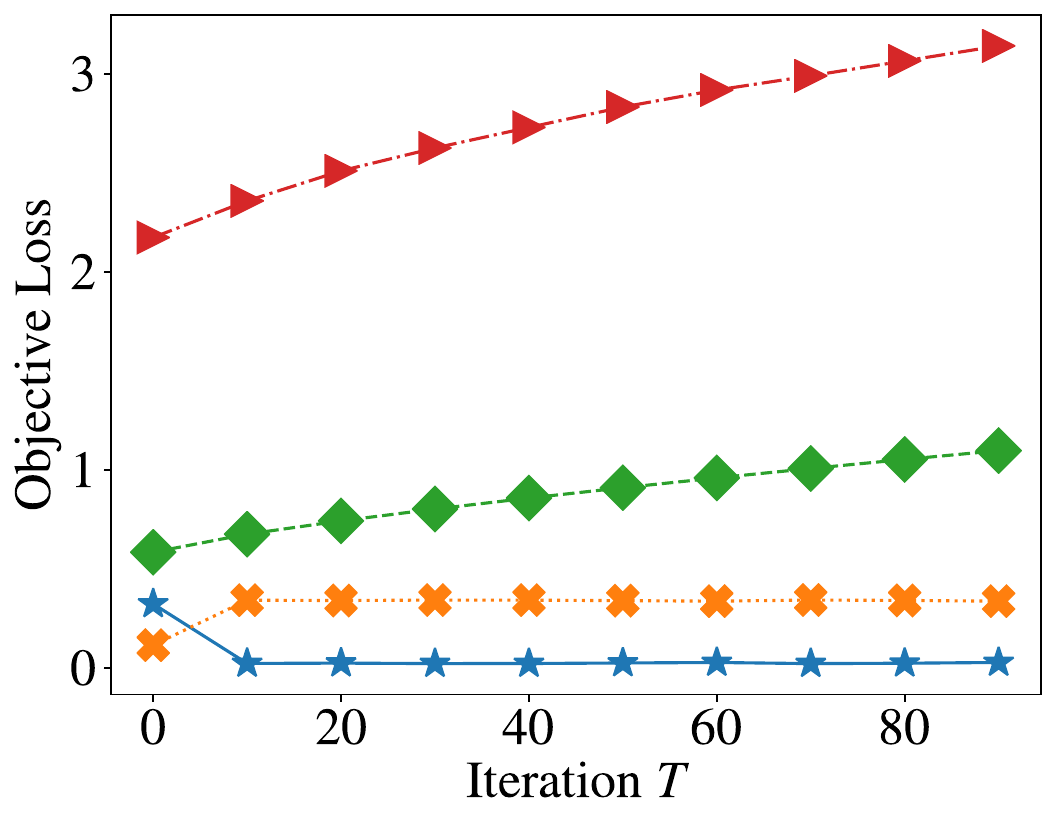}
	}
	\subfloat[~$\beta$ = 0.99]{
		\includegraphics[width=0.2\textwidth]{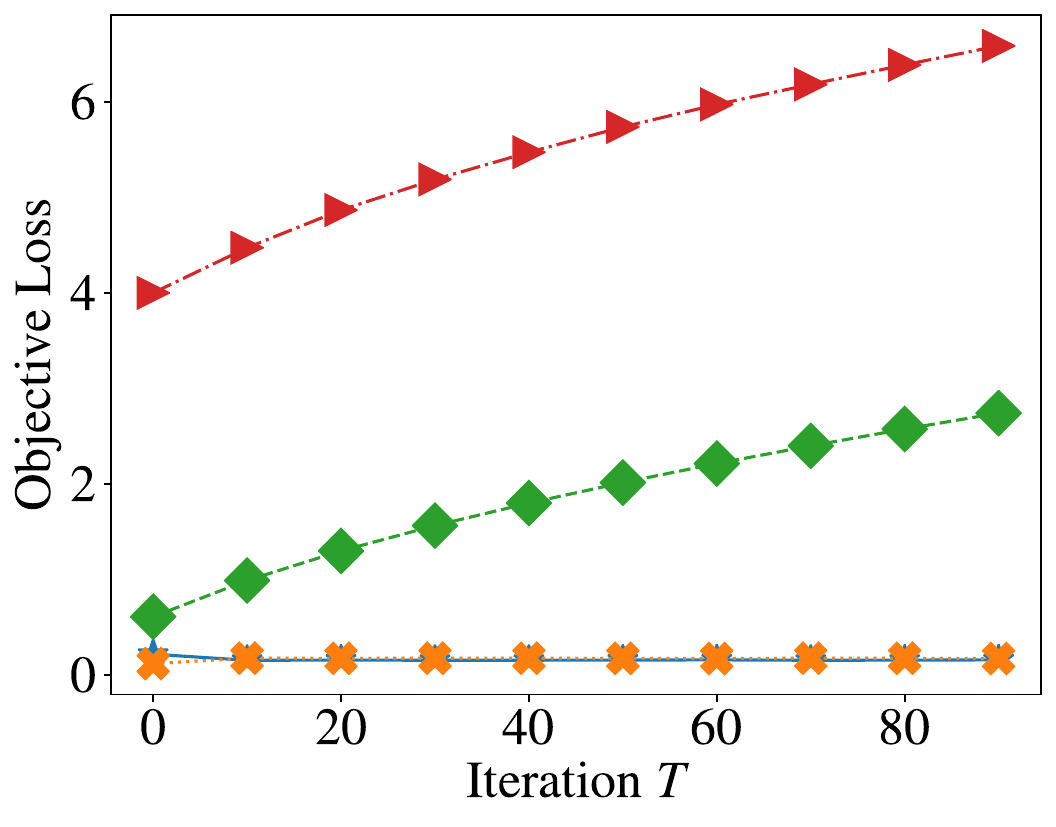}
	}\\
	\subfloat[~$W(0)$ = 0.2]{
		\includegraphics[width=0.2\textwidth]{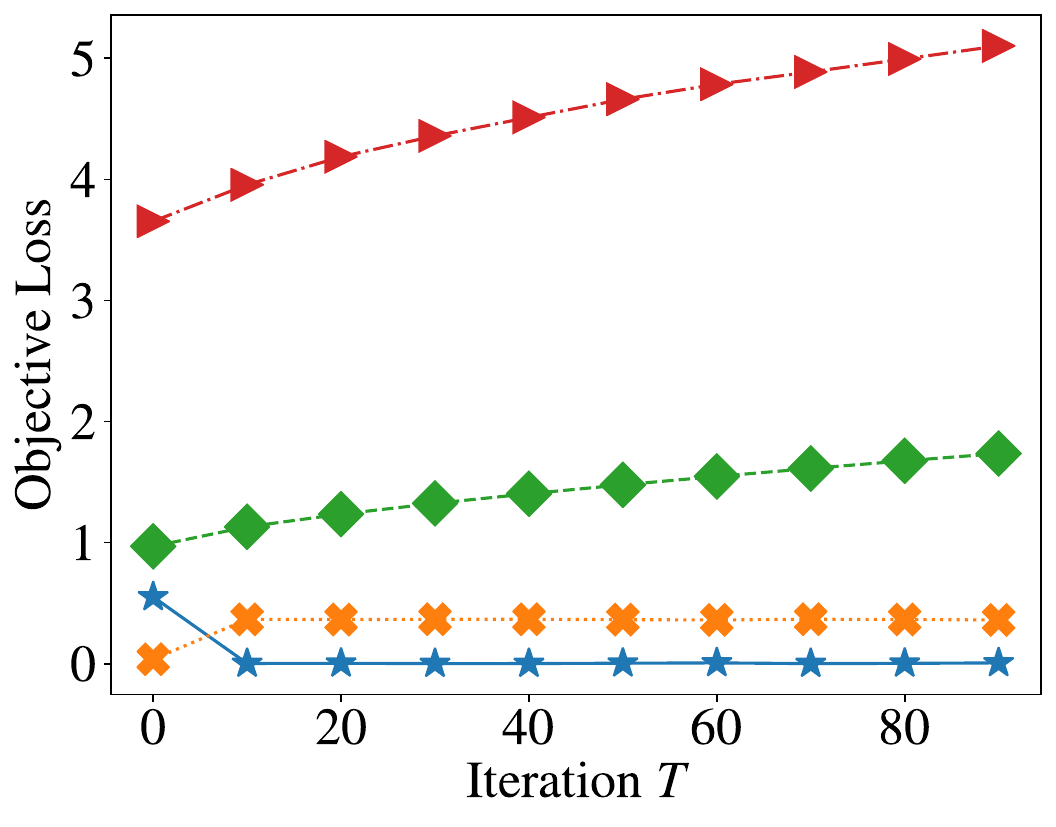}}
	\subfloat[~$W(0)$ = 0.4]{
		\includegraphics[width=0.2\textwidth]{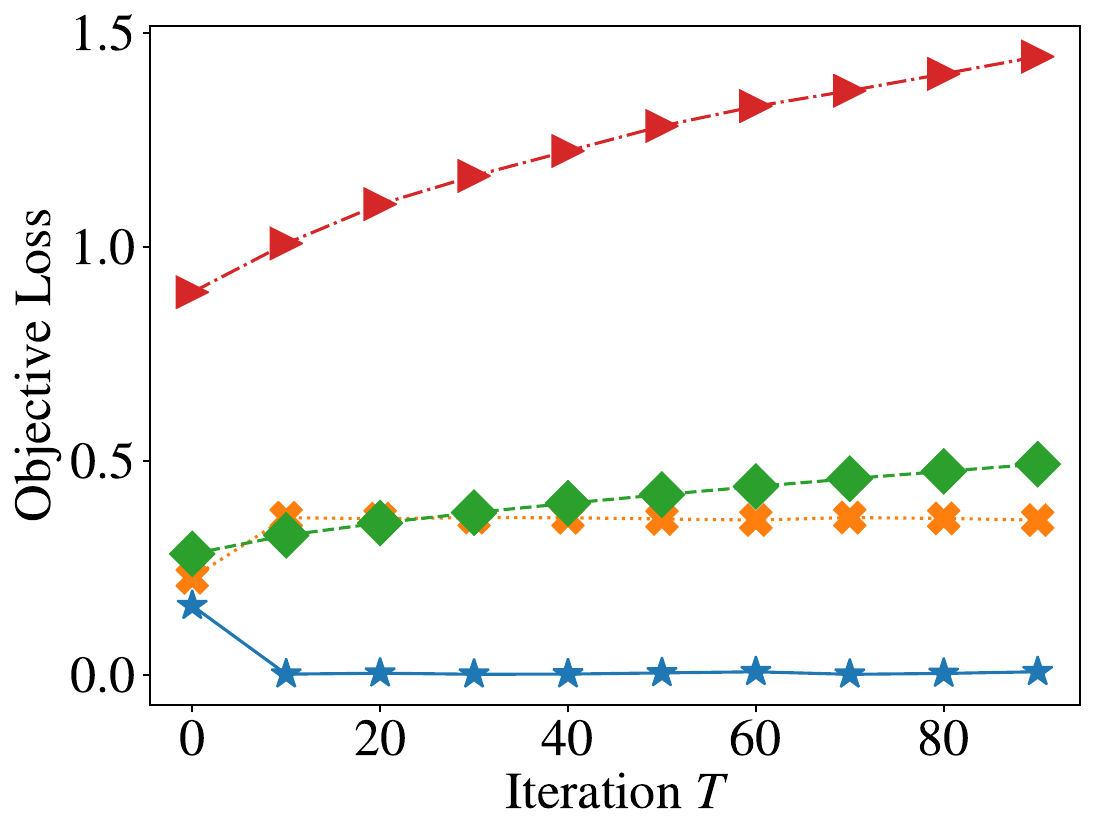}
	}
	\subfloat[~$W(0)$ = 0.6]{
		\includegraphics[width=0.2\textwidth]{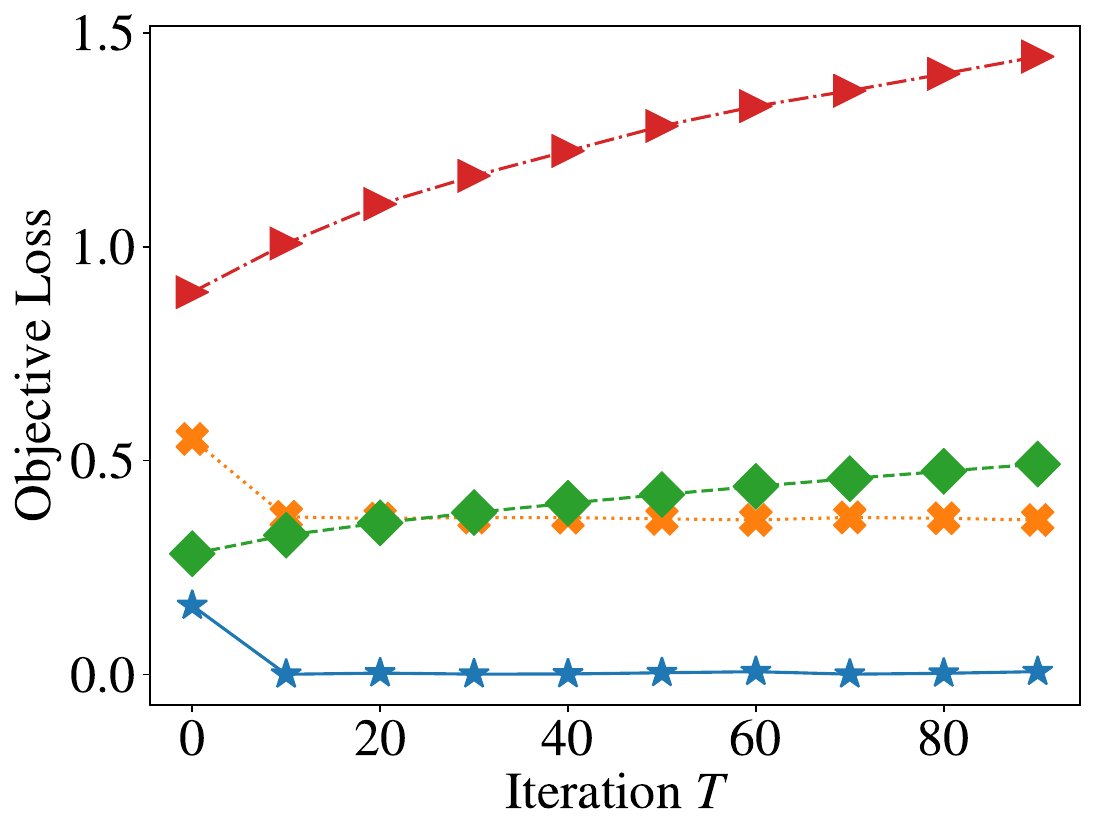}
	}
	\subfloat[~$W(0)$ = 0.8]{
		\includegraphics[width=0.2\textwidth]{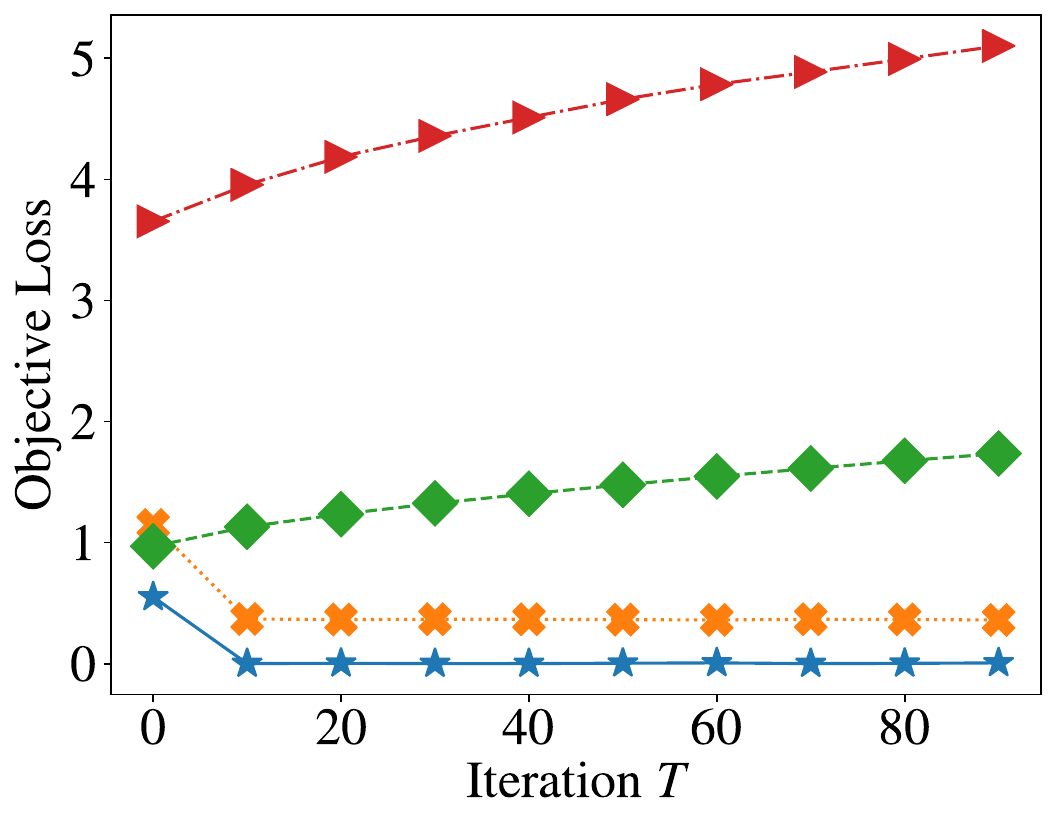}
	}
    \caption{Comparisons for testing results of true generalization error(blue), generalization error bound (green), true excess risk (orange) and excess risk bound(red). We set a series of parameters $\alpha=0.5, p^{\prime}=w^{*}=0.1, p=0.9, w_{\mathrm{ERM}}=0.5, T=100$, $K_{S T}(0)=10, \eta(0)=0.1, \sigma_{t}=\sqrt{\theta \eta(t) / t}, \theta=0.001, W(0)=0.3$ and $\delta=0.01$ to be fixed for all experiments. For comparison tests, we set $\beta=0.5, W(0)=0.3$ for the first row, $n=1000, W(0)=0.3$ for the second row, and $\beta=0.5, n=1000$ for the last row, respectively.}\label{fig:result}
\end{figure*}

\subsection{Transfer with stochastic gradient descent} \label{sec:bernoulli_sgd}
In this example, we assume that data samples $(Z'_1,\ldots, Z'_{\beta n})\sim Ber(p)$ and $(Z_{\beta n+1}, \ldots, Z_n)\sim Ber (p')$ where $Ber(p)$ denotes the Bernoulli distribution with probability $p$.  For any $z_i \in (S,S')$, the loss function is defined as binary cross-entropy
\begin{equation*}
\ell(w, z_i) = -(z_i\log(w) + (1-z_i)\log(1-w)).
\end{equation*} 
In this case, there is no closed-form solution to the ERM algorithm so  we apply the algorithm in (\ref{eq:iterative}) to obtain an hypothesis $W(T)$ where we choose $n(t)\sim N(0,\sigma^2_t)$. 

The numerical results are shown in Figure~\ref{fig:result}, where the calculation for the upper bounds and detailed experimental setups  can be found in Appendix~\ref{setup:bernoulli}.  In the first row, we compare different number of samples while in second row we compare the bounds when $\beta$ varies, and the last row investigates the effect of initialization value $W(0)$.

From the results, it is obvious that both the excess risk and generalization error are upper bounded by  our developed upper bounds. We notice that that the tightness of the bound varies for different parameters of the algorithm. For example, 
if the initial value $W(0)$ is close to $W_{\ERM}$, the upper bounds are tighter in this case. Also one observes that our bounds are in general becomes tighter if the number of samples $n$ increases.  The results  confirms that the bounds captures the dependence of the input data and output hypothesis, as well as the stochasticity of the learning algorithm.

\subsection{Logistic regression transfer} \label{subsec:logit}
In this section, we apply our bound in a typical classification problem. Consider the following logistic regression problem in a 2-dimensional space shown in Figure~\ref{fig:LogisticRegression}. For each $w \in \mathbb{R}^2$ and $z_i = (x_i,y_i) \in \mathbb{R}^{2} \times \{0,1\}$, the loss function is given by
\begin{align*}
    \ell(w,z_i) := & -(y_i\log (\sigma(w^Tx_i)) \\ 
    &+ (1-y_i)\log (1 - \sigma(w^Tx_i))),
\end{align*}
where $\sigma(x) = \frac{1}{1+e^{-x}}$. 
\begin{figure}[h!]
    \centering
    \includegraphics[width=0.4\textwidth]{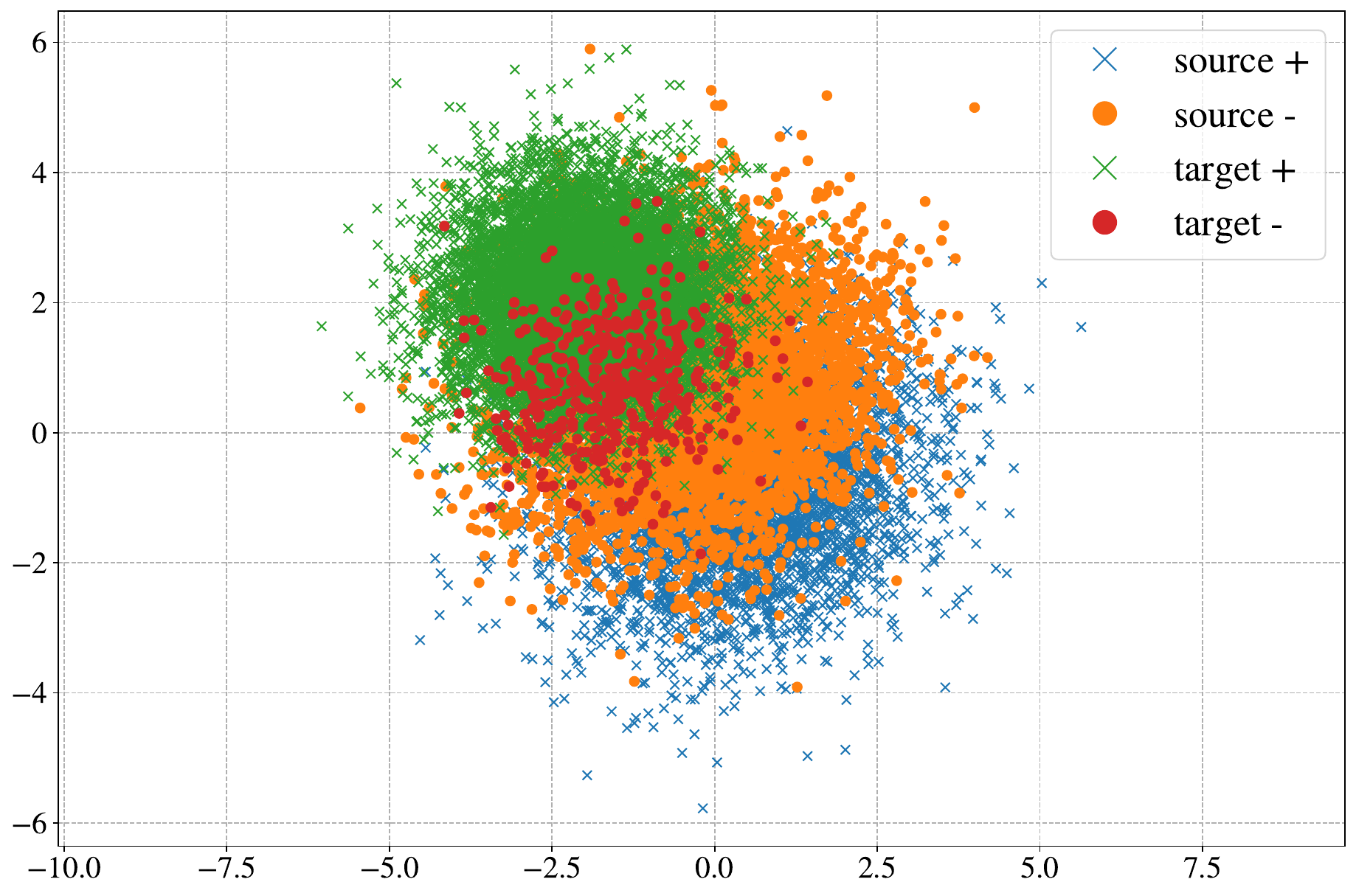}
    \caption{The source data $x_i$ are sampled from the {\bf truncated} Gaussian distribution $\mathcal{N}_{tc} \sim (\mathbf{0},2\mathbf{I})$ while the target data are sampled from the {\bf truncated} Gaussian distribution $\mathcal{N}_{tc} \sim ((-2,2),\mathbf{I})$. The according label $y \in \{0, 1 \}$, is generated from the Bernoulli distribution with probability $p(1) = \frac{1}{1+e^{-w^Tx}}$, where $w_s = (0.5,-1)$ for the source and $w_t = (-0.5,1.5)$ for the target.}
    \label{fig:LogisticRegression}
\end{figure}

Here we truncate the Gaussian random variables $x_i = \{(x_1,x_2) \big| \|x_1\|_2 < 6, \|x_2\|_2<6 \}$, for $i = 1,\cdots,n$. We also restrict hypothesis space as $\mathcal{W} = \{w: \|w\|_2 < 3\}$ where $W_{\ERM}$ falls in this area with high probability. It can be easily checked that $\mu \ll \mu'$ and the loss function is bounded, hence we can upper bound generalization error using Corollary \ref{coro:general_bound}. To this end, we firstly fix the source samples $n_s = 10000$, while the target samples $n_t$ varies from 100 to 100000 and $\alpha = \beta = \frac{n_t}{n_s + n_t}$ following the guideline from \cite{ben-david_theory_2010, zhang2012generalization}. We give the empirical estimation for $r^2$ within the according hypothesis space such that
\begin{small}
\begin{align*}
      r^2 = \frac{\left(\max_{Z \in \mathcal{Z}, w \in \mathcal{W}}{\ell(w,Z)} - \min_{Z \in \mathcal{Z}, w \in \mathcal{W}}{\ell(w,Z)}\right)^2}{4}. 
\end{align*}
\end{small}
To evaluate the mutual information $I(W_{\ERM},Z_i)$ efficiently, we follow the work \cite{moddemeijer1989estimation} by repeatedly generating $W_{\ERM}$ and $Z_i$. As $\mu \ll \mu'$, we decompose $D(\mu(X,Y)\|\mu'(X',Y')) = D(\mu(X)\|\mu'(X))+  D(\mu(Y|X)\|\mu'(Y|X)|X)$ in terms of the feature distributions and conditional distributions of the labels. The first term $D(P_X\|P_{X'})$ can be calculated using the parameters of Gaussian distributions. The latter term denotes the expected KL-divergence over $P_{X}$ between two Bernoulli distributions, which can be evaluated by generating abundant samples from the source domain. Further, we apply Theorem \ref{thm:excess} to upper bound the excess risk, where we give a data-dependent estimation for the term $d_{\mathcal W}(\mu, \mu')$ as
\begin{align*}
\hat d_{\mathcal W}(\mu, \mu')=\sup_{w\in\mathcal W} |\hat L(w, S)-\hat L(w, S')|.
\end{align*}
To demonstrate the usefulness of our algorithm, we compare the bound in the following theorem using the Rademacher complexity under the same domain adaptation framework. 

\begin{theorem} (Generalization error of ERM with Rademacher complexity) \cite[Theorem 6.2]{zhang2012generalization}
Assume that for any $w\in\mathcal W$, the loss function $\ell(w, Z)$ is bounded between $[a,b]$ for any $w$ and $z$. Then for any $\delta>0$, the following inequality holds with probability at least $1-\delta$  (over the randomness of samples and the learning algorithm):
\begin{align*}
&\gen(w_{\ERM})  \leq (1-\alpha) d_{\mathcal W}(\mu,\mu') + 2 \alpha \mathbb{E}_{\sigma\otimes \mu}\left[ \sup_{w \in \mathcal{W}}\sigma \ell(w,Z) \right] \\ 
& + \frac{2(1-\alpha)}{\beta n} \mathbb{E}_{\sigma}\left[ \sup_{w \in \mathcal{W}} \sum_{i=1}^{\beta n} \sigma_{i} \ell \left(w, z_i \right) \right] + 3 \alpha \sqrt{\frac{(b-a)\ln (4 / \delta)}{2\beta n}} \\
& + (1-\alpha) \sqrt{\frac{(b-a)^2}{2}\ln (\frac{2}{\delta})\left(\frac{\alpha^{2}}{\beta n}+\frac{(1-\alpha)^{2}}{(1-\beta) n}\right)},
\end{align*} 
where $\sigma$ (and $\sigma_i$) are randomly selected from \{-1,+1\} with equal probability.
\label{thm:RC-ER}
\end{theorem}
The comparisons of generalization error bound and excess risk bound are shown in figure \ref{res:LogisticRegression}. It is obvious that the true losses are bounded by our developed upper bounds. The result also suggests that our bound is tighter than the Rademacher complexity bound in terms of both generalization error and excess risk. This is possibly due to that the generalization error bound with Rademacher complexity is characterized by the domain difference in the whole hypothesis space, while our bound is data-algorithm dependent, which is only concerned with $W_{\ERM}$. As expected, the data-algorithm dependent bound captures the true behavior of generalization error while the Rademacher complexity bound fails to do so. It is noteworthy that both bounds converge as $n$ increases. The result confirms that the bounds capture the dependence of the input data and output hypothesis, as well as the stochasticity of the algorithm.
\begin{figure*}[!htp]
    \centering
    \includegraphics[width=0.8\textwidth]{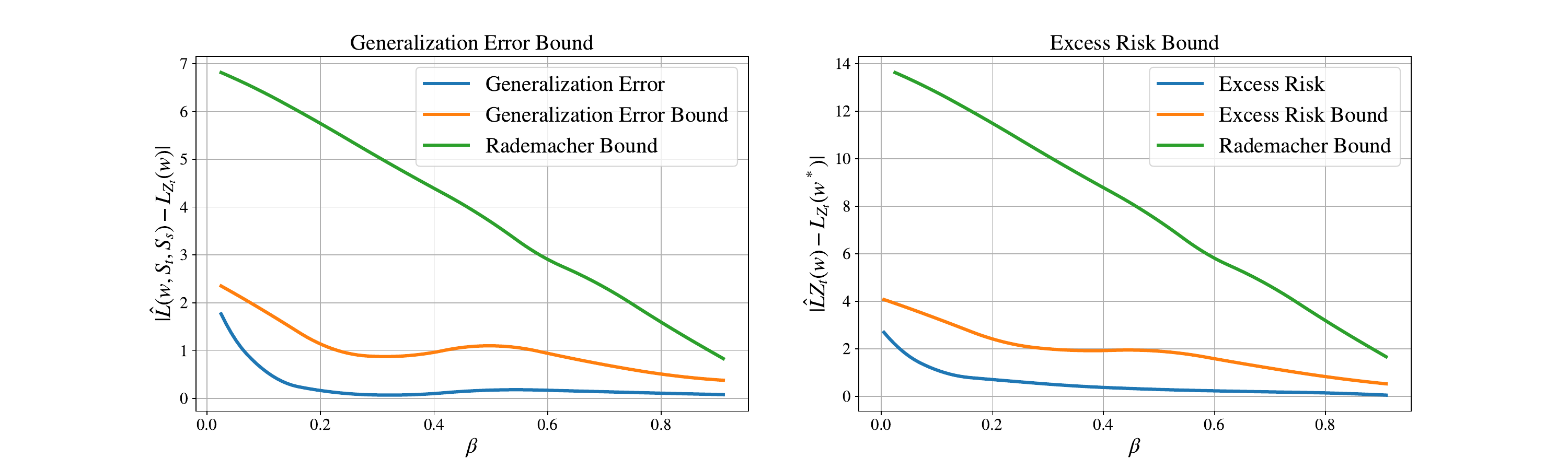}
    \caption{Comparisons for generalization error and excess risk where we fix $n_s = 10000$ and vary $n_t$ by setting $\alpha = \beta$.}
    \label{res:LogisticRegression}
\end{figure*}

\subsection{Fast Rate Logistic Regression}\label{sec:example-fast}
In this section, we apply our fast-rate bound in the logistic regression problem in a 2-dimensional space to further evaluate the effectiveness of the bounds in Equation~(\ref{eq:fast-erm}). Again, we assume each $w \in \mathbb{R}^2$ and $z_i = (x_i,y_i) \in \mathbb{R}^{2} \times \{0,1\}$ as a similar setup in the previous section. Now we assume the source-only scenario where $\alpha = \beta = 0$, and each $x_i$ is drawn from a standard multivariate Gaussian distribution $\mathcal{N}(0,\mathbf{I}_{2})$ and Let $w_s = (0.5,0.5)$ and $w_t = (-0.5, -0.5)$ for a different setup for generalizing the experimental results. We also restrict hypothesis space as $\mathcal{W} = \{w: \|w\|_2 < 2\}$ where $W_{\ERM}$ falls in this area with high probability. Since the hypothesis is bounded and under the log-loss, then the learning problem will satisfy the central and witness condition \cite{van2015fast,Grunwald2020}. Therefore, it will satisfy the $(\eta,c)$-central condition. To estimate $\eta$, $c$, and the mutual information $I(W_{\ERM}, Z_i)$ efficiently, we repeatedly generate samples of $W_{\ERM}$ and $Z_i$ and use their empirical density for estimation. Specifically, we vary the sample size $n$ in the range $[50, 400]$ and for each value of $n$, we repeat the logistic regression algorithm $2000$ times to generate a set of $W_\ERM$ samples. By setting $\eta = 0.8$ as an example, we can empirically estimate the CGF and the expected excess risk using the data sample and a set of ERM hypotheses, which results in an estimate of $c \approx 0.195$. Importantly, our experiments indicate that once $\eta$ is fixed, the choice of $c$ remains independent of the sample size $n$, providing empirical support for the ($\eta,c$)-central condition. For the mutual information estimation, we used a similar method as illustrated in the previous section by decomposing the mutual information into marginal and conditional divergences. To demonstrate the usefulness of the results, we also compare the bounds among the true excess risk, the slow rate excess risk in Equation~(\ref{eq:excess_erm_bound_mi}) and the fast rate excess risk in Equation~(\ref{eq:fast-erm}). The comparisons are shown in Figure~\ref{fig:logistic}.
\begin{figure*}[htp]
	\centering
	\subfloat[Excess Risk]{\includegraphics[width=0.33\textwidth]{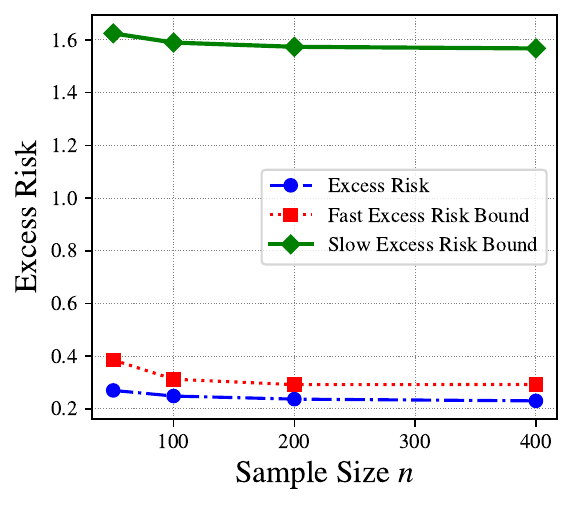}}	\quad
	\subfloat[Convergence]{\includegraphics[width=0.38\textwidth]{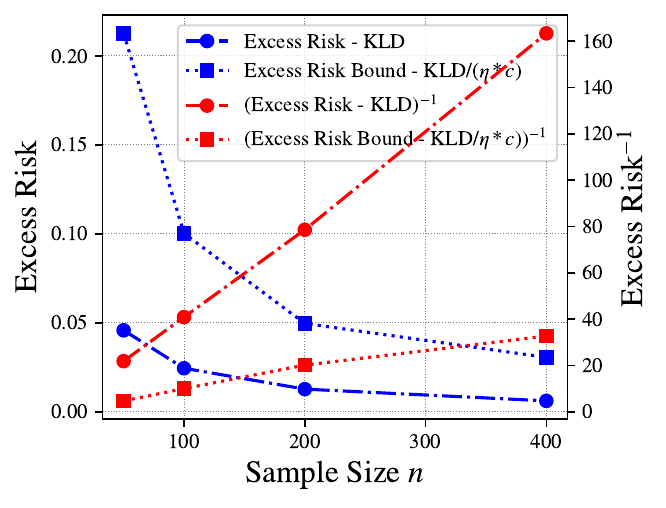}}
    \caption{We represent the true expected generalization error in (a) along with its bounds in Theorem~\ref{thm:excess} and Theorem~\ref{thm:central-transfer}. Here we vary $n$ from 50 to 400. To show the convergence up to the domain divergence, we also plot the quantity $\mathbb{E}_W[R_{\mu'}(W)] - D(\mu\|\mu')$ and the fast rate bound $\frac{1}{c\eta n}\sum_{i=1}^{n}I(W;Z_i)$, along with their reciprocals to show the rate w.r.t. sample size $n$. All results are derived by 2000 experimental repeats.}\label{fig:logistic}
\end{figure*}
From the figure, it can be seen that the fast excess risk bound in Equation~(\ref{eq:fast-erm}) is even tighter than the slow excess risk bound in Equation~(\ref{eq:excess_erm_bound_mi}). Moreover, both the excess risk and its fast rate bound exhibit linear convergence scaling as $O(\frac{1}{n})$, up to the domain divergence (with some leading constant). Importantly, this simple toy example showed that the bounds presented in Theorem~\ref{thm:central-transfer} are tight, accurately reflecting the true behaviours at the same decay rate.

\subsection{Algorithms and Real Dataset}\label{sec:info}
Many information-theoretic bounds are challenging to apply directly to real-world machine-learning tasks primarily because they require knowledge of the data and the hypothesis distributions, and estimating these bounds can be very difficult. Recognizing this gap, we introduce a heuristic boosting algorithm named \texttt{InfoBoost} where this innovative approach leverages the concepts behind the proposed bounds.

With the standard empirical risk minimization algorithm, both the source and target instances are equally weighted within their domains as we define the empirical risk by Eq~(\ref{eq:er}). However, in many real-world scenarios, the weight of each instance can be different, especially for the transfer learning problem. To select the source data that are useful for the prediction on the target domain, we can re-weight each source instance so that the source data will have a similar performance as the target data. To interpret, we start from the following fact on change of measure for the expected risk if $\mu'$ is absolutely continuous w.r.t. $\mu$:
\begin{equation}
\mathbb{E}_{\mu'}[\ell(W,Z)] = \mathbb{E}_{\mu}\left[\frac{\mu'(Z)}{\mu(Z)}\ell(W,Z) \right].
\end{equation}
Ideally, if we can find the distribution density ratio $\frac{\mu'(z_i)}{\mu(z_i)}$ for each instance in the source domain, the expected risk induced by the source is obviously an unbiased estimate of the expected risk under the target distribution. However, this quantity is non-trivial to estimate. For the supervised learning problem where $Z_i = (X_i, Y_i)$, under the co-variate assumption such that $P(Y|X)$ remains unchanged across different domains, this ratio can be estimated through statistical methods such as correcting sample selection bias \cite{huang2006correcting}, kernel mean matching \cite{gretton2009covariate} and direct density ratio estimation \cite{sugiyama2012density}. On the other hand, the iterative algorithm based on model aggregation (e.g., boosting) has been proposed to adaptively adjust the instance weights in transfer learning problems. To name a few, \citet{dai_boosting_2007} take the first step in the application of boosting to transfer learning where the weights for the target data are adjusted using Adaboost \cite{freund1997decision} while for the source the weights are gradually decreasing following the weighted majority algorithm. Such a strategy is known as the TradaBoost algorithm. \citet{eaton2011selective} further develop the TransferBoost algorithm where the source instances are re-weighted according to whether the source domain can improve or hurt the prediction performance on the target. More recently, \citet{wang2019transfer} propose the GapBoost algorithm by taking the hypothesis performance gap between two domains into consideration.

Inspired by the information-theoretic bounds and the boosting algorithm, we heuristically develop an information-theoretic-based boosting algorithm where the re-weighting scheme will involve mutual information $I(W;Z_i)$ and the domain divergence $D(\mu\|\mu')$, which appear in our derived upper bounds. From the excess risk bound~(\ref{eq:excess_erm_bound_mi}), we can see that the learning performance is controlled by the mutual information terms $I(W;Z_i)$, the domain divergence $D(\mu\|\mu')$, the summation of the weights in the target domains controlled by $\alpha$. Instead of using the same weights for all instances in the same domain, we assign different weights to different instances to show their importance and then introduce the novel boosting type algorithm based on the mutual information bound, namely, the \texttt{InfoBoost} algorithm. For simplicity, we consider the binary classification task where $\mathcal{Z} = \mathcal{X} \times \mathcal{Y}$ and present the algorithm in Alg~\ref{alg:infoboost}. The multi-label task and regression boosting algorithms can be extended following similar steps as shown in \cite{freund1997decision}.
\begin{algorithm}[h!]
\SetAlgoLined
\SetKwInOut{Input}{Input}
\SetKwInOut{Output}{Output}
 \Input{Source Sample $S$, Target Sample $S'$, Iteration $T$, Loss function $\ell(h,x,y)$}
 Initialize weights $\gamma_1(i) = \frac{1}{n}$ for all $i$\;
 Choose hyper-parameters $\Gamma$, $\zeta$, $\eta$ and iteration $T$\;
 \For{$t = 1,\cdots,T$}{ 
 Learn a base hypothesis $h_t$ using $S \cup S'$\; 
 Learn $h^{-i}_t$ using $S \cup S'$\; 
 Learn domain hypothesis $h_{S'}$ and $h_S$ using $S'$ and $S$ separately\;
 Learn a domain discriminator $h_{\textup{dis}}$ for the domain features\;
 $\epsilon_{t}=\sum_{i=1}^{n\beta } \gamma_t(i) \mathbf{1}_{h_{t}\left(x_{i}\right) \neq y_{i}}+\sum_{i=\beta n + 1}^{n} \gamma_t(i) \mathbf{1}_{h_{t}\left(x_{i}\right) \neq y_{i}}$\;
 $\alpha_t = \log \frac{1-\epsilon_{t}}{\epsilon_t}$\;
 \For{$i = 1,\cdots,\beta n$}{
    Incur stability loss $d_i(h_t, h^{-i}_t)$ according to (\ref{eq:mi-stable})\;
    $\gamma_{t+1}(i) = \gamma_{t}(i) \cdot e^{ \alpha_t\ell(h_t,x_{i}, y_{i}) - \eta d_i(h_t, h^{-i}_t)}$ \;
 }
  \For{$i = \beta n + 1,\cdots, n $}{
    Incur stability loss $d_i(h_t, h^{-i}_t)$ according to (\ref{eq:mi-stable})\;
    Incur domain discrimination loss $\ell(h_{\textup{dis}}(x_i))$ according to (\ref{eq:domain})\;
    Incur labelling divergence loss $d(h_{S'}(x_i, y_i), h_{S}(x_i, y_i))$ according to (\ref{eq:labeling})\;
	$\gamma_{t+1}(i) = \gamma_{t}(i) \cdot e^{\alpha_t \ell(h_t,x_{i}, y_{i})  - \eta d_i(h_t, h^{-i}_t) - \zeta (\ell(h_{\textup{dis}}(x_i)) + d_i(h_S, h_{S'})) )}$\;
 }
 Let $Z_t = \sum_{i=1}^{n}\gamma_t(i)$\;
 Normalize $\gamma_t(i) = \frac{\gamma_t(i)}{Z_t}$ for $i = 1,2,\cdots,n$\;
 }
 \Output{$f(x) = \mathbf{1}_{\sum_{t=1}^{T}\alpha_th_t(x) \geq \frac{1}{2}\sum_{t=1}^{T}\alpha_t}$.}
 \caption{InfoBoost Algorithm}
 \label{alg:infoboost}
\end{algorithm}

 The procedures of the algorithm are sketched as follows:
\begin{itemize}
    \item Given both source and target training data, firstly we initialize the weights $\gamma_1(i), i = 1,\cdots, n$ uniformly (arbitrary initialization is also acceptable). The subscript for all variables in the algorithm indicates the timestamp. 
    
    \item We will then update weights iteratively in each round. At each timestamp $t$, we first learn a base hypothesis denoted by $h_t: \mathcal{X} \rightarrow \mathcal{Y}$ with the weights $\gamma_t$ using both source and target data. To introduce the concept of information stability for a single instance $z_i$, we need to train the hypothesis without $i$-th sample and compare it with $h_t$, which we denote by $h^{-i}_t$.  
    
    \item To capture the domain divergence, we decompose the KL divergence as $D(\mu(X,Y)\|\mu'(X,Y)) = D(\mu(X),\mu'(X)) + D(\mu(Y|X)\|\mu'(Y|X)|X)$ into two terms. Instead of directly estimating the KL divergence, we may use surrogate quantities learned from the data to approximate these two KL divergences. In practice, we train a domain discriminator $h_{\textup{dis}}: \mathcal{X} \rightarrow [0,1]$ to capture the feature divergence $D(\mu(X),\mu'(X))$, which takes the features $X$ as the input and the probability of being the target domain as the output. In other words, if $h_{\textup{dis}}(x)$ is closer to 1, then $x$ is more likely to be drawn from the target domain, and the divergence should be small. In addition, we train two domain-specific classifiers $h_{S'}: \mathcal{X} \rightarrow \mathcal{Y}$ and $h_S: \mathcal{X} \rightarrow \mathcal{Y}$ using $S'$ and $S$ respectively to capture the condition divergence term $D(\mu(Y|X)\|\mu'(Y|X)|X)$, where the differences of their predictions on the same input can be heuristically interpreted as the labelling divergence. 
    
    \item Having outlined the basic concept, let us explain in detail how the algorithm works. We firstly evaluate $h_t$ over $S$ and $S'$ to obtain the error rate $\epsilon_t$ and  $\alpha_t$ similarly as in \cite{freund1997decision}. For each target data, we use the exponentially updating rule following the boosting strategy in \cite{freund1997decision}. In addition to that, we also examine the difference between $h_t$ and $h^{-i}_t$, which can be heuristically interpreted as the information stability loss by:
    \begin{align}
        d_i(h_t, h^{-i}_t) = |h_t(x_i) - h^{-i}_t(x_i)|. \label{eq:mi-stable}
    \end{align}
    A smaller $d_i$ indicates that the sample $z_i$ has less effect on learning the hypothesis. Nonetheless, learning $h^{-i}_t$ may be very time-consuming due to the large training sample size $n$. Instead, by evenly splitting the training sample into $K$ folds where $K$ is much smaller than the sample size $n$, we will train a set of base classifier $h_k, k = 1,2,\cdots,K$ for each fold from the rest $K-1$ folds. That is, we will use the same hypothesis $h_k$ for all data in $k$-th fold to approximate $h^{-i}_t$. We choose $K = 20$ for the subsequent experiments. For each source data, we will further take the domain divergence into account, which is decomposed into two parts as mentioned earlier. The first is the divergence of the feature of $X$ induced by the domain discriminator $h_{\textup{dis}}$: 
    \begin{align}
        \ell(h_{\textup{dis}}(x_i)) = 1 - h_{\textup{dis}}(x_i) . \label{eq:domain}
    \end{align}
    Here, if $\ell(h_{\textup{dis}}(x_i))$ is large, then $i$-th sample is more different from the target data and we shall assign a smaller weight to this instance. Regarding the labelling divergence $D(\mu(Y|X)\|\mu'(Y|X)|X)$, we will utilize the hypothesis $h_S$ and $h_{S'}$ trained from the source and target domains, which can be seen as a mapping from $X$ to $Y$, and define the following quantity to measure their differences:
    \begin{align}
        d_i(h_S,h_{S'}) = |\ell(h_S,x_i,y_i) - \ell(h_{S'},x_i,y_i)|. \label{eq:labeling}
    \end{align}
    A smaller labelling divergence implies that $h_S$ and $h_{S'}$ are more similar. In summary, we use $d_i(h_t, h^{-i}_t)$ as a proxy for the mutual information $I(W;Z_i)$ and $\textup{div}(x_i,y_i) = \ell(h_{\textup{dis}}(x_i)) + d_i(h_S,h_{S'})$ as a proxy for the domain divergence $D(\mu\|\mu')$. We also introduce the hyper-parameters $\zeta$ and $\eta$ to allow more flexible control of these two quantities, and an appropriate choice of which could also prevent us from focusing too much on the particular data.
    
    \item After $T$ iterations, we aggregate all base classifiers $h_t$ with the corresponding weights $\alpha_t$ for some new input $x$ and output the prediction.
\end{itemize}

We now evaluate the proposed InfoBoost algorithm on several typical transfer learning tasks to show its effectiveness. The datasets used are listed as follows.
\begin{itemize}
    \item Office-Caltech-10 \cite{gong2012geodesic}: This dataset contains four subsets, and each domain has a set of office photos with the same 10 classes. In particular, the four subsets are \textbf{W}ebcam ($W$ for short), \textbf{D}SLR ($D$ for short), \textbf{A}mazon ($A$ for short) and \textbf{C}altech dataset ($C$ for short). We use each subset as a domain. Consequently, we get four domains ($A, C, D$ and $W$), leading to 12 transfer learning problems. Since each domain shares the same 10 classes, we therefore constructed 5 binary classification tasks and reported the average error in each transfer learning problem. We use the SURF features as described in \cite{gong2012geodesic} encoded with a visual dictionary of 800 dimensions.
    \item 20 Newsgroups \footnote{\url{http://qwone.com/~jason/20Newsgroups/}}: the dataset contains approximately 20000 reviews from 7 major categories that can be split into 20 subcategories. The source and target domains were picked from the same major categories but different subcategories in each transfer learning task, in the same way as in \cite{wang2019transfer}. 
    \item MNIST and USPS: These two datasets contain black and white hand-written digits from 0 to 9 where MNIST \footnote{\url{http://yann.lecun.com/exdb/mnist/}}  has approximately 70000 images and USPS \footnote{\url{https://www.csie.ntu.edu.tw/~cjlin/libsvmtools/datasets/multiclass.html##usps}} has approximately 10000 images in total. Since each domain shares the same digits from 0 to 9 (but a different writing style), we therefore constructed 5 binary classification tasks and reported the classification error for each task. We use all samples from USPS datasets but only use 10000 images from MNIST.
\end{itemize}
The benchmarks we compare with are the typical boosting algorithms for transfer learning with the same base classifier, e.g., the logistic regression model with very small regularization. We list all the competitors as follows.
\begin{itemize}
    \item \textbf{AdaBoost$_\mathcal{T}$ \cite{freund1997decision}:} The first baseline method is the AdaBoost algorithm with the target training data only, and the initial weights are assigned uniformly.
    \item  \textbf{AdaBoost$_{\mathcal{S} \& \mathcal{T}}$:} We also directly apply the AdaBoost algorithm with both source and target data, and the initial weights are assigned uniformly over all instances.
    \item \textbf{TrAdaBoost \cite{dai_boosting_2007}:}  TrAdaBoost algorithm is the firstly proposed boosting method for transfer learning where at each iteration less weights are assigned to the source and we adaptively focus more on the target data.
    \item \textbf{TransferBoost \cite{eaton2011selective}:} TransferBoost is another boosting algorithm that selects the useful source data by examining whether the source domain improves the learning performance in the target domain. We use the same way as described in \cite{eaton2011selective} to choose the hyper-parameters $\alpha_t$ and $\beta_t$. 
    \item \textbf{GapBoost \cite{wang2019transfer}:} The GapBoost algorithm minimizes the proposed performance gap between the source and target domains by training auxiliary classifiers on both source and target domains. The gap of the predictions of the auxiliary classifiers is regarded as the performance gap. Such a quantity is taken into consideration when updating instance weights. We use the same hyper-parameters as described in the experiments section in \cite{wang2019transfer}. 
\end{itemize}

\begin{table*}[!htp]
\centering
\caption{Accuracy in $\%$ with 10 target data for SURF Office-Caltech Dataset}\label{tab:acwd}
\begin{tabular}{|c|c|c|c|c|c|c|}
\hline
& Adaboost$_\mathcal{T}$           & Adaboost$_{\mathcal{S} \& \mathcal{T}}$           & Tradaboost           & Transfer Boost & GapBoost & InfoBoost \\
\hline
$A\rightarrow W$ & 64.69 & 75.47 & 75.88 & 76.42 & 77.78 & \textbf{78.30}  \\
$A\rightarrow D$ & 71.60 & 75.83 & 76.16 & 75.42 & 76.89 &  \textbf{78.44} \\
$A\rightarrow C$ & 57.19 & 76.88 & \textbf{77.24} & 74.97& 75.08& 74.84  \\
$D\rightarrow A$ & 54.94 & 56.98 & 57.12 & 57.68 & 57.87 & \textbf{58.04}\\
$D\rightarrow W$ & 48.91 & 51.84 & 51.54 & 52.05 & 52.31 & \textbf{52.62} \\
$D\rightarrow C$ & 54.57 & 60.97 & 60.88 & 61.01 & \textbf{61.11} & 61.05 \\
$W\rightarrow A$ & 54.76 & 62.90 & 62.75 & 63.04 & 62.86 & \textbf{63.28} \\
$W\rightarrow D$ & 59.41 & 61.32 & \textbf{61.44} &  60.66 & 60.96 & 61.20 \\
$W\rightarrow C$ & 55.91 & 63.07 & 63.01 & 63.13 & 63.07 & \textbf{63.27}\\
$C\rightarrow A$ & 57.68 & 76.34 & \textbf{78.05} & 76.51 & 75.14 & 77.87\\
$C\rightarrow D$ & 68.24 & 68.61 & 74.19 & 70.93 & 70.08 & \textbf{75.80} \\
$C\rightarrow W$ & 65.28 & 75.01 & 74.85 & 76.39 & 74.16 & \textbf{80.49} \\
\hline
Average  & 59.43 & 67.10 & 67.76 & 67.35 & 67.36 & \textbf{68.77} \\
\hline
\end{tabular}
\end{table*}

\begin{table*}[!htp]
\centering
\caption{Accuracy in $\%$ with 10\% training target data for 20 Newsgroups data}\label{tab:20news}
\begin{tabular}{|c|c|c|c|c|c|c|}
\hline
Tasks & AdaBoost$_\mathcal{T}$    & AdaBoost$_{\mathcal{S} \& \mathcal{T}}$   & TradaBoost  & TransferBoost & GapBoost & InfoBoost\\
\hline
rec vs talk &  $90.30\pm 3.14$ & $87.42\pm 2.21$ & $71.97 \pm 5.89$ & $89.42 \pm 2.99$ & $91.77\pm 3.51$ & $\mathbf{93.11 \pm 1.71}$ \\
comp vs sci & $92.87 \pm 1.03$ & $84.66 \pm 1.53$ & $92.71 \pm 1.78$ & $93.35 \pm 1.05$ & $93.55 \pm 2.54$ & $\mathbf{94.15 \pm 1.04}$ \\
rec vs sci & $90.17 \pm 1.02$ & $86.68\pm 1.06$ & $90.49 \pm 1.36$ & $\mathbf{93.62 \pm 0.69}$ & $90.70 \pm 2.14$ & $92.58 \pm 0.89$\\
talk vs sci  & $88.66 \pm 2.74$ & $71.72 \pm 3.68$ & $81.39  \pm 4.98$ & $87.78 \pm 3.25$ & $90.14 \pm 2.55$ & $\mathbf{90.40 \pm 1.78}$\\
comp vs rec & $91.44 \pm 1.35$ & $87.90  \pm 1.80$ & $91.52\pm 1.12$ & $\mathbf{94.95 \pm 0.60}$ & $92.62 \pm 2.95$ & $94.00 \pm 1.58$\\
comp vs talk & $93.86 \pm 1.46$ & $89.57 \pm 1.32$ & $75.78 \pm 1.14$ & $94.83 \pm 0.78$ & $94.90 \pm 1.02$ &  $\mathbf{95.08 \pm 0.89}$ \\
\hline
Average & 91.22 & 84.66 & 83.98 & 92.33 & 92.28 & $\mathbf{93.22}$\\
\hline
\end{tabular}
\end{table*}

\begin{table*}[!htp]
\centering
\caption{Accuracy in $\%$ with 1\% training target data for MNIST and USPS datasets} \label{tab:mu}
\begin{tabular}{|c|c|c|c|c|c|c|}
\hline
Tasks & AdaBoost$_\mathcal{T}$    & AdaBoost$_{\mathcal{S} \& \mathcal{T}}$   & TradaBoost  & TransferBoost & GapBoost & InfoBoost\\
\hline
U to M$_{1 \textup{ vs } 7}$ & $61.52 \pm 4.15$  &  $57.73 \pm 1.92$ &  $52.67 \pm 1.08$ & $63.47 \pm 2.34$  & $63.58 \pm 1.99$ & $\mathbf{64.51 \pm 2.11}$ \\
U to M$_{2 \textup{ vs } 3}$ & $59.69 \pm 3.96$  & $56.49 \pm 1.87$  & $52.18 \pm 3.87$  & $63.29 \pm 1.71$  & $\mathbf{63.37 \pm 1.87}$  & $62.83 \pm 1.57$  \\
U to M$_{5 \textup{ vs } 6}$ & $61.05 \pm 2.87$  &  $57.22 \pm 1.24$ & $55.15 \pm 1.44$  & $61.54 \pm 1.74$ &  $61.46 \pm 1.70$ & $\mathbf{62.66 \pm 1.63}$  \\
U to M$_{0 \textup{ vs } 8}$ &  $59.78 \pm 3.75$  &  $57.61 \pm 2.57$ & $51.48 \pm 1.29$ &  $64.50 \pm 1.88$ & $64.77 \pm 1.88$ & $\mathbf{65.00 \pm 2.02}$  \\
U to M$_{4 \textup{ vs } 9}$ &  $57.97 \pm 3.12$  & $58.03 \pm 4.24$ &  $57.15 \pm 4.69$ & $61.87 \pm 1.64$  & $\mathbf{61.91 \pm 1.71}$  & $61.51 \pm 2.01$ \\
\hline
M to U$_{1 \textup{ vs } 7}$ &  $87.42 \pm 4.94$ & $90.44 \pm 12.54$  & $72.64 \pm 13.98$  & $92.23 \pm 3.66$   &  $91.26 \pm 3.50$ & $\mathbf{92.43 \pm 3.75}$ \\
M to U$_{2 \textup{ vs } 3}$ &  $74.38 \pm 5.33$ & $74.32 \pm 11.68$  & $76.80 \pm 7.34$ & $\mathbf{77.01} \pm 7.08$  &  $76.94 \pm 5.83$ & $76.29 \pm 6.34$ \\
M to U$_{5 \textup{ vs } 6}$ &  $71.87 \pm 7.14$ & $59.11 \pm 8.62$  & $69.53 \pm 3.79$ & $69.42 \pm 7.92$   & $75.33 \pm 5.83$ & $\mathbf{75.63 \pm 5.60}$ \\
M to U$_{0 \textup{ vs } 8}$ &  $74.47 \pm 3.85$ & $74.84 \pm 4.87$ &  $71.83 \pm 4.57$ & $77.89 \pm 4.08$ &  $\mathbf{78.43 \pm 4.22}$ & $77.96 \pm 4.43$  \\
M to U$_{4 \textup{ vs } 9}$ & $87.88 \pm 4.04$  &  $88.80 \pm 5.93$  & $79.24 \pm 1.43$  &  $91.07 \pm 4.03$  & $91.21 \pm 3.78$  & $\mathbf{91.33 \pm 3.94}$ \\
\hline
Average & 69.60  & 67.46 & 63.87 & 72.23 & 72.83 & \textbf{73.02} \\
\hline 
\end{tabular}
\end{table*}

\begin{figure*}[!htp]
	\centering
	\subfloat[rec vs talk]{
		\includegraphics[width=0.2\textwidth]{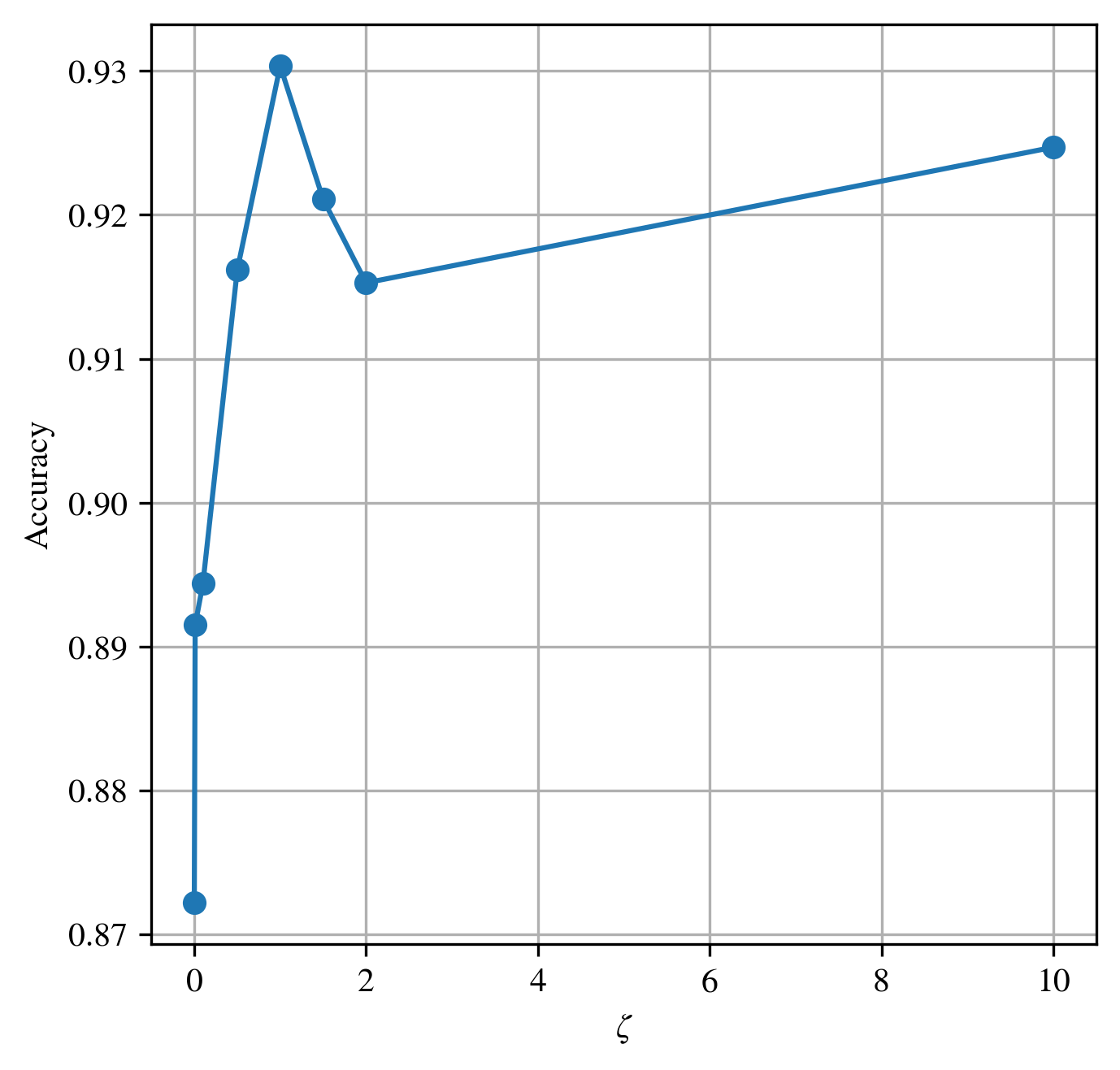}
	}
	\subfloat[comp vs sci]{
		\includegraphics[width=0.2\textwidth]{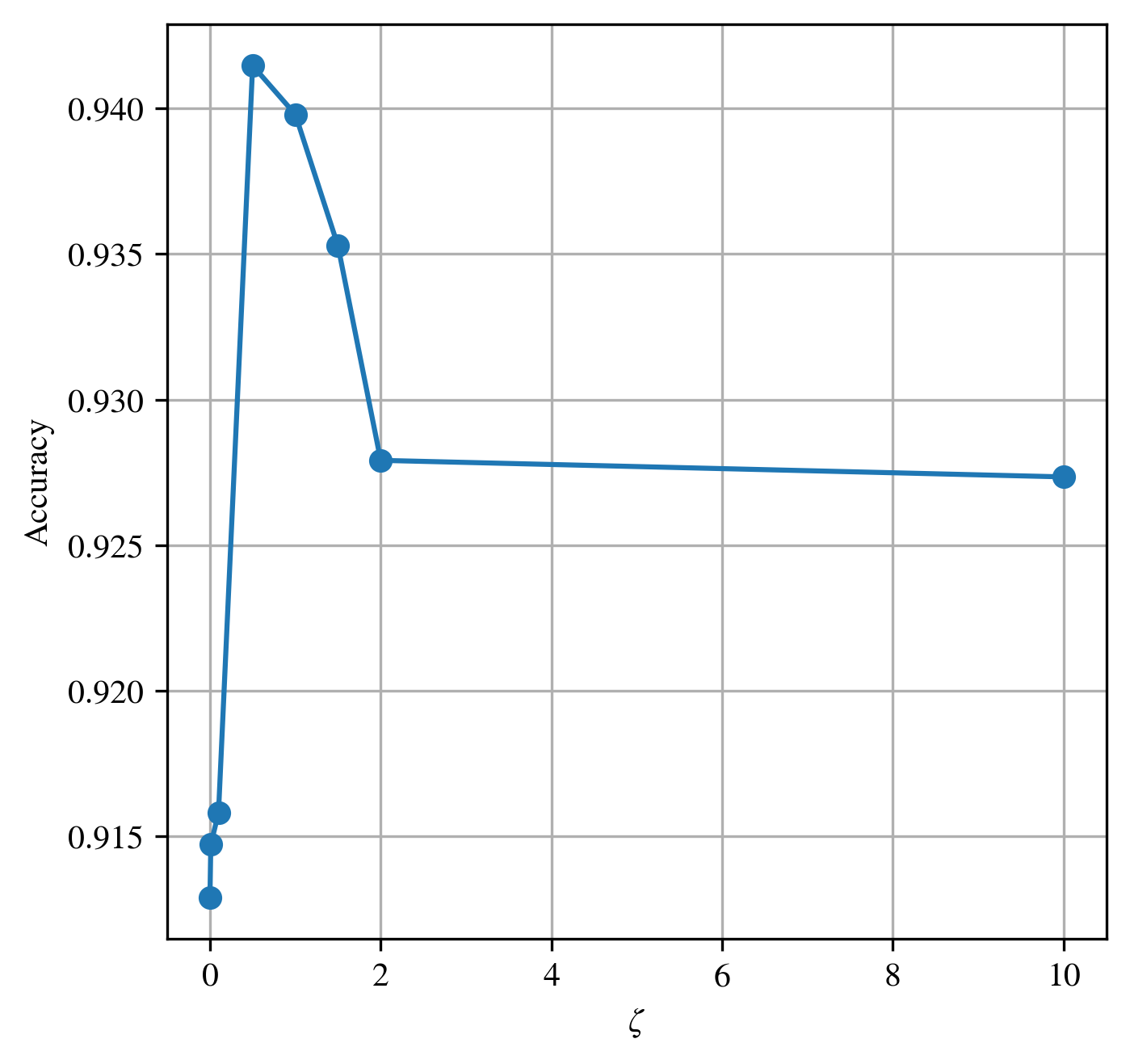}
	}
	\subfloat[rec vs sci]{
		\includegraphics[width=0.2\textwidth]{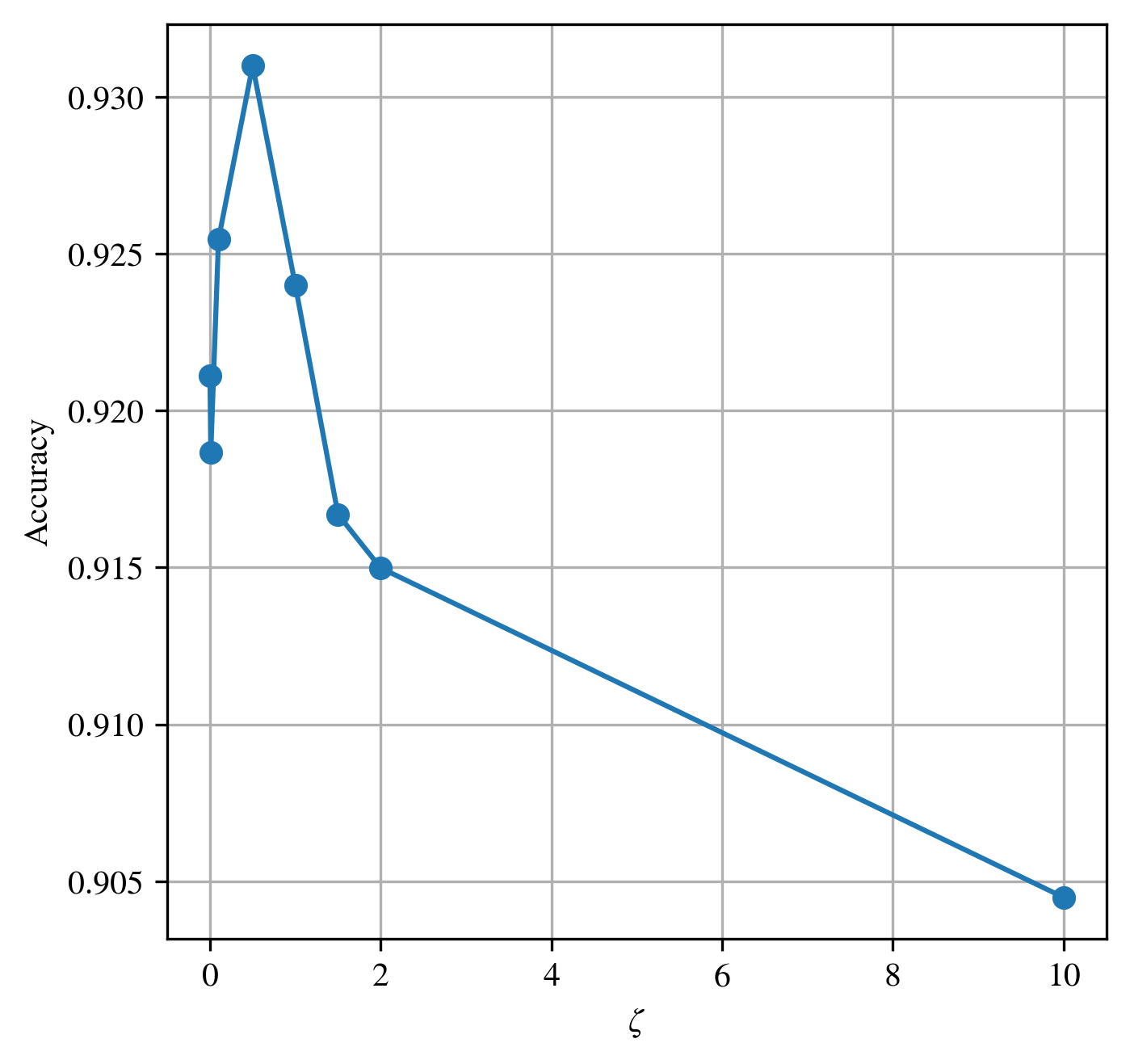}
	}\\
	\subfloat[talk vs sci]{
		\includegraphics[width=0.2\textwidth]{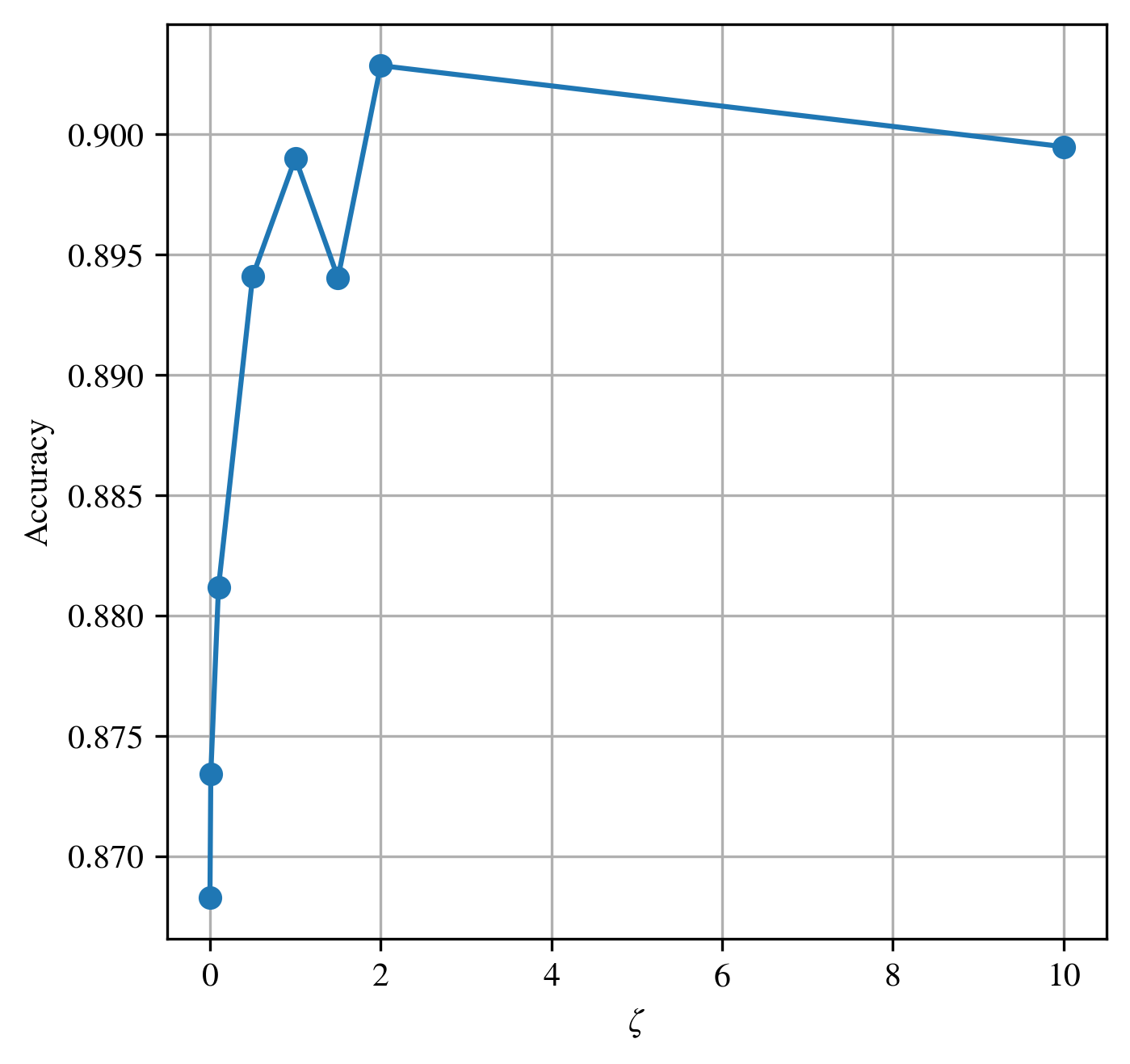}}
	\subfloat[comp vs rec]{
		\includegraphics[width=0.2\textwidth]{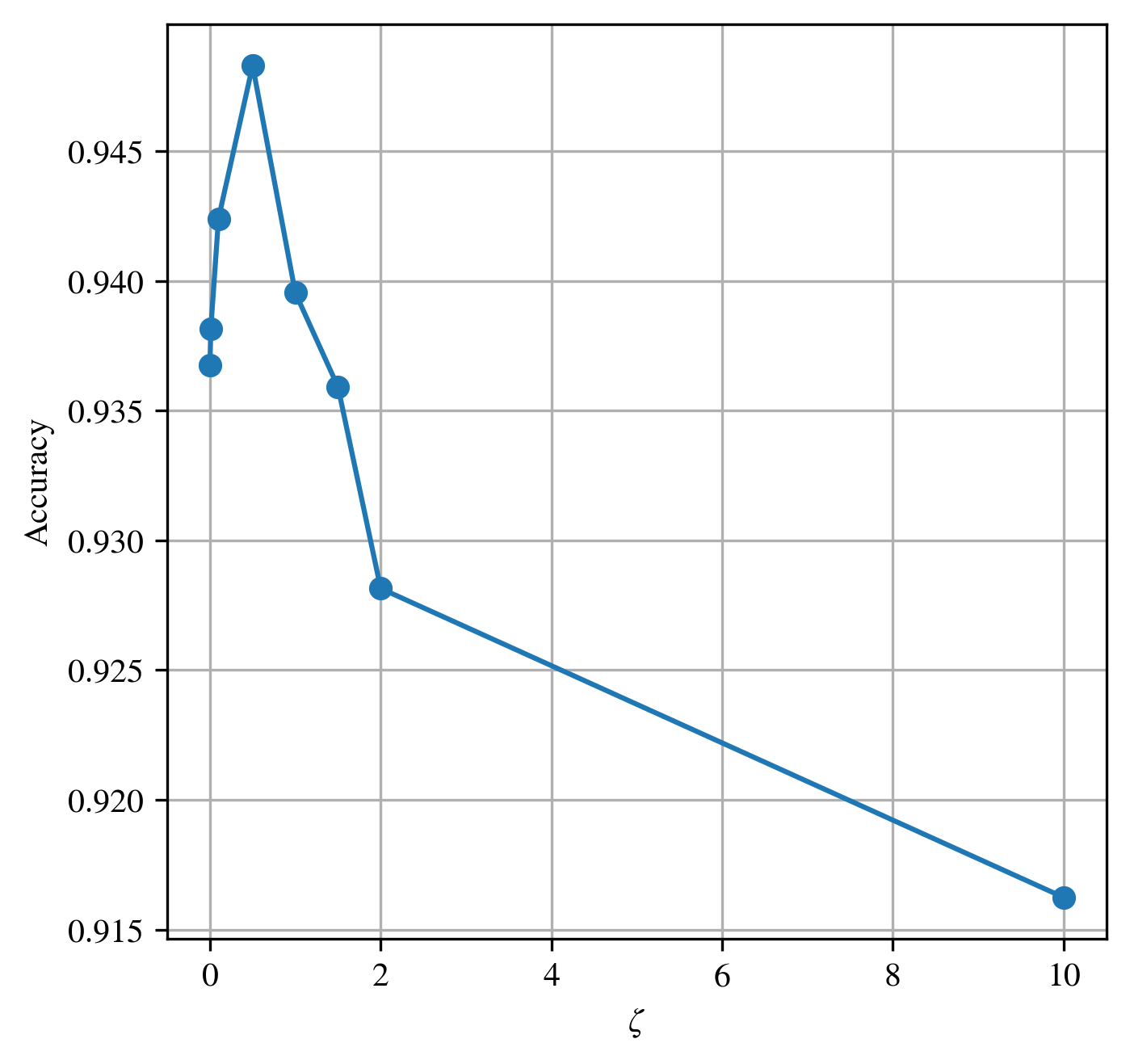}
	}
	\subfloat[comp vs talk]{
		\includegraphics[width=0.2\textwidth]{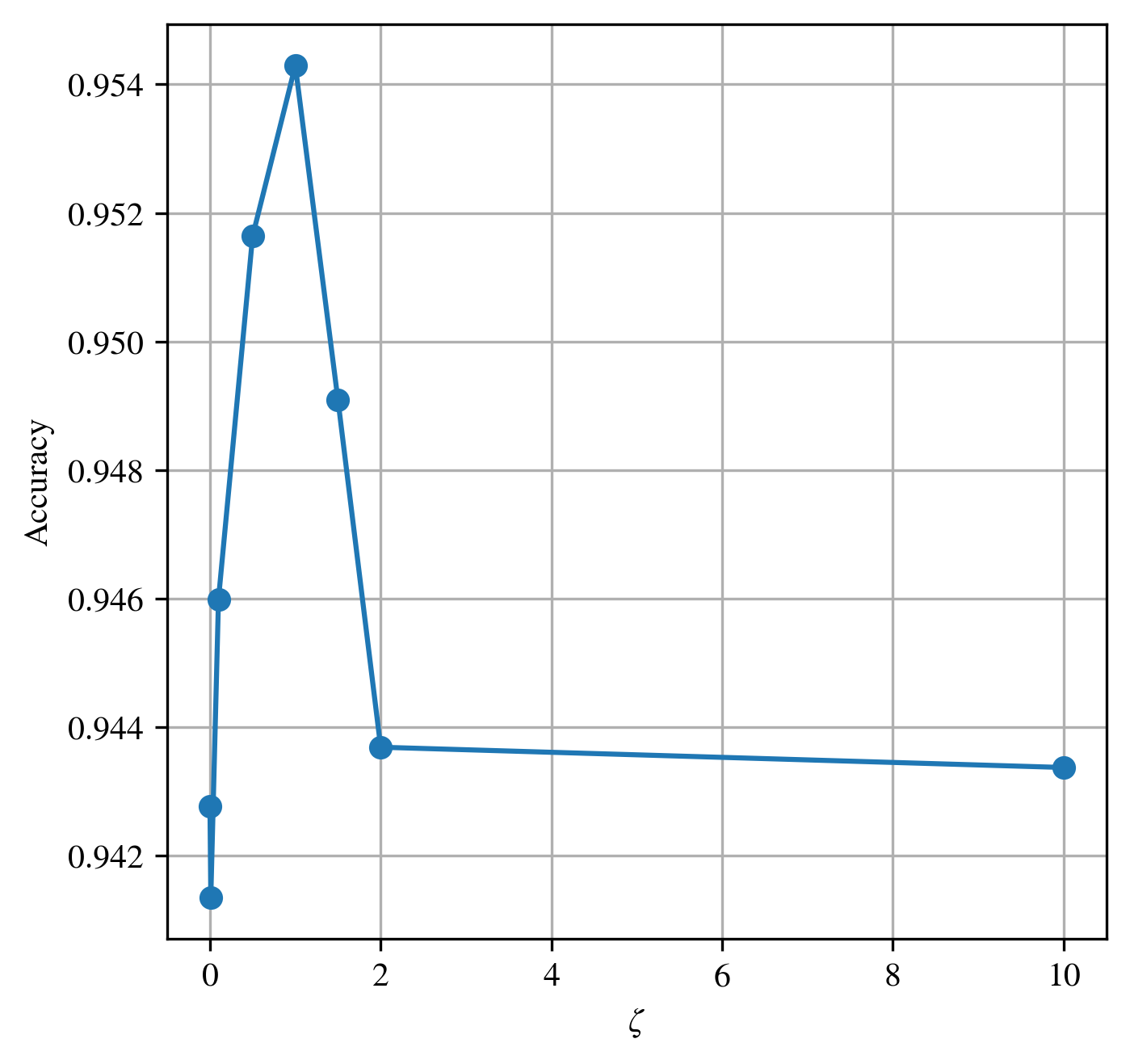}
	}
    \caption{Effect of $\zeta$ varying from 0 to 10 when fixing $\eta = 1$. The results are averaged over 20 experiments.}\label{fig:zeta}
\end{figure*}

\begin{figure*}[!htp]
	\centering
	\subfloat[rec vs talk]{
		\includegraphics[width=0.2\textwidth]{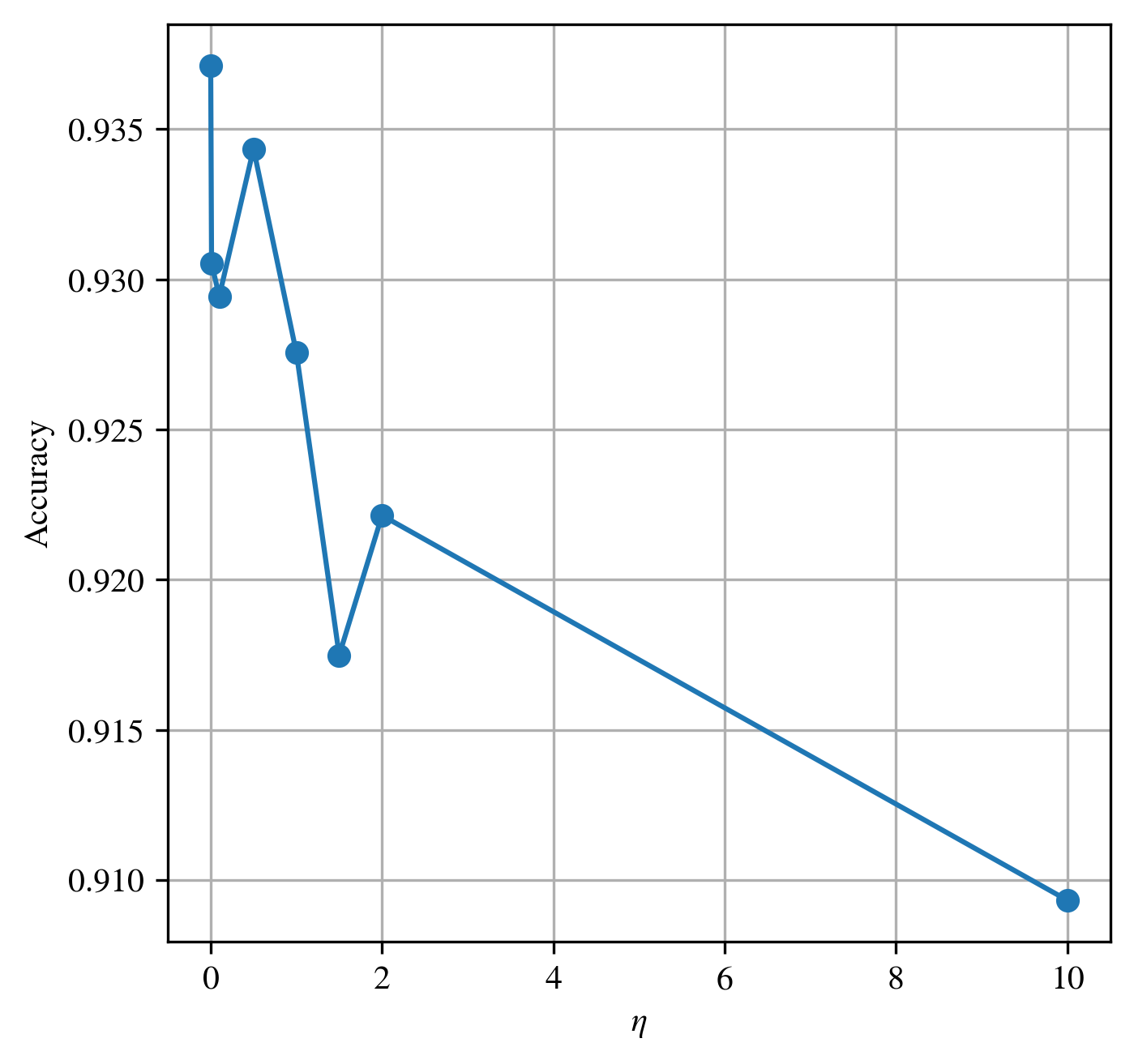}
	}
	\subfloat[comp vs sci]{
		\includegraphics[width=0.2\textwidth]{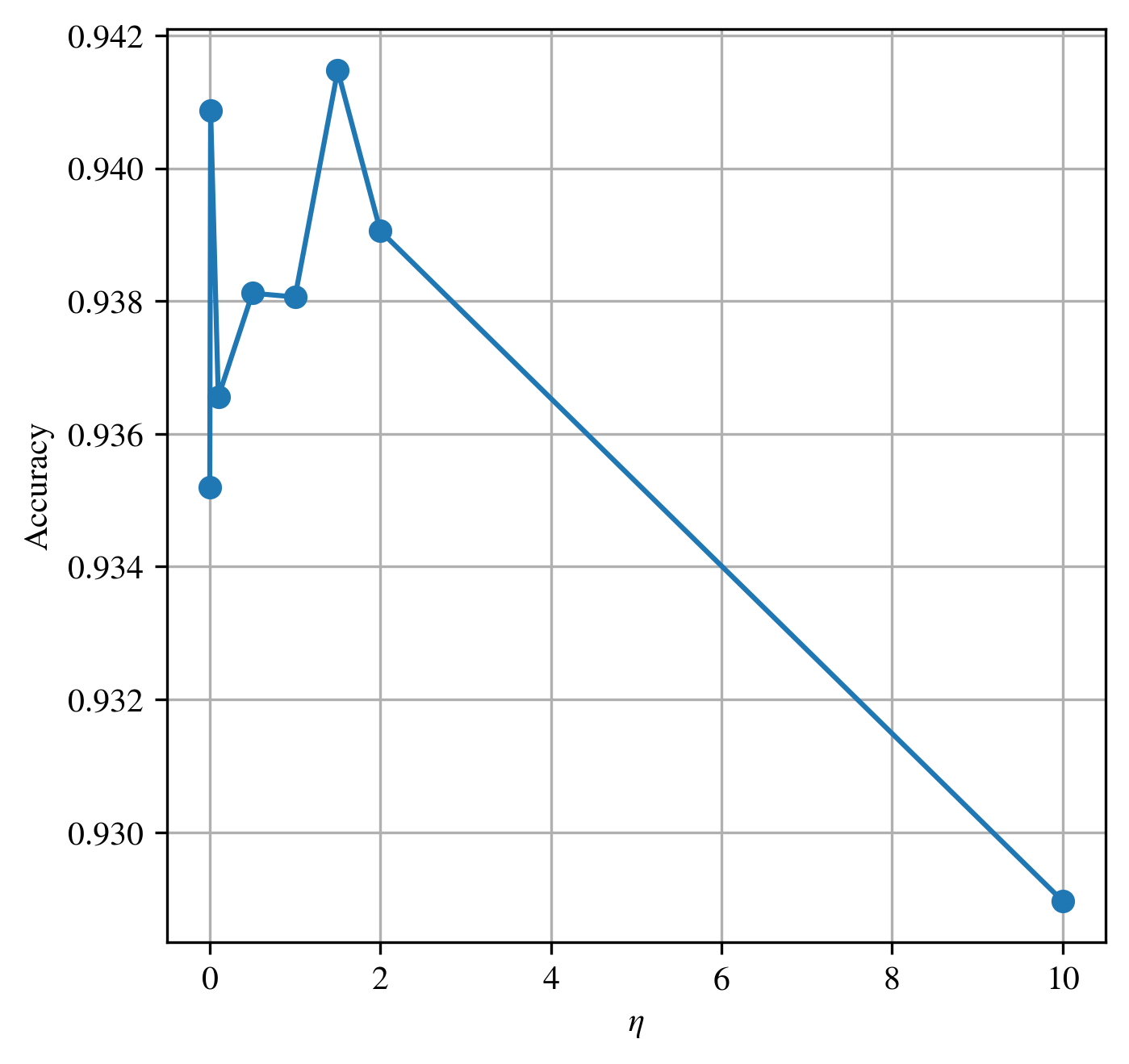}
	}
	\subfloat[rec vs sci]{
		\includegraphics[width=0.2\textwidth]{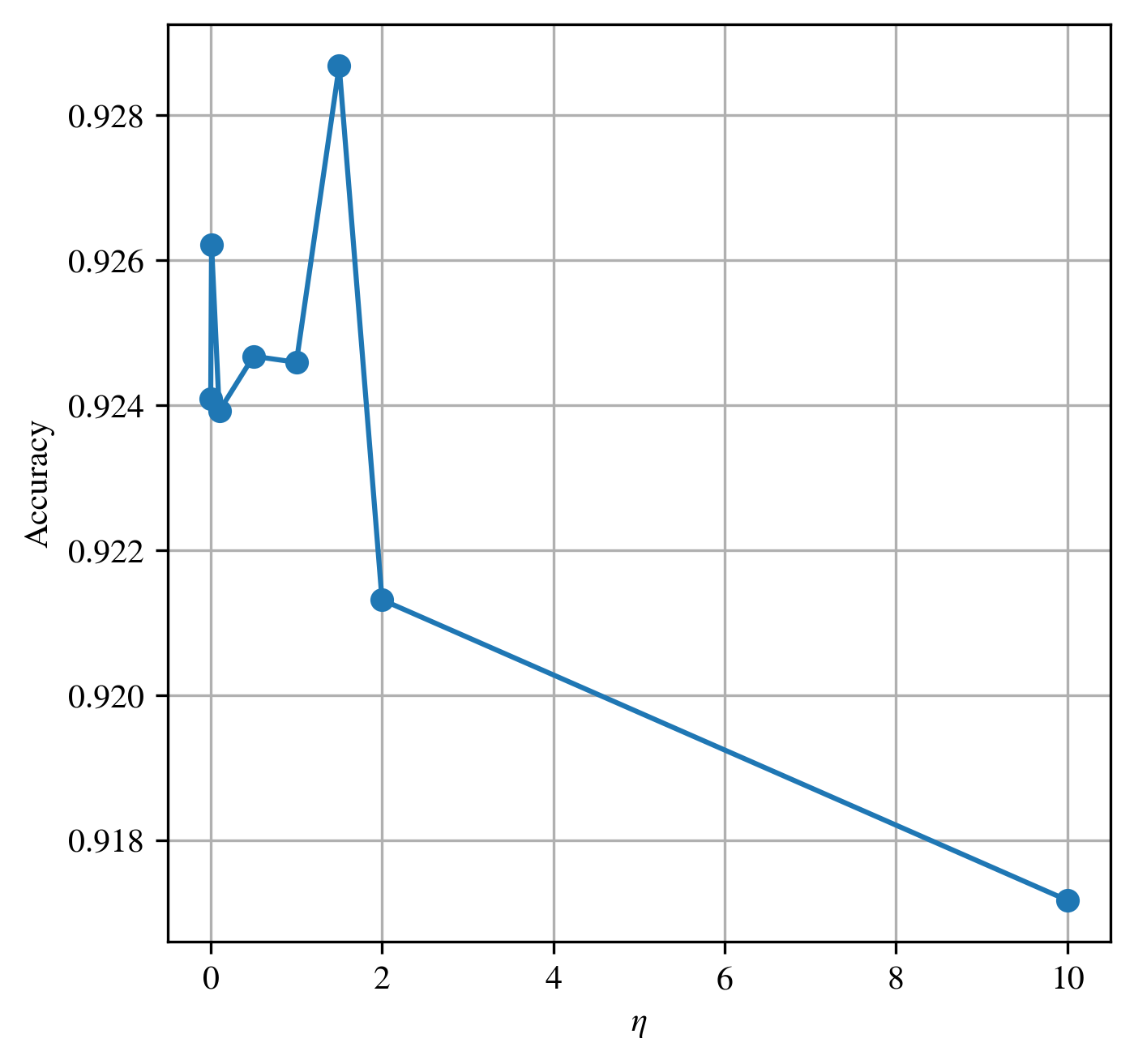}
	}\\
	\subfloat[talk vs sci]{
		\includegraphics[width=0.2\textwidth]{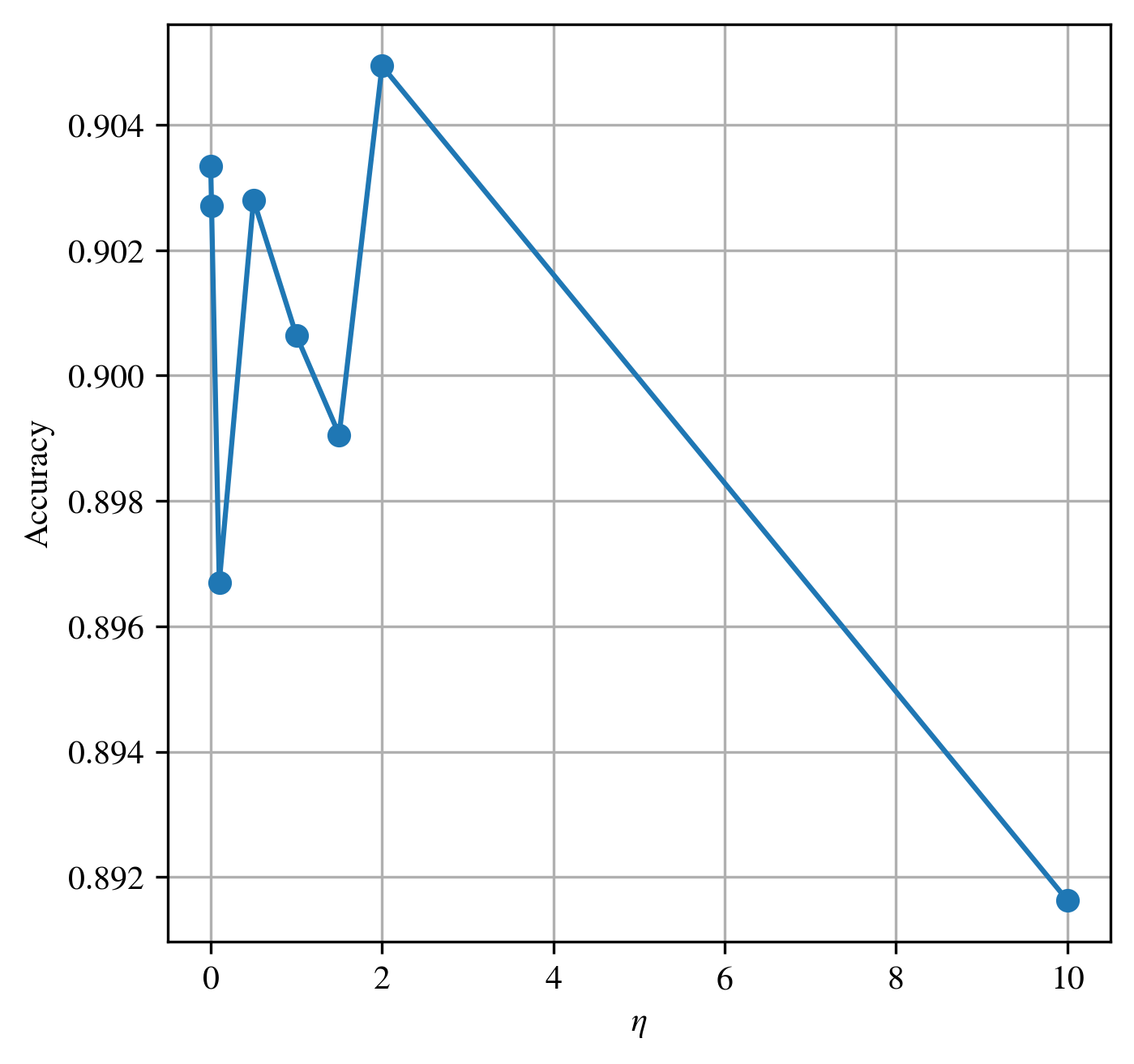}}
	\subfloat[comp vs rec]{
		\includegraphics[width=0.2\textwidth]{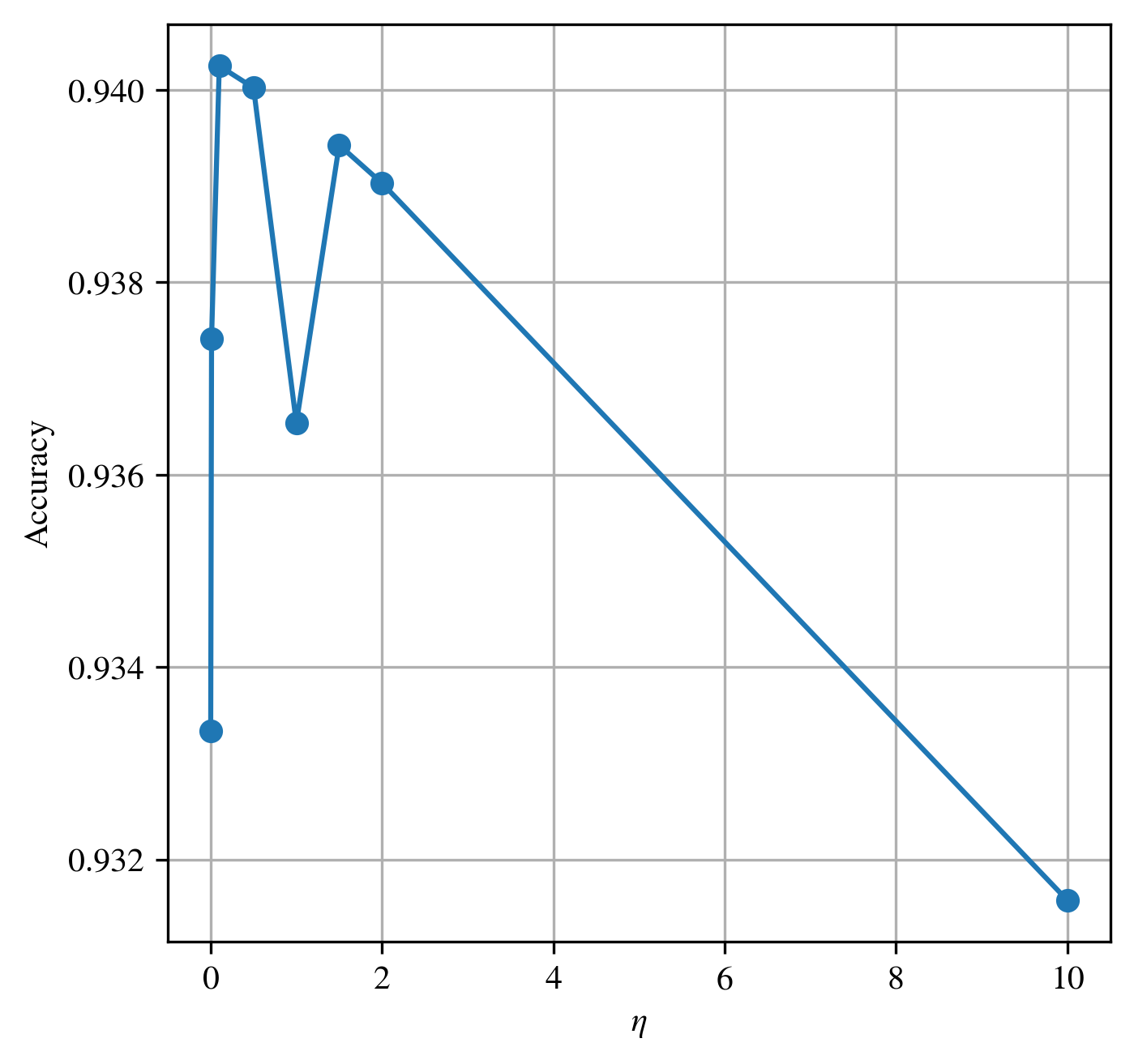}
	}
	\subfloat[comp vs talk]{
		\includegraphics[width=0.2\textwidth]{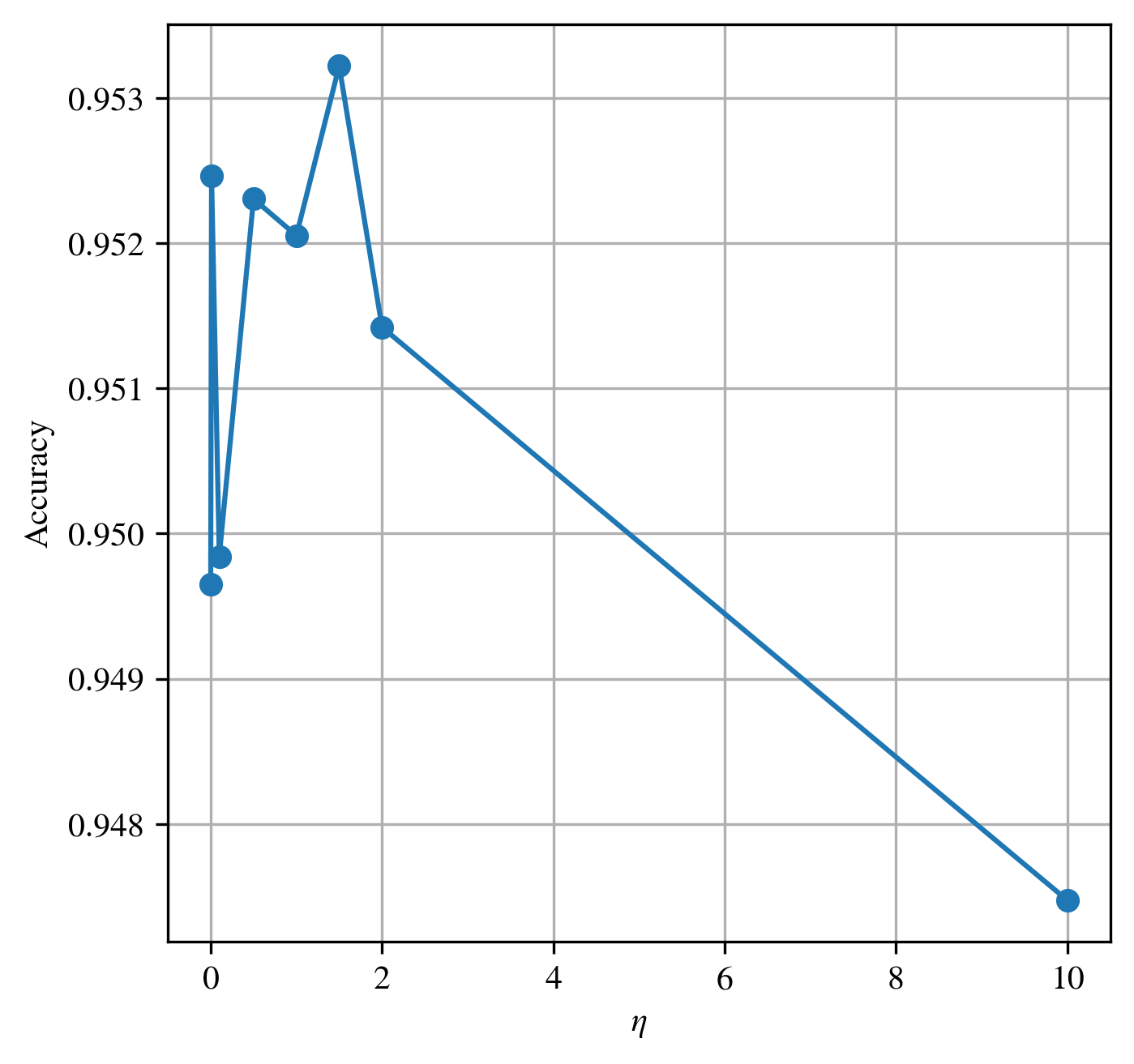}
	}
    \caption{Effect of $\eta$ varying from 0 to 10 when fixing $\zeta = 1$. The results are averaged over 20 experiments.}\label{fig:eta}
\end{figure*}

\begin{figure*}[!htp]
	\centering
	\subfloat{
		\includegraphics[width=0.75\textwidth]{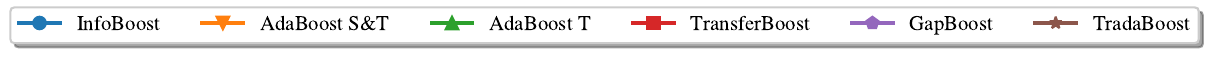}}	
	\addtocounter{subfigure}{-1}
	\\
	\subfloat[rec vs talk]{
		\includegraphics[width=0.25\textwidth]{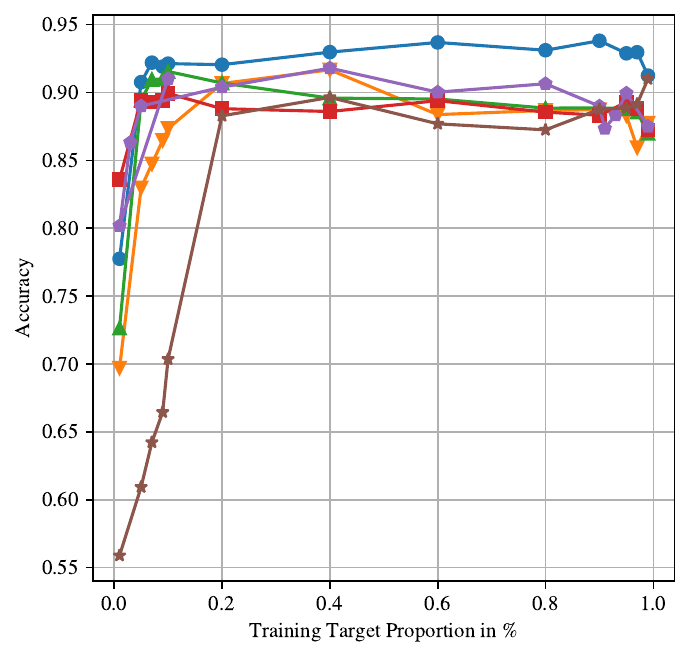}
	}
	\subfloat[comp vs sci]{
		\includegraphics[width=0.25\textwidth]{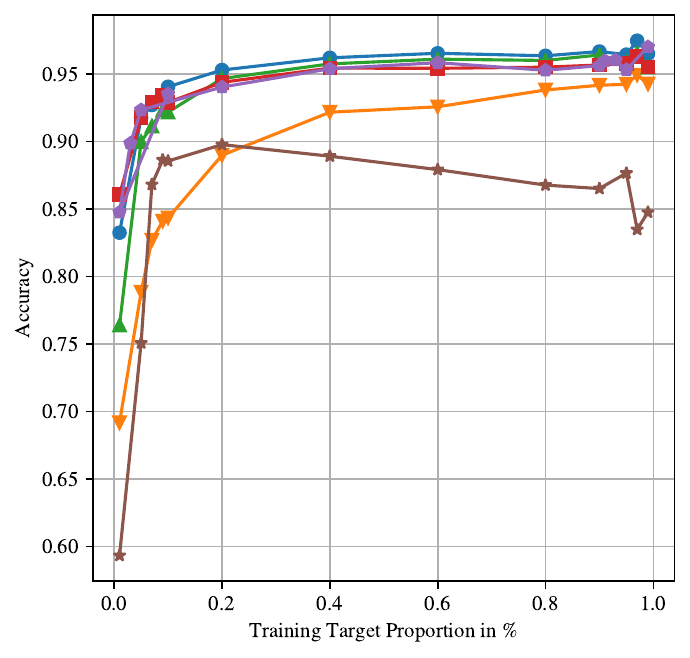}
	}
	\subfloat[rec vs sci]{
		\includegraphics[width=0.25\textwidth]{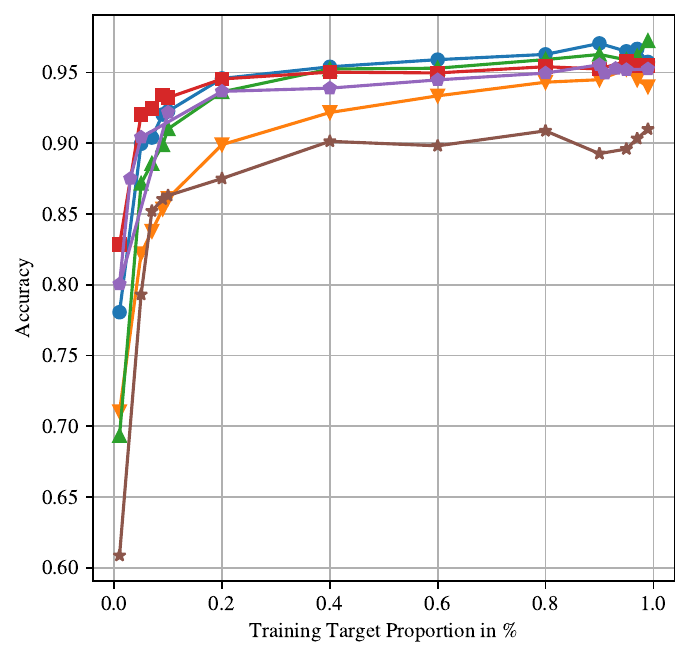}
	}\\
	\subfloat[talk vs sci]{
		\includegraphics[width=0.25\textwidth]{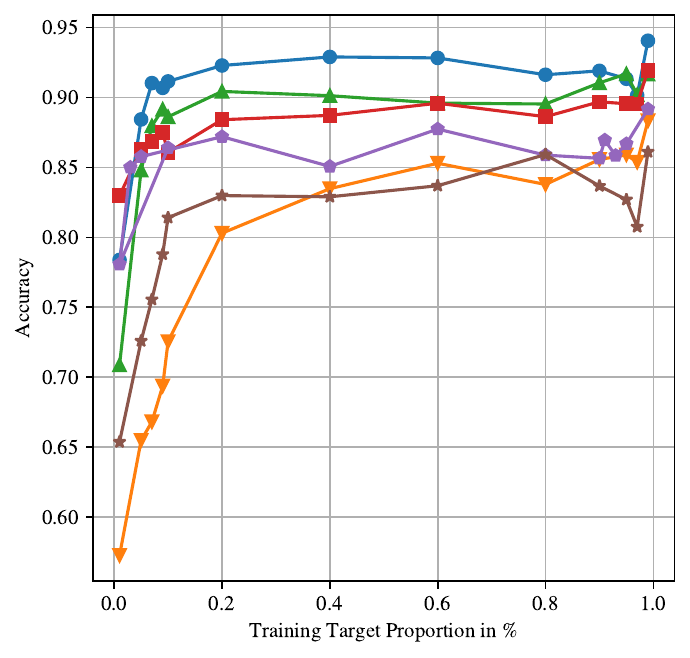}
		}
	\subfloat[comp vs rec]{
		\includegraphics[width=0.25\textwidth]{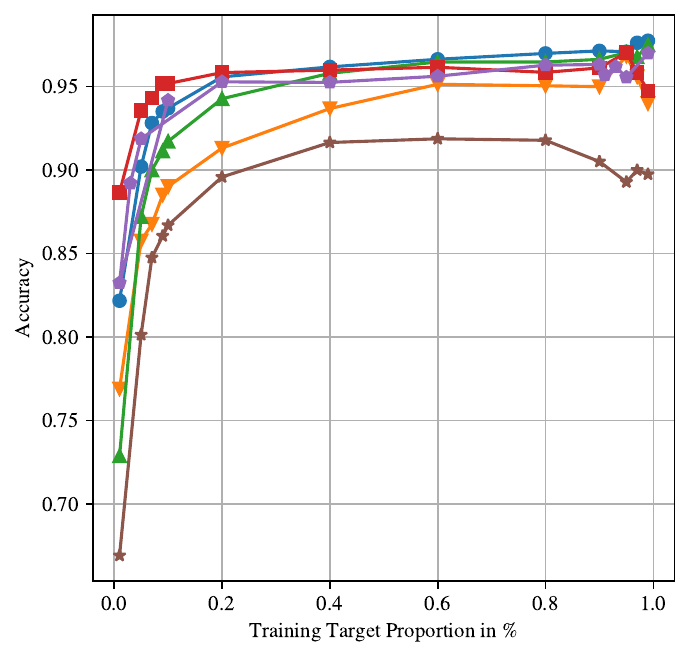}
	}
	\subfloat[comp vs talk]{
		\includegraphics[width=0.25\textwidth]{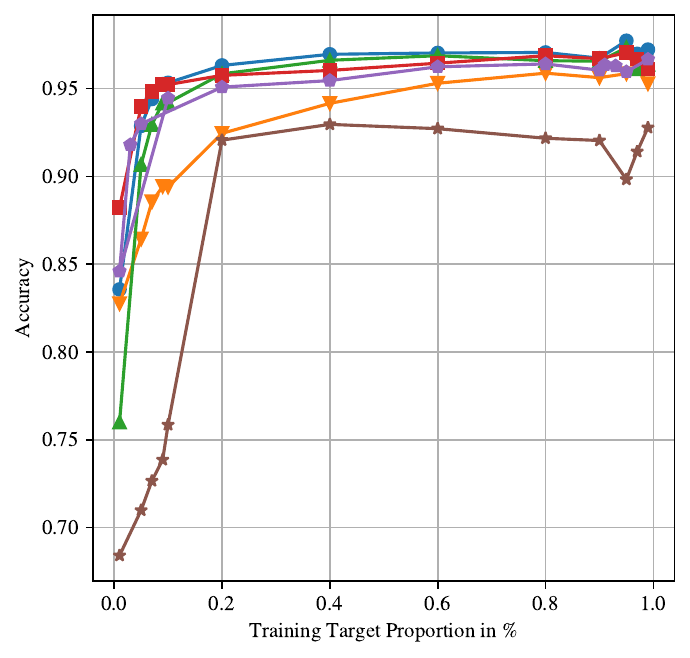}
	}
    \caption{Accuracy with different proportions of target training sizes for different tasks. We utilize all the source data and the results are averaged over 20 experiments.}\label{fig:proportion}
\end{figure*}

\noindent\textbf{Performance Comparisons} In all comparisons, $\eta$ and $\zeta$ are both set to 1 in the InfoBoost algorithm. We used the entire source data set for different tasks. As each dataset has a different target sample size, we selected 10 target instances in the training phase for the Office-Caltech dataset, 10\% of the target data for the 20 Newsgroup dataset for training, and 1\% target instances for training on the handwritten digit datasets. We used generalized linear regression (e.g., logistic regression) for various boosting algorithms as our base classifier. All performance comparisons are listed on Table~\ref{tab:acwd}, \ref{tab:20news}, and \ref{tab:mu}. From the comparisons, we can see that the InfoBoost algorithm outperforms other competitors in most cases for all three different transfer learning problems, showing the benefits of taking the domain divergence and the mutual information into account.  

\noindent\textbf{Hyperparameter Sensitivity} We carefully examine the effects of $\zeta$ and $\eta$ by fixing one variable and varying another. The results are shown in Figure~\ref{fig:zeta} and \ref{fig:eta}, respectively. One the one hand, we can see from Figure~\ref{fig:zeta}, $\zeta = 1$ achieves higher accuracy compared to other choices. This is understandable because in the case of too small $\zeta$, the domain divergence is not well taken into account in the source data updating step and the algorithm may be over-fitted. If $\zeta$ is too large, then the domain divergence will overwhelm the effect of the empirical risk, which will easily lead to under-fitting. Therefore, a proper selection of $\zeta$ is needed to yield satisfactory performance. On the other hand, if we fix $\zeta = 1$ and vary $\eta$ from 0 to 10, from Figure~\ref{fig:eta} we can then see that a reasonable choice of $\eta$ (e.g., less and equal than 2) can lead to a good performance in most cases. Even with a very small $\eta$, the accuracy is still comparable to the best result.  But the performance will be degraded if we choose a large $\eta$ (e.g., $\eta = 10$); this is because after several iterations, we mainly select the samples that have less effect on the prediction and a large $\eta$ overwhelms both empirical risk and domain divergence, which may lead to under-fitting in a similar manner. Although the accuracy rate dropped, the drop was not as great as the change of $\zeta$. Therefore, the effect of $\eta$ on the model performance is less than that of $\zeta$. To achieve the optimal choice of the hyper-parameters, we recommend setting both $\zeta$ and $\eta$ to 1 around.

We also examine the effect of the size of the target sample. Specifically, we split the target dataset into the training and testing sets by varying the proportion of the training size from 0.01 to 0.99 for 20 Newsgroup dataset as an example, and the results are shown in Figure~\ref{fig:proportion}. From the comparisons, when the target training sample size is sufficient, the InfoBoost algorithm will have the best performance compared to other competitors. However, when there is a lack of a target training sample, the InfoBoost algorithm might not have enough data to precisely capture the domain divergence and the mutual information, and the performance will be degraded with a small target training sample size.  

\noindent\textbf{Further Discussion}
Different from the works \cite{eaton2011selective} and \cite{wang2019transfer} that utilize the domain performance gap, we propose a new boosting algorithm inspired by the mutual information and the domain divergence for transfer learning. We empirically verify the effectiveness of our proposed algorithm and particularly examine the sensitivity of the hyper-parameters. Since the information-theoretic quantities are usually hard to estimate, we use the surrogate quantities (e.g., $d_i(h_t, h^{-i}_t)$, $\textup{div}(x_i,y_i)$ that heuristically represents the mutual information and the domain divergence. Instead, one could also improve the algorithm with other related quantities by information conditioning and processing techniques, such as the conditional mutual information \cite{steinke2020reasoning} and $f$-information proposed in \cite{harutyunyan2021information}, which might be easier to estimate from the data, which is left as our future work.

%% file: 7.conclusion.tex
\section{Concluding remarks} \label{sec:conclusion}
In this work, we developed a set of upper bounds on the generalization error and the excess risk for general transfer learning algorithms under an information-theoretic framework. The derived bounds are particularly useful for various widely used algorithms in machine learning such as ERM, regularized ERM, Gibbs algorithm, and stochastic gradient descent algorithms. Moreover, we extend the results with different divergences other than the KL divergence, such as the $\phi$-divergence and the Wasserstein distance, which can give a tighter bound and handle more general scenarios where the KL divergence may be vacuous. We also tighten the bounds for the ERM and regularized ERM algorithms, and give the fast rate characterization of the transfer learning problems. Finally, we propose the InfoBoost algorithm that in each re-weighting iteration, the quantities of the mutual information bounds are utilized and the empirical verification shows the effectiveness of our algorithm. However, in some cases, it is hard to estimate the information-theoretic bound in practice. In particular, the density estimation from the source and target data might be inaccurate only given a few data samples in a high dimensional data space. One can thus relax the condition by either assuming the source and target distributions are lying in a restricted space and the divergence term will not be too deviated even with a small sample size or by finding a completely different divergence that is easier to estimate from the data, which is left as our future work.

In this paper we only studied the upper bound of transfer learning, ensuring the performance for some commonly used algorithms. However, since it is an upper bound, we are not sure when the introduction of source data hurts or improves the performance on the target domain. It is often known as the ``negative transfer" if the source data hurts. Therefore, under the information-theoretic framework, another possible future direction is to develop the lower bound for transfer learning algorithms to rigorously determine whether the source data will be useful or not, which might help mitigate the effect of the negative transfer.

\section*{Acknowledgement}
This research is supported by Melbourne Research Scholarships (MRS), and Australian Defence Science and Technology Group (DSTG) under the scheme The Artificial Intelligence for Decision Making Initiative 2022 and in part by Australian Research Council under project DE210101497.

%% file: 8.appendix.tex
\appendices

\section{Proof of Theorem \ref{thm:exp_gen}} \label{apd:main_theorem}
\begin{proof}
First, we rewrite the expectation of the generalization error for certain algorithm $P_{W|SS'}$ as in (\ref{eq:proof_theorem_1}) where the joint distribution $P_{WSS'}(w,s, s')$ on $(W, S, S')$ is given by $P_S(s)P_{S'}(s')P_{W|SS'}(w|s, s')$. 
\begin{figure*}[!htb]
\normalsize
\begin{align}
&\Esub{WSS'}{L_{\mu'}(W)-\hat L_{\alpha}(W)}=\Esub{WSS'}{L_{\mu'}(W)-(1-\alpha)\hat L(W,S)-\alpha\hat L(W,S')} \nonumber \\
&=\Esub{WSS'}{(1-\alpha)L_{\mu'}(W)+\alpha L_{\mu'}(W)-\frac{1}{n}\sum_{i=\beta n+1}^{n}\frac{1-\alpha}{1-\beta} \ell(W,Z_i)-\frac{1}{n}\sum_{i=1}^{\beta n}\frac{\alpha}{\beta} \ell(W,Z'_i)} \nonumber \\
&=\frac{1}{n}\Esub{WSS'}{\sum_{i=1}^{\beta n} \frac{\alpha}{\beta}(L_{\mu'}(W)- \ell(W,Z'_i))+\sum_{i=\beta n+1}^{n}\frac{1-\alpha}{1-\beta}(L_{\mu'}(W)-\ell(W, Z_i)) } \nonumber\\
&=\frac{1}{n}\frac{\alpha}{\beta}\sum_{i=1}^{\beta n}\Esub{WZ_i}{(L_{\mu'}(W)-\ell (W, Z'_i))}+\frac{1}{n}\frac{1-\alpha}{1-\beta}\sum_{i=\beta n+1}^n\Esub{WZ_i}{L_{\mu'}(W)-\ell(W, Z_i))}, \label{eq:proof_theorem_1}
\end{align} 
\vspace*{4pt}
\hrulefill
\end{figure*}
Recall that the variational representation of the KL divergence between two distributions $P$ and $Q$ defined over $\mathcal X$ is given as: (see, e.g. \cite{boucheron_concentration_2013})
\begin{align}
D(P||Q)=\sup_{f}\left\{\Esub{P}{f(X)}-\log\Esub{Q}{e^{f(x)}} \right\}, \label{eq:variational}
\end{align}
where the supremum is taken over all measurable functions such that $\Esub{Q}{e^{f(x)}}$ exists. For each $i=1,\ldots, n$, we define the joint distribution $P_{WZ_i}(w, z_i)$ ($P_{WZ'_i}(w, z'_i)$) between an individual sample $Z_i$ ($Z'_i$) and the hypothesis $W$ as induced by $P_{WSS'}(w, z^n)$ by marginalizing all samples other than $z_i$, and let $P_W$ be the marginal distribution on $W$ induced from $P_{WSS'}$.  

We first show the first inequality in the Theorem. For any $i=1,\ldots, \beta n$, let $P=P_{WZ'_i}$, $Q=P_W\otimes\mu'$ in (\ref{eq:variational}),  
and define $f:=\lambda \ell(W,Z'_i)$ for some  $\lambda$. The representation in (\ref{eq:variational}) implies that 
\begin{align*}
\Esub{WZ'_i}{\lambda\ell (W,Z'_i)}\leq D(P_{WZ'_i}||P_W\otimes \mu')+\log \E{e^{\lambda \ell(W,Z'_i)}},
\end{align*}
where the expectation on the R.H.S. is taken w.r.t. the distribution $P_W\otimes\mu'$. By the assumption that
\begin{align*}
\log \E{e^{\lambda(\ell(W,Z'_i)-\E{\ell(W,Z'_i)})}}\leq \psi(\lambda)
\end{align*}
for some $\lambda\in[b_{-},0]$ under the distribution $P_W\otimes \mu'$, we have
\begin{align*}
& \Esub{WZ'_i}{\lambda(\ell(W,Z'_i)-\Esub{WZ'_i\sim P_W\otimes \mu'}{\ell(W,Z'_i)} )} \\ 
& \leq  D(P_{WZ'_i}||P_W\otimes \mu') +\psi(\lambda),
\end{align*}
which is equivalent to
\begin{align*}
&\Esub{WZ'_i}{L_{\mu'}(W)-\ell(W,Z'_i)} \\ &\leq -\frac{1}{\lambda}\left(D(P_{WZ'_i}||P_W\otimes \mu') +\psi(\lambda))\right)\\
&=-\frac{1}{\lambda}\left(I(W;Z'_i)+D(P_{Z'_i}||\mu') +\psi(\lambda))\right)\\
&=-\frac{1}{\lambda}\left(I(W;Z'_i) +\psi(\lambda))\right),
\end{align*}
as $P_{Z'_i}=\mu'$ for $i=1,\ldots, \beta n$.  The best upper bound is obtained by minimizing the R.H.S., giving
\begin{align}
&\Esub{WZ'_i}{L_{\mu'}(W)-\ell (W,Z'_i)}\nonumber \\
&\leq \min_{\lambda\in[0, -b_{-}]}\frac{1}{\lambda}(I(W;Z'_i)+\psi(-\lambda))= \psi^{*-1}(I(W;Z'_i)). \label{eq:upper_target}
\end{align}
For $i=\beta n+1,\ldots, n$ in the source domain, using the same argument we can show that
\begin{align}
\Esub{WZ_i}{L_{\mu'}(W)-\ell (W,Z_i)}\leq \psi^{*-1}(I(W;Z_i)+D(\mu||\mu'))
\label{eq:upper_source}
\end{align}
Summing over $i$ using the upper bounds in (\ref{eq:upper_target}) and (\ref{eq:upper_source}), we obtain the first inequality in the theorem.

The second inequality is shown in the same way by using the fact that the cumulant generating function is upper bounded by $\psi(\lambda)$ in $[0,b_{+}) $.
\end{proof}

\section{Proof of Corollary \ref{coro:gen_beta0}}
\label{proof:corollary_beta0}
\begin{proof}
In the case when $\beta=0$, for any hypothesis $W$ induced by the sample $S$ and the learning algorithm $P_{W|S}$,  the generalization error is
\begin{align*}
&\Esub{WS}{L_{\mu'}(W)-\hat L(W,S)}\\
&=\Esub{WS}{L_{\mu'}(W)-\frac{1}{n}\sum_{i=1}^n\ell(W,Z_i)}\\
&=\frac{1}{n}\sum_{i=1}^n \Esub{WZ_i}{L_{\mu'}(W)-\ell(W,Z_i)}.
\end{align*}
From here, we can use the same argument as in the proof of Theorem \ref{thm:exp_gen} to upper and lower bound the term $\Esub{WZ_i}{L_{\mu'}(W)-\ell(W,Z_i)}$ to arrive at the result. Notice that here, we do not assume that $W$ is the solution of the ERM algorithm.
\end{proof}

\section{Multi-Source transfer learning problem} \label{apd:MSTLP}
In this subsection, we consider the learning regime that we have more than one source domain. Specifically, assume that the target domain consists of $n_t$ samples drawn IID from $\mu'$, e.g., $S' = \{Z'_1,Z'_2,\cdots,Z'_{n_t}\}$. We also have $k$ source domains $S^k_{1} = \left(S_{1},S_{2},\cdots,S_{k}\right)$ and each source domain $S_i$ consists of $n_i$ samples IID drawn from $\mu_i$, e.g., $S_i =\{Z_{i,1}, Z_{i,2}, \cdots, Z_{i,n_i}\}$, for $i = 1,\cdots, k$. Then we assign the weight $\alpha_i$ to the source domain $i$ and the weight $\alpha_t$ for the target domain where $\sum_{i=1}^{k} \alpha_i + \alpha_t = 1$. We then define the empirical loss by
\begin{align}
    \hat{L}_{\alpha}(w, S^k_1, S') = \sum_{i=1}^{k} \alpha_i \hat{L}(w,S_i) + \alpha_t\hat{L}(w,S').
\end{align}
and the corresponding combined expected risk is defined by
\begin{align}
    L_{\alpha}(w) = \sum_{i=1}^{k} \alpha_i {L}_{\mu_i}(w,S_i) + \alpha_t L_{\mu'}(w).
\end{align}
Then the generalization error for the multi-source transfer learning is defined as 
\begin{align}
    \gen (w, S^k_{1}, S') = L_{\mu'}(w) - \hat{L}_{\alpha}(w, S^k_1, S').
\end{align}
We can easily extend the result of the generalization error from the single source domain to multiple source domains.
\begin{corollary}[Generalization error for multi-source transfer learning]\label{coro:mst}
Let $P_W$ be the marginal distribution induced by $S^{k}_{1}, S'$ and $P_{W|S^k_1S'}$ for some algorithm (not necessarily the ERM solution). If $\ell(W, Z)$ is $r^2$-subgaussian under the distribution $P_W\otimes \mu'$, then the expectation of the generalization error is upper bounded as
\begin{align}
&|\Esub{WS^k_1S'}{\gen(W, S^k_1, S')}|\leq  \frac{\alpha_t\sqrt{2r^2}}{ n_t}\sum_{i=1}^{n_t}\sqrt{I(W;Z'_i)} \nonumber \\
& + \sum_{i=1}^{k}\frac{(1-\alpha_i)\sqrt{2r^2}}{n_i}\sum_{j=1 }^{n_i}\sqrt{ (I(W;Z_{i,j})+D(\mu_i||\mu'))}.
\end{align}
\end{corollary}
\noindent For ERM solution, we can decompose the excess risk in the following way:
\begin{figure*}[!htb]
\normalsize
\begin{align}
L_{\mu'}(w_{\ERM})-L_{\mu'}(w^*)&=L_{\mu'}(w_{\ERM})-\hat L_{\alpha}(w_{\ERM},S^k_1, S')+\hat L_{\alpha}(w_{\ERM},S^k_1, S')-\hat L_{\alpha}(w^*,S^k_1, S')\nonumber\\
&\quad  +\hat L_{\alpha}(w^*,S^k_1, S')-L_{\alpha}(w^*)
+L_{\alpha}(w^*)-L_{\mu'}(w^*)\nonumber\\
&\leq \gen(w_{\ERM}, S, S') +\hat L_{\alpha}(w^*,S^k_1, S')- L_{\alpha}(w^*) + \sum_{i=1}^{k}\alpha_i(L_{\mu}(w^*)-L_{\mu'}(w^*)). 
\end{align}
\vspace*{4pt}
\hrulefill
\end{figure*}

Therefore, we can give the upper bounds on the excess risk.
\begin{corollary}[Excess risk for multi-source transfer learning]
Let $P_W$ be the marginal distribution induced by $S^k_{1},S'$ and $P_{W|S^k_1S'}$ for the ERM algorithm, assume the loss function $\ell(W, Z)$ is $r^2$-subgaussian under the distribution $P_{W} \otimes \mu'$. Then the following inequality holds. 
\begin{align}
& \mathbb{E}_{W}[R_{\mu'}(W_\ERM)]\leq \frac{\alpha_t\sqrt{2r^2}}{ n_t}\sum_{i=1}^{n_t}\sqrt{I(W_\ERM;Z'_i)} \nonumber  \\
&+ \sum_{i=1}^{k}\frac{\alpha_i\sqrt{2r^2}}{n_i}\sum_{j=1 }^{n_i}\sqrt{ (I(W_\ERM;Z_{i,j})+D(\mu_i||\mu'))} \nonumber \\
&+ \sum_{i=1}^{k}\alpha_i d_{\mathcal W}(\mu_i, \mu').
\end{align}
\end{corollary}
\begin{remark}
To minimize the upper bound, we will be optimizing $\alpha_t$ and $\alpha_i$ for target and source domains, which is not trivial as these weights are also implicitly embedded in the mutual information $I(W_\ERM;Z_i)$. Intuitively, less weights need to be assigned to those source domains which have large KL divergence $D(\mu_i\|\mu')$. To this end, one can apply the InfoBoost algorithm for optimizing $\alpha$.
\end{remark}



\section{Transfer learning with pre-trained hypothesis}\label{apd:pre-train}
In transfer learning, it is quite common to first pre-train on the source domain and subsequently fine-tune on the target dataset. This methodology has been extensively adopted and affirmed across various situations, especially in the setups of deep learning \cite{liu2021transtailor,you2021logme}. Using our method, we can further apply this bounding technique to the pre-trained hypothesis. 

First, let us introduce some new notations for this particular setup: we denote the target dataset by $S' = \left\{Z'_1, Z'_2, \cdots, Z_{n_t}'\right\}$ and the source dataset by $S = \left\{Z_{1}, Z_{2}, \cdots, Z_{n_s}\right\}$. For deep learning model pre-training, it is common for the source and target to have some shared layers, represented by $w_C$. Initially, we train a model $w_{\text{source}} = (w_C, w_S)$ using the source data, where $w_S$ denotes the domain-specific model parameters. Subsequently, we use $w_C$ as a fixed layer and fine-tune the model $w_{\text{target}} = (w_C, w_T)$ for the target domain using the target data. Let $\ell(w_C, w_{\text{domain}}, Z)$ be the loss function where $w_{\text{domain}}$ denotes the domain specific model parameters. In this case, the learning procedures are described as follows: 
\begin{itemize}
    \item We initially train $w_{\text{source}}$ by ERM with the source data:
    \begin{align}
        w_C, w_S =  \argmin_{w_C, w_S} \frac{1}{n_s} \sum_{i=1}^{n_s} \ell(w_C, w_S, Z'_i) 
    \end{align}
    \item Given the trained $w_C$, we next train $w_T$ by ERM with the target data:
    \begin{align}
         w_T =  \argmin_{w_T} \frac{1}{n_t} \sum_{i=1}^{n_t} \ell(w_C, w_T, Z'_i) 
    \end{align}
    \item The output $w_T$ and $w_C$ will be tested with the data from the target distribution.
\end{itemize}
With the above definitions, we can then define the generalization error on the target domain as:
\begin{align}
    \text{gen}(w_T, w_C, S') = &\frac{1}{n_t} \sum_{i=1}^{n_t} \ell(w_C, w_T, Z'_i) \nonumber \\
    & - \mathbb{E}_{Z'_i}\left[\ell(w_C, w_T, Z'_i) \right]
\end{align}
Now we can state the theorem as follows.
\begin{figure*}[!htb]
\normalsize
\begin{align}
\left|\mathbb{E}_{W_TZ'_i|w_C}\left[\ell(W_T, w_C, Z'_i)\right]\right| - \mathbb{E}_{W_T \otimes \mu'|w_C}\left[\ell(W_T, w_C, Z'_i)\right] \leq \sqrt{2r^2 D(P_{W_TZ'_i|w_C} \| P_{W_T\otimes \mu'|w_C})}. \label{eq:condition_bound}
\end{align}
\vspace*{4pt}
\hrulefill
\end{figure*}
\begin{figure*}[!htb]
\normalsize
\begin{align}
\left|\mathbb{E}_{W_TW_CS'}\left[\text{gen}(W_T, W_C, S')\right]\right| &\leq \frac{1}{n_t} \sum_{i=1}^{n_t} \mathbb{E}_{W_C}\left[\sqrt{2r^2I(W_T;Z'_i|w_C)} \right] \nonumber \\
& \leq \frac{1}{n_t} \sum_{i=1}^{n_t} \left[\sqrt{2r^2I(W_T;Z'_i|W_C)}\right], \label{eq:gen_error_multi}
\end{align}
\vspace*{4pt}
\hrulefill
\end{figure*}
\begin{figure*}[!htb]
\normalsize
\begin{align}
    \mathbb{E}_{W_TW_C}\left[R_{\mu'}(W_T, W_C)\right] \leq & \frac{1}{n_t} \sum_{i=1}^{n_t} \sqrt{2r^2I(W_T;Z'_i|W_C)} \nonumber  \\
    &+  \frac{1}{n_t} \sum_{i=1}^{n_t} \left( \mathbb{E}_{W_CW_T \otimes Z'_i}\left[\ell(W_C, W_T, Z'_i)\right] - \underset{{w_T, w_C}}{\arg\min} \mathbb{E}_{Z'_i}\left[\ell(w_C, w_T, Z'_i)\right] \right) \label{eq:multi_excess_risk}
\end{align}
\vspace*{4pt}
\hrulefill
\end{figure*}

\begin{theorem}
Let $(w_C, w_S)$ be the hypothesis learned from the source data $S$. For any $w_C$, we assume that the loss function is $r^2$-subgaussian under the distribution $P_{W_T|w_C} \otimes \mu'$ where $W_T$ is the hypothesis learned from the target data $S'$ by fixing $w_C$. Then the generalization error can be bounded as:
\begin{align}
\mathbb{E}_{W_TW_CS'}\left[\gen(W_T, W_C, S')\right] & \leq \frac{1}{n_t} \sum_{i=1}^{n_t} \sqrt{2r^2I(W_T;Z'_i|W_C)}.
\end{align}
\end{theorem}

\begin{proof}
The proof of the above theorem simply follows with re-arranging the Donsker-Varadhan variational representation for the KL divergence $D(P_{W_TZ'_i|w_C} \| P_{W_T\otimes \mu'|w_C})$, and for any $w_C$ we will arrive at (\ref{eq:condition_bound}) for each $Z'_i$. Then we naturally have the generalization error bound shown in (\ref{eq:gen_error_multi}) where the second inequality follows from Jensen's inequality. Then it is also natural to derive the excess risk bound in (\ref{eq:multi_excess_risk}) since $w^* = \underset{{w_T, w_C}}{\arg\min} \mathbb{E}_{Z'_i}\left[\ell(w_C, w_T, Z'_i)\right]$.

\end{proof}
The above theorem has several implications: 
\begin{itemize}
    \item If not using $w_C$, we then don't use any source data, and we can recover the typical results in \cite{xu2017information}.
    \item If we totally depend on $w_C$, the generalization error will be zero, but the excess risk may be large due to the second difference term. 
    \item The insights from this bound is that the pre-trained model $w_C$ plays an important role in the generalization error: a good source-induced hypothesis would ensure that $w_C$ would lead to low $\underset{{w_T}}{\arg\min} \mathbb{E}_{Z'_i}\left[\ell(w_C, w_T, Z'_i)\right]$ and a good fine-tuning would ensure that the conditional mutual information $I(W_T;Z'_i|w_C)$ is small given $w_C$. 
\end{itemize}
While the outcome of the above theorem appears straightforward and intuitive, it offers only minimal direction regarding the training of $w_C$. At a glance, one might expect that a smaller domain divergence would lead $w_C$ closer to $w^*_C$. However, the manner in which the source influences the common hypothesis $w_C$ remains somewhat implicit. 

\section{Transfer learning with unlabeled source data}\label{apd:unlabel}
Unsupervised transfer learning has been extensively explored in numerous literature, especially in the medical area where the data may be difficult to get \cite{guo2018deep, niu2021distant}. In many real-world situations, there are often abundant unlabeled source data but a limited amount of labeled target data accessible for training. In this section, we will examine the generalization error within the domain of unlabeled source data, and the framework can be readily adapted to encompass both labeled source and unlabeled target data regions. Let $S' = (x'_1, x'_2,\cdots, z'_{\beta n}) \in (\mathcal{X} \times \mathcal{Y})^{\beta n}$ and each $Z'_i = (x'_i,y'_i)$ is a feature-label pair. For the source data, let $S = (X_{\beta n +1}, X_{\beta n +2}, \cdots, X_n) \in \mathcal{X}^{(1-\beta)n}$ as each instance $X_i$ is unlabeled, where $\mathcal{X}$ and $\mathcal{Y}$ are the feature and label spaces. Assume each source instance is i.i.d. drawn from $\mu_X$, and each target pair is i.i.d. drawn from $\mu'_{XY}$. For the loss function, we make the definitions for the supervised and unsupervised metrics (also leveraged from \cite{aminian2022information}):
\begin{itemize}
    \item Supervised loss function: A loss function $\ell: \mathcal{W} \times \mathcal{X} \times {Y} \rightarrow \mathbb{R}$ that measures the performance of the prediction.
    \item Unsupervised loss function: A loss function $\ell_u: \mathcal{W} \times \mathcal{X} \rightarrow \mathbb{R}$ that measures the performance on the unlabeled data.
\end{itemize}
Then we can define the empirical risk similarly as~(\ref{eq:er}):
\begin{align*}
      \hat{L}_{\alpha}(w, S, S') = &\frac{\alpha}{\beta n}\sum_{i=1}^{\beta n} \ell(w, x'_i, y'_i) \\
      &+ \frac{1-\alpha}{(1-\beta)n} \ell_u(w, x_i).
\end{align*}
The excess risk is defined as:
\begin{align*}
    L_{\mu'}(w) = \mathbb{E}_{\mu'_{XY}}\left[\ell(w, X', Y')\right].
\end{align*}
Then we define the generalization error as:
\begin{align*}
    \text{gen}(w, S, S') = \hat{L}_{\alpha}(w, S, S') - L_{\mu'}(w).
\end{align*}
Now we state the results as follows:
\begin{theorem}
Assume that the supervised loss functions $\ell(w,x,y)$ is $r^2_{l}$-subgaussian under the $\mu'_{X} \otimes \mu'_{Y|X}$ for all $w \in W$ and $\ell_u(w,x)$ is  $r^2_u$-subgaussian under marginal distribution $\mu'_X$ for all $w \in W$. The following upper bound holds on the expected generalization error:
\begin{align*}
     &| \text{gen} (w, S, S')| \leq \frac{\alpha}{\beta n}\sum_{i=1}^{\beta n} \sqrt{2r^2_l I(W;X'_i,Y'_i)} \\
     & + \frac{(1-\alpha)}{(1-\beta)n} \sum_{i=\beta n + 1}^{n} \sqrt{2r^2_u I(W;X'_i) + D(\mu_X\|\mu'_X)} \\
     &+ (1-\alpha)\mathbb{E}_{W\otimes \mu'_X}[\ell_u(W, X) - \mathbb{E}_{\mu'_{Y|X}}[\ell(W,X,Y)]].
\end{align*}
\end{theorem}
\begin{proof}
The proof simply follows the decomposition of $\text{gen}(w, S, S')$ as:
\begin{align}
    &\text{gen}(w, S, S') = \hat{L}_{\alpha}(w, S, S') - L_{\mu'}(w) \nonumber \\
    &=  \frac{\alpha}{\beta n}\sum_{i=1}^{\beta n} \left(\ell(w, x'_i, y'_i) - \mathbb{E}_{\mu'_{XY}}[\ell(w, x'_i, y'_i)] \right) \label{eq:labeled_part}\\
    &\quad + \frac{1-\alpha}{(1-\beta)n} \sum_{i=\beta n + 1}^{n}\left( \ell_u(w, x_i) - \mathbb{E}_{\mu'_X}[\ell_u(w, X_i)] \right) \label{eq:unlabeled_part} \\
    &\quad + (1-\alpha) \left(\mathbb{E}_{\mu'_X}[\ell_u(w, X)] - \mathbb{E}_{\mu'_{XY}}[\ell(w, x'_i, y'_i)]\right)
\end{align}
From the assumption, we follow similar procedures as in Corollary~\ref{coro:general_bound} as:
\begin{align*}
    &\mathbb{E}_{WX'_iY'_i}[\ell(W, X'_i, Y'_i)] - \mathbb{E}_{W\otimes\mu'_{XY}}[\ell(W, X'_i, Y'_i)] \\ 
    & \leq \sqrt{2r^2_l I(W;X'_i, Y'_i)},
\end{align*}
and we have:
\begin{align*}
    &\mathbb{E}_{WX_i}[\ell_u(W, X_i)] - \mathbb{E}_{W\otimes \mu'_X}[\ell_u(W, X_i)] \leq \\
    &\sqrt{2r^2_u I(W;X_i) + D(\mu_X\|\mu'_X)}.
\end{align*}
By taking the expectation over the generalization error, we get:
\begin{align*}
    &|\mathbb{E}_{WSS'}\left[\text{gen}(W, S, S')\right]| \leq \frac{\alpha}{\beta n}\sum_{i=1}^{\beta n} \sqrt{2r^2_l I(W;X'_i,Y'_i)} \\
    & + \frac{(1-\alpha)}{(1-\beta)n} \sum_{i=\beta n + 1}^{n} \sqrt{2r^2_u I(W;X'_i) + D(\mu_X\|\mu'_X)} \\
    &+ (1-\alpha)\mathbb{E}_{W\otimes \mu'_X}[\ell_u(W, X) - \mathbb{E}_{\mu'_{Y|X}}[\ell(W,X,Y)]].
\end{align*}
\end{proof}
The above results are very similar to Corollary~\ref{coro:gen_beta0} and Theorem~\ref{thm:exp_gen} for the supervised settings. However, there are still some new implications:
\begin{itemize}
    \item Since we do not have the labels for the source domain, the second term in the R.H.S. only contains the domain divergence of the source feature. However, we do have an extra term $|\mathbb{E}_{W\otimes \mu'_X}[\ell_u(W, X) - \mathbb{E}_{\mu'_{Y|X}}[\ell(W,X,Y)]]|$ that captures the conditional distribution shift. 
    \item The effect of the conditional distribution shift is reflected in the choice of the unsupervised loss function $\ell_u$, for any $w$ and $x$, if $\ell_u(w,x)$ can approach $\mathbb{E}_{\mu'_{Y|X}}[\ell(w,x,Y)]$, then we would have a better transfer.
    \item This now is directly related to the pseudo-labeling techniques for choosing the $\ell_u$, say if we rewrite $\ell_u(w,x) = \mathbb{E}_{\hat{\mu}_{Y|x}}[\ell(w,x,Y)]$, then we can directly view the unsupervised loss function as the average loss of $\ell(w,x,Y)$ under the approximated distribution $\hat{\mu}_{Y|x}$, which could be derived by many empirical methods such as deep learning and statistical transfer methods.
    \item Overall, if we want to achieve a small generalization error with the unlabeled target data, we may require that the algorithm would yield low mutual information, the divergence between $\mu_X$ and $\mu'_X$ should also be small, and the approximated conditional distribution $\hat{\mu}_{Y|x}$ that is used in pseudo-labeling should be close to the true conditional distribution ${\mu}'_{Y|x}$.
\end{itemize}

\section{Proof of Theorem \ref{thm:excess}} \label{proof:thm_excess}
\noindent The proof is built up on (\ref{eq:excess_decompose}): 
\begin{align}
&L_{\mu'}(w_{\ERM})-L_{\mu'}(w^*)
\leq  \gen(w_{\ERM}, S, S') \nonumber \\ &+\hat L_{\alpha}(w^*)-L_{\alpha}(w^*) \nonumber \\
&+(1-\alpha)(L_{\mu}(w^*)-L_{\mu'}(w^*)).
\end{align}
The Corollary \ref{coro:general_bound} provides an upper bound on the generalization error $\gen(w_{\ERM}, S, S')$, and the claim follows immediately as the expectation of $\hat{L}_{\alpha}(w^*) - L_{\alpha}(w^*)$ is zero. 

\section{Proof of Theorem~\ref{thm:central-transfer}}
\begin{proof} \label{proof:central-transfer}
We will build up on the decomposition of the expected excess risk:
\begin{align}
& \mathbb{E}_{W}[L_{\mu'}(W)-L_{\mu'}(w^*)] = \mathbb{E}_{W}[L_{\mu'}(W)- L_{\mu'}(w^*)] \nonumber  \\
&- \mathbb{E}_{WSS'}[\hat L_{\alpha}(W) - L_{\alpha}(w^*)]  \nonumber \\
&+ \mathbb{E}_{WSS'}[\hat L_{\alpha}(W)-\hat L_{\alpha}(w^*)], \label{eq:decompose2} 
\end{align}
with the fact that $\mathbb{E}_{SS'}[\hat{L}_{\alpha}(w^*,S,S')] = L_{\alpha}(w^*)$. For the target domain, we rewrite the excess risk for any $w$ by:
\begin{align}
    R_{\mu'}(w) &= \mathbb{E}_{\mu'}[\ell(w,Z)] - \mathbb{E}_{\mu'}[\ell(w^*,Z)] \nonumber \\ 
    &= \mathbb{E}_{S'}[\hat{{R}}(w, S')].
\end{align}
Then the expected excess risk in (\ref{eq:decompose2}) can be written as,
\begin{align}
    \mathbb{E}_{W}[L_{\mu'}(W)-L_{\mu'}(w^*)] = \mathbb{E}_{WSS'}[\mathcal{E}(W,S,S')] \nonumber \\
    + \mathbb{E}_{WSS'}[\hat{R}_{\alpha}(W,S,S')].
\end{align}
We will bound the first term by taking the expectation w.r.t. $W$ learned from $S$ and $S'$ as in (\ref{eq:eta_decompose})-(\ref{eq:eta_source}). Again, we use the variational representation of the KL divergence between two distributions $P$ and $Q$ defined over $\mathcal X$ is given as 
\begin{align}
D(P||Q)=\sup_{f}\left\{\Esub{P}{f(X)}-\log\Esub{Q}{e^{f(x)}} \right\}.
\end{align}
\begin{figure*}[!htb]
\normalsize
\begin{align}
    \mathbb{E}_{WSS'}[\mathcal{E}(W,S,S')] &= \alpha \left( \mathbb{E}_{P_W  \otimes \mu'^{\otimes \beta n}}[\hat{{R}}(W, S')] - \mathbb{E}_{WS'}[\hat{{R}}(W, S')]\right) \label{eq:eta_decompose}\\
    &\quad +(1-\alpha ) \left(\mathbb{E}_{P_W  \otimes \mu'^{\otimes (1-\beta)n}}[\hat{{R}}(W, S)] - \mathbb{E}_{WS}[\hat{{R}}(W, S)] \right) \\
    &= \frac{\alpha}{\beta n}\sum_{i=1}^{\beta n} \mathbb{E}_{P_W \otimes \mu'}[r(W,Z'_i)] - \mathbb{E}_{WZ'_i}[r(W,Z'_i)] \label{eq:eta_target}\\
    &\quad +  \frac{1-\alpha}{(1-\beta)n}\sum_{i=\beta n+1}^{n} \mathbb{E}_{P_W \otimes \mu'}[r(W,Z_i)] - \mathbb{E}_{WZ_i}[r(W,Z_i)] \label{eq:eta_source}.
\end{align} 
\vspace*{4pt}
\hrulefill
\end{figure*} 
We firstly examine the summation of the target data portion in (\ref{eq:eta_target}) for $i = 1,2,\cdots, \beta n$. Under the expected $(\eta,c)$-central condition, for any $0 < \eta' \leq \eta$, let $f(w,z'_i) = -\eta' r(w,z'_i)$, we will arrive at (\ref{eq:MI_KL}).
\begin{figure*}[!htb]
\normalsize
\begin{align}
    D(P_{WZ'_i}\|P_{W}\otimes P_{Z'_i}) &\geq \Esub{P_{WZ'_i}}{-\eta' r(W,Z'_i)} - \log \mathbb{E}_{P_{W}\otimes \mu'}[e^{-\eta'(r(W,Z'_i))}] \nonumber \\
    &= \Esub{P_{WZ'_i}}{-\eta' r(W,Z'_i)} - \log \mathbb{E}_{P_{W}\otimes \mu'}[e^{-\eta'(r(W,Z'_i) - \mathbb{E}_{P_{W}\otimes \mu'}[r(W,Z'_i)])  }] + \Esub{P_W \otimes \mu'}{\eta' r(W,Z'_i)} \nonumber \\
    &= \eta'\left(\Esub{P_W \otimes \mu'}{r(W,Z'_i)} - \Esub{P_{WZ'_i}}{r(W,Z'_i)}\right) - \log \mathbb{E}_{P_{W}\otimes \mu'}[e^{\eta'( \mathbb{E}_{P_{W}\otimes \mu'}[r(W,Z'_i)]  - r(W,Z'_i))  }]. \label{eq:MI_KL}
\end{align}
\vspace*{4pt}
\hrulefill
\end{figure*}

Next we will upper bound the second term $\log \mathbb{E}_{P_{W}\otimes \mu'}[e^{\eta'( \mathbb{E}_{P_{W}\otimes \mu}[r(W,Z'_i)]  - r(W,Z'_i))  }]$ in R.H.S. using the expected $(\eta,c)$-central condition. From the $(\eta, c)$-central condition, we have,
\begin{align}
    &\log \mathbb{E}_{P_{W}\otimes \mu'}[e^{\eta( \mathbb{E}_{P_{W}\otimes \mu'}[r(W,Z'_i)]  - r(W,Z'_i))  }] \nonumber \\
    & \leq (1-c)\eta \mathbb{E}_{P_{W}\otimes \mu'}[r(W,Z'_i)].
\end{align}
Since $\eta' \leq \eta$, Jensen's inequality yields:
\begin{align}
    &\log \mathbb{E}_{P_{W}\otimes \mu'}[e^{\eta'( \mathbb{E}_{P_{W}\otimes \mu'}[r(W,Z'_i)]  - r(W,Z'_i))  }] \nonumber \\
    &=  \log \mathbb{E}_{P_{W}\otimes \mu'}[e^{\frac{\eta'}{\eta}\eta( \mathbb{E}_{P_{W}\otimes \mu'}[r(W,Z'_i)]  - r(W,Z'_i))  }] \nonumber \\
    &\leq \log \left( \mathbb{E}_{P_{W}\otimes \mu'}[e^{\eta( \mathbb{E}_{P_{W}\otimes \mu'}[r(W,Z'_i)]  - r(W,Z'_i))  }] \right)^{\frac{\eta'}{\eta}} \nonumber  \\
    &\leq \frac{\eta'}{\eta} (1-c)\eta \mathbb{E}_{P_{W}\otimes \mu'}[r(W,Z'_i)] \nonumber \\
    &= \eta'(1-c) \mathbb{E}_{P_{W}\otimes \mu'}[r(W,Z'_i)].
    \label{eq:bound-central}
\end{align}
Substitute (\ref{eq:bound-central}) into (\ref{eq:MI_KL}), we arrive at,
\begin{align}
    I(W;Z'_i) \geq &  \eta' \left(\Esub{P_W \otimes \mu'}{r(W,Z'_i)} - \Esub{P_{WZ'_i}}{r(W,Z'_i)}\right) \nonumber \\
    & - (1-c)\eta' \mathbb{E}_{P_{W}\otimes \mu'}[r(W,Z'_i)].
\end{align}
Divide $\eta'$ on both sides, we have that
\begin{align}
    \frac{I(W;Z'_i)}{\eta'} \geq & \Esub{P_W \otimes \mu'}{r(W,Z'_i)} - \Esub{P_{WZ'_i}}{r(W,Z'_i)} \nonumber \\ 
    & - (1-c) \mathbb{E}_{P_{W}\otimes \mu'}[r(W,Z'_i)].
\end{align}
Rearrange the equation and yields,
\begin{align}
    c\Esub{P_W \otimes \mu'}{r(W,Z'_i)} \leq \Esub{P_{WZ'_i}}{r(W,Z'_i)} + \frac{I(W;Z'_i)}{\eta'}.
\end{align}
Therefore,
\begin{align}
    &\Esub{P_W \otimes \mu'}{r(W,Z'_i)} - \Esub{P_{WZ'_i}}{r(W,Z'_i)} \nonumber  \\
    & \leq  (\frac{1}{c} - 1)\left({\mathbb{E}_{P_{WZ'_i}}}[r(w,Z'_i)]\right) + \frac{I(W;Z'_i)}{c\eta'}.
\end{align}
Using the similar arguments for the source domain portion in (\ref{eq:eta_source}), we have that for any $i=\beta n +1, \cdots, n$:
\begin{align}
    &\Esub{P_W \otimes \mu'}{r(W,Z_i)} - \Esub{P_{WZ_i}}{r(W,Z_i)} \nonumber \\
    & \leq  (\frac{1}{c} - 1)\left({\mathbb{E}_{P_{WZ_i}}}[r(w,Z_i)]\right) + \frac{I(W;Z_i)+ D(\mu\|\mu')}{c\eta'}.
\end{align}
Summing up every term for $Z_i$, we end up with (\ref{eq:fast_rate_gen_proof}).
\begin{figure*}[!htb]
\normalsize
\begin{align}
    \Esub{WSS'}{\mathcal{E}(W,S,S')} \leq  &\frac{1}{c\eta'} \frac{\alpha}{\beta n} \sum_{i=1}^{\beta n} I(W;Z'_i) + \frac{1}{c\eta'} \frac{1-\alpha}{(1-\beta)n} \sum_{i=\beta n+ 1}^{n} \left(I(W;Z_i) + D(\mu\|\mu')\right) \nonumber \\
    &+ (\frac{1}{c} - 1)\left({\mathbb{E}_{WSS'}}[\alpha\hat{R}(W,S')+(1-\alpha)\hat{R}(W,S)]\right). \label{eq:fast_rate_gen_proof}
\end{align}
\vspace*{4pt}
\hrulefill
\end{figure*}

Since $\frac{1}{c} - 1 > 0$ and ${\mathbb{E}_{WSS'}}[\alpha\hat{R}(W,S')+(1-\alpha)\hat{R}(W,S)]$ will be negative for $W_\ERM$, we completes the proof for ERM by:
\begin{align}
     &\Esub{W}{R_{\mu'}(W_\ERM)} \leq  \frac{1}{c\eta'} \frac{\alpha}{\beta n} \sum_{i=1}^{\beta n} I(W_\ERM;Z'_i) \nonumber \\
     &+ \frac{1}{c\eta'} \frac{1-\alpha}{(1-\beta)n} \sum_{i=\beta n+ 1}^{n} \left(I(W_\ERM;Z_i) + D(\mu\|\mu')\right).
\end{align}
If $(\eta,c)$-central condition holds for general $\hat{W}$, following the same analysis, we arrive at,
\begin{align}
    &\Esub{\hat{W}SS'}{\mathcal{E}(\hat{W},S,S')} \leq  \frac{1}{c\eta'} \frac{\alpha}{\beta n} \sum_{i=1}^{\beta n} I(\hat{W};Z'_i) \nonumber  \\
    & + \frac{1}{c\eta'} \frac{1-\alpha}{(1-\beta)n} \sum_{i=\beta n+ 1}^{n} \left(I(\hat{W};Z_i) + D(\mu\|\mu')\right) \nonumber \\
    &+ (\frac{1}{c} - 1)\mathbb{E}_{\hat{W}SS'}\left[\hat{R}_{\alpha}(\hat{W},S,S')\right].
\end{align}
Therefore, we complete the proof by,
\begin{align}
    & \Esub{\hat{W}}{R_{\mu'}(\hat{W})} 
    \leq \frac{1}{c\eta'} \frac{\alpha}{\beta n} \sum_{i=1}^{\beta n} I(\hat{W};Z'_i) \nonumber \\
    & + \frac{1}{c\eta'} \frac{1-\alpha}{(1-\beta)n} \sum_{i=\beta n+ 1}^{n} \left(I(\hat{W};Z_i) + D(\mu\|\mu')\right) \nonumber \\
    & + \frac{1}{c} \mathbb{E}_{\hat{W}SS'}\left[\hat{R}_{\alpha}(\hat{W},S,S')\right].
\end{align}
\end{proof}

\section{Proof of Lemma~\ref{lemma:rerm}}\label{proof:lemma_rerm}
\begin{proof}
We firstly define the combined regularized loss as,
\begin{align}
    \hat{{L}}_{\textup{reg}}(w,S,S') := \hat{{L}}_{\alpha}(w,S,S') + \frac{\lambda}{n}g(w).
\end{align}
Based on Theorem~\ref{thm:central-transfer}, we can bound the excess risk for $W_{\sf{RERM}}$ with (\ref{eq:fast_excess_risk_start})-(\ref{eq:fast_excess_risk_end}) where (a) follows since $|g(w^*) - g(W_{\sf{RERM}}))| \leq B$ the expected empirical risk is negative for $W_{\ERM}$ and (b) holds due to that $W_{\sf{RERM}}$ is the minimizer of the regularized loss. 
\begin{figure*}[!htb]
\normalsize
\begin{align}
     \mathbb{E}_{W}[{R}_{\mu'}(W_{\sf{RERM}})]  \leq & \frac{1}{c} \mathbb{E}_{{WSS'}}[\hat{{R}}_{\alpha}\left(W_{\sf{RERM}}, S,S'\right)]  + \frac{1}{c\eta'} \frac{\alpha}{\beta n} \sum_{i=1}^{\beta n} I(W_{\sf{RERM}};Z'_i) \nonumber \\
     & + \frac{1}{c\eta'} \frac{1-\alpha}{(1-\beta)n} \sum_{i=\beta n +1}^{n} \left(I(W_{\sf{RERM}};Z_i) + D(\mu\|\mu')\right) \label{eq:fast_excess_risk_start} \\
     = & \frac{1}{c}  \mathbb{E}_{{WSS'}}[\hat{{L}}_{\alpha}\left(W_{\sf{RERM}}, {S}, S' \right) - \hat{{L}}_{\alpha}\left(w^*, {S}, S' \right)] + \frac{1}{c\eta'} \frac{\alpha}{\beta n} \sum_{i=1}^{\beta n} I(W_{\sf{RERM}};Z'_i) \nonumber \\
     & + \frac{1}{c\eta'} \frac{1-\alpha}{(1-\beta)n} \sum_{i=\beta n +1}^{n} \left(I(W_{\sf{RERM}};Z_i) + D(\mu\|\mu')\right) \\
     \overset{(a)}{\leq} & \frac{1}{c} \left( \mathbb{E}_{WSS'}[\hat{L}_{\textup{reg}}\left(W_{\sf{RERM}}, S,S' \right)] - \mathbb{E}_{SS'}[\hat{L}_{\textup{reg}}\left(w^*, S,S' \right)]  \right)    + \frac{1}{c\eta'} \frac{\alpha}{\beta n} \sum_{i=1}^{\beta n} I(W_{\sf{RERM}};Z'_i) \nonumber \\
     & + \frac{1}{c\eta'} \frac{1-\alpha}{(1-\beta)n} \sum_{i=\beta n +1}^{n} \left(I(W_{\sf{RERM}};Z_i) + D(\mu\|\mu')\right) + \frac{\lambda B}{cn} \\
     = & \frac{1}{c}  \mathbb{E}_{WSS'}[\hat{R}_{\textup{reg}} \left(W_{\sf{RERM}}, S,S' \right)] + \frac{\lambda B}{cn} + \frac{1}{c\eta'} \frac{\alpha}{\beta n} \sum_{i=1}^{\beta n} I(W_{\sf{RERM}};Z'_i) \nonumber \\
     & + \frac{1}{c\eta'} \frac{1-\alpha}{(1-\beta)n} \sum_{i=\beta n +1}^{n} \left(I(W_{\sf{RERM}};Z_i) + D(\mu\|\mu')\right) \\
     \overset{(b)}{\leq} & \frac{\lambda B}{cn} + \frac{1}{c\eta'} \frac{\alpha}{\beta n} \sum_{i=1}^{\beta n} I(W_{\sf{RERM}};Z'_i) + \frac{1}{c\eta'} \frac{1-\alpha}{(1-\beta)n} \sum_{i=\beta n +1}^{n} \left(I(W_{\sf{RERM}};Z_i) + D(\mu\|\mu')\right). \label{eq:fast_excess_risk_end}
\end{align}
\vspace*{4pt}
\hrulefill
\end{figure*}
\end{proof}

\section{Proof of Lemma~\ref{lemma:intermediate}} \label{apd:lemma_inter}
\begin{proof}
Firstly we examine the bounds for target instances $z'_i$ for $i = 1,2,\cdots, \beta n$. We will build upon~(\ref{eq:MI_KL}). With the $(v,c)$-central condition, for any $\epsilon \geq 0$ and any $ 0 < \eta' \leq v(\epsilon)$, the Jensen's inequality yields:
\begin{align}
&\log \mathbb{E}_{P_{W}\otimes \mu'}[e^{\eta'( \mathbb{E}_{P_{W}\otimes \mu'}[r(W,Z'_i)]  - r(W,Z'_i))  }]  \nonumber  \\
&=  \log \mathbb{E}_{P_{W}\otimes \mu'}[e^{\frac{\eta'}{v(\epsilon)}v(\epsilon)( \mathbb{E}_{P_{W}\otimes \mu'}[r(W,Z'_i)]  - r(W,Z'_i))  }] \nonumber \\
&\leq \log \left( \mathbb{E}_{P_{W}\otimes \mu'}[e^{v(\epsilon)( \mathbb{E}_{P_{W}\otimes \mu'}[r(W,Z'_i)]  - r(W,Z'_i))  }] \right)^{\frac{\eta'}{v(\epsilon)}} \nonumber \\
&\leq \frac{\eta'}{v(\epsilon)} \left( (1-c)v(\epsilon) \mathbb{E}_{P_{W}\otimes \mu'}[r(W,Z'_i)] + v(\epsilon) \epsilon \right) \nonumber \\
&= \eta'(1-c) \mathbb{E}_{P_{W}\otimes \mu'}[r(W,Z'_i)] + \eta'\epsilon. \label{eq:bound-v-central}
\end{align}
Substitute (\ref{eq:bound-v-central}) into (\ref{eq:MI_KL}), we arrive at,
\begin{align}
    & I(W;Z'_i) \geq - \eta' \epsilon + \eta' \Esub{P_W \otimes \mu'}{r(W,Z'_i)} \nonumber \\ 
    &  - \eta' \Esub{P_{WZ'_i}}{r(W,Z'_i)} - (1-c)\eta' \mathbb{E}_{P_{W}\otimes \mu'}[r(W,Z'_i)] .
\end{align}
Divide $\eta'$ on both sides, we arrive at,
\begin{align}
    \frac{I(W;Z'_i)}{\eta'} \geq & \Esub{P_W \otimes \mu'}{r(W,Z'_i)} - \Esub{P_{WZ'_i}}{r(W,Z'_i)} \nonumber \\
    & - (1-c) \mathbb{E}_{P_{W}\otimes \mu'}[r(W,Z'_i)] - \epsilon.
\end{align}
Rearrange the equation and yields,
\begin{align}
    c\Esub{P_W \otimes \mu'}{r(W,Z'_i)} \leq \Esub{P_{WZ'_i}}{r(W,Z'_i)} \nonumber  \\ 
    + \frac{I(W;Z'_i)}{\eta'} +\epsilon.
\end{align}
Therefore,
\begin{align}
    &\Esub{P_W \otimes \mu'}{r(W,Z'_i)} - \Esub{P_{WZ'_i}}{r(W,Z'_i)} \nonumber  \\ 
    & \leq  (\frac{1}{c} - 1)\left({\mathbb{E}_{P_{WZ'_i}}}[r(w,Z'_i)]\right) + \frac{I(W;Z'_i)}{c\eta'} + \frac{\epsilon}{c}.
\end{align}
Regarding the source instances, following the same procedures we have the following inequality for $z_i, i = \beta n +1, \cdots, n$:
\begin{align}
    &\Esub{P_W \otimes \mu'}{r(W,Z_i)} - \Esub{P_{WZ_i}}{r(W,Z_i)} \nonumber \\
    & \leq  (\frac{1}{c} - 1)\left({\mathbb{E}_{P_{WZ_i}}}[r(w,Z_i)]\right) \nonumber \\
    & \quad + \frac{I(W;Z_i) + D(\mu\|\mu')}{c\eta'} + \frac{\epsilon}{c}.
\end{align}
Summing up every term for $Z_i$ and dividing the summation by $n$, we end up with,
\begin{align}
     &\mathbb{E}_{W}[L_{\mu'}(W) - L_{\mu'}(w^*)] \leq  \frac{1}{c} \mathbb{E}_{{WSS'}}[\hat{{R}}_{\alpha}\left(W, S,S' \right)] \nonumber \\
     & + \frac{\alpha}{c \beta n} \sum_{i=1}^{\beta n} \left( \frac{I(W;Z_i)}{\eta'} + \epsilon \right) \nonumber \\
      & + \frac{1-\alpha}{c(1-\beta)n} \sum_{i=\beta n+ 1}^{n} \left(\frac{I(W;Z_i) +  D(\mu\|\mu')}{\eta'} + \epsilon \right).
\end{align}
In particular, if $v(\epsilon) = \epsilon^{1-\beta}$ for some $\beta \in [0,1]$, then by choosing $\eta' = v(\epsilon)$ and $\frac{I(W;Z_i)}{c\eta'} + \frac{\epsilon}{c}$ is optimized when $\epsilon = I(W;Z_i)^{\frac{1}{2-\beta}}$ for the target instances and $\epsilon = \left(I(W;Z_i) + D(\mu\|\mu')\right)^{\frac{1}{2-\beta}}$ for the source instances. Then the bound becomes,
\begin{align}
     &\mathbb{E}_{W}[L_{\mu'}(W) - L_{\mu'}(w^*)] \leq  \frac{1}{c} \mathbb{E}_{{WSS'}}[\hat{{R}}_{\alpha}\left(W, S,S' \right)] \nonumber \\ 
     & + \frac{\alpha}{c \beta n} \sum_{i=1}^{\beta n}  I(W_\ERM;Z_i)^{\frac{1}{2-\beta}} \nonumber \\
     & + \frac{1-\alpha}{c(1-\beta)n} \sum_{i=\beta n+ 1}^{n} \left(I(W;Z_i) +  D(\mu\|\mu')\right)^{\frac{1}{2-\beta}},
 \end{align}
which completes the proof.
\end{proof}

\section{Proof of Theorem \ref{thm:Tbound}} \label{proof:thm_Tbound}
\begin{proof}
\noindent The following lemma is used to prove the theorem.
\begin{lemma} 
For all t, if the noise $n(t) \sim \mathcal{N}(0,\sigma_tI_d)$, we have
\begin{equation*}
I\left(W(t) ; S | W(t-1)\right) \leq \frac{d}{2} \log \left(1+\frac{\eta_{t}^{2} K_{ST}^{2}}{d \sigma_{t}^{2}}\right),
\end{equation*}
\begin{equation*}
I\left(W(t) ; S^{\prime} | W(t-1)\right) \leq \frac{d}{2} \log \left(1+\frac{\eta_{t}^{2} K_{ST}^{2}}{d \sigma_{t}^{2}}\right).
\end{equation*}
\label{lemma:mutual_condition_bound}
\end{lemma}
\begin{proof}[Proof of Lemma~\ref{lemma:mutual_condition_bound}]
From the definition of mutual entropy
\begin{align*}
I\left(W(t) ; S | W(t-1)\right)= &h\left(W(t) | W(t-1)\right) \\
& -h\left(W(t) | W(t-1), S\right).
\end{align*}
Each term can be bounded separately. First, we have the update rule on $W(t)$:
\begin{align*}
W(t) = & W(t-1) - \eta_t(\alpha\nabla\hat{L}_\alpha(W(t-1),S^{\prime}) \\
& + (1-\alpha)\nabla\hat{L}_\alpha(W(t-1),S)) + n(t).
\end{align*}
Note that
\begin{equation*}
h\left(W(t)-W(t-1) | W(t-1)\right)=h\left(W(t) | W(t-1)\right),
\end{equation*}
since the subtraction term does not affect the entropy of a random variable. Also the perturbation $n(t)$ is independent of the gradient term, thus we can compute the upper bound of the expected squared norm of $w(t) - w(t-1)$:
\begin{align*} 
&\mathbb{E}\left[\left\| W(t)-W(t-1)\right\|_{2}^{2}\right] \\
& \leq  \eta_t^2(\alpha K_{S} + (1-\alpha) K_{T})^2 + d\sigma_t^2 \\
& \leq \eta_{t}^{2} K_{ST}^{2}+d \sigma_{t}^{2},
\end{align*}
\noindent where in the expression above, we used the assumption that $n_t \sim N(0,\sigma^2_tI_d)$. Among all random variables $X$ with a fixed expectation bound $\mathbb{E}\left\|X \right\|_{2}^{2} < A$, then the norm distribution $Y \sim N(0,\sqrt{\frac{A}{d}}I_d)$ has the largest entropy given by:
\begin{equation*}
h(Y) = d\log\left( \sqrt{2\pi e \sigma_Y^2} \right) = \frac{d}{2}\log\left( \frac{2\pi e A}{d}\right)
\end{equation*}
which indicates that:
\begin{equation*}
h\left(W(t) | W(t-1)\right) \leq \frac{d}{2} \log \left(2 \pi e \frac{\eta_{t}^{2} K_{ST}^{2}+d \sigma_{t}^{2}}{d}\right)
\end{equation*}
By entropy power inequality \cite{dembo_information_1991}, we have:
\begin{align*} 
&h\left(W(t) | W(t-1), S \right) \\
&= h\left(n_{t} + \eta_t\alpha\nabla\hat{L}_\alpha(W(t-1),S^{\prime}) | W(t-1), S\right) \\
&\geq  \frac{1}{2}\log(e^{2h(n_t)}+e^{2h\left( \eta_t\alpha\nabla\hat{L}_\alpha(W(t-1),S^{\prime}) | W(t-1), S\right)})\\
&\geq  h\left(n_{t}\right).
\end{align*} 
This leads to the following desired bound for the mutual entropy $I(W(t);S|W(t-1))$:
\begin{align*}
&h\left(W(t) | W(t-1)\right) - h\left(W(t) | S, W(t-1)\right) \\
& \leq \frac{d}{2} \log \left(2 \pi e \frac{\eta_{t}^{2} K_{ST}^{2}+d \sigma_{t}^{2}}{d}\right)- h(n_t). \\ 
\end{align*}
Similarly, we can achieve the upper bound for the mutual entropy $I(W(t); S^{\prime}|W(t-1))$:
\begin{align*} 
& h\left(W(t) | W(t-1)\right)  - h\left(W(t) | S', W(t-1)\right) \\
& \leq \frac{d}{2} \log \left(2 \pi e \frac{\eta_{t}^{2} K_{ST}^{2}+d \sigma_{t}^{2}}{d}\right)- h(n_t). \\ 
\end{align*}
Therefore, we consider the mutual information $I(W(t); S^\prime|W(t-1))$ and $I(W(t); S|W(t-1))$ with the Gaussian noise $n(t)$, e.g., $h(t) = \frac{d}{2}\log 2\pi e \sigma_t^2$, we can write
\begin{align*}
&I(W(t);S^\prime |W(t-1)) \\
&= h\left(W(t) | W(t-1)\right) - h\left(W(t) | S^{\prime}, W(t-1)\right) \\ 
& \leq \frac{d}{2} \log \left(2 \pi e \frac{\eta_{t}^{2} K_{ST}^{2}+d \sigma_{t}^{2}}{d}\right)-\frac{d}{2} \log 2 \pi e \sigma_{t}^{2} \\ 
&=\frac{d}{2} \log \frac{\eta_{t}^{2} K_{ST}^{2}+d \sigma_{t}^{2}}{d \sigma_{t}^{2}} \\ 
&=\frac{d}{2} \log \left(1+\frac{\eta_{t}^{2} K_{ST}^{2}}{d \sigma_{t}^{2}}\right).
\end{align*}
Similarly, we have:
\begin{align*}
I(W(t);S |W(t-1)) & \leq \frac{d}{2} \log \left(1+\frac{\eta_{t}^{2} K_{ST}^{2}}{d \sigma_{t}^{2}}\right).
\end{align*}
\end{proof}
\noindent With Lemma~\ref{lemma:mutual_condition_bound}, we reach the following bound by Jensen's inequality:
\begin{align}
&\mathbb{E}_{W S S^{\prime}}\left[\operatorname{gen}\left(W(T), S, S^{\prime}\right)\right] \nonumber \\
&\leq \frac{\alpha \sqrt{2 r^{2}}}{\beta n} \sum_{i=1}^{\beta n} \sqrt{I\left(W(T) ; Z'_{i}\right)} \nonumber \\
&+\frac{(1-\alpha) \sqrt{2 r^{2}}}{(1-\beta) n} \sum_{i=\beta n+1}^{n} \sqrt{\left(I\left(W(T) ; Z_{i}\right)+D\left(\mu \| \mu^{\prime}\right)\right)} \nonumber \\
&\leq \alpha \sqrt{\frac{2 r^{2}}{\beta n} I(W(T);S^{\prime})} \nonumber \\
& + (1-\alpha)\sqrt{2 r^{2}\left(\frac{I(W(T) ; S)}{(1-\beta)n}+D\left(\mu \| \mu^{\prime}\right)\right)}.
\end{align}
Let $W^T =(W(1),W(2),W(3),\cdots,W(T))$, with the characteristic of the gradient descent algorithm, we can show that
\begin{equation}
h(W(t)|W^{(t-1)},S) = h(W(t)|W(t-1),S),
\label{markov_condition}
\end{equation}
which follows from the Markov chain that $S \rightarrow W(1) \rightarrow W(2) \cdots \rightarrow W(T)$.
Using lemma~\ref{lemma:mutual_condition_bound}, both the mutual information $I(W(T);S)$ and $I(W(T);S')$ are bounded as:
\begin{align*}
&I(W(T);S) \leq I(W^T;S) \\
=& I(W(1);S|W(0)) + I(W(2);S|W(1)) \\
& + I(W(3);S|W(2),W(1)) + \cdots + I(W(T);S|W^{T-1}) \\
=& \sum^T_{t=1}I(W(t);S|W(t-1))\\
\leq&  \frac{d}{2} \sum_{t=1}^{T}\log \left(2 \pi e \frac{\eta_{t}^{2} K_{ST}^{2}+d \sigma_{t}^{2}}{d}\right)- \sum_{t=1}^{T}h(n(t)),
\end{align*}
where the first inequality follows from the Markov chain $S \rightarrow W^T$.
\end{proof}

\section{Proof of Corollary \ref{coro:convex_er_bound}} \label{proof:convex_er_bound}
\noindent We leverage the following proposition that
\begin{proposition}
Under the given assumptions, we define $\kappa = \frac{\nu}{\mathcal{L}}\in (0,1)$, setting $\eta = \frac{1}{\mathcal{L}}$, for all $T \geq 1$, we have:
\begin{align*}
&\hat{L}_{\alpha}(W(T),S,S')- \hat{L}_{\alpha}(w_\ERM,S,S')  \\ 
&\leq K_{ST}\left \| W(T) - w_\ERM \right \| \\
&\leq K_{ST}(1-\kappa)^T(\left\| W(0) - w_\ERM \right \| + \hat{A}_T),
\end{align*}
where we define $\hat{A}_T$
\begin{equation*}
\hat{A}_T := \sum_{t=1}^T(1-\kappa)^{-t}\left\| n(t) \right\|.
\end{equation*}
\label{prop:strong_convex}
\end{proposition}
We firstly claim that $\hat{L}_{\alpha}$ is $K_{ST}$-Lipschitz continuous with $K_{ST}$ bounded gradient, then the proof follows the proposition 3 in the work \cite{schmidt_convergence_2011}.
\begin{proof}(of Corollary~\ref{coro:convex_er_bound})
We firstly decompose the excess risk $L_{\mu^\prime}(W(T)) - L_{\mu^{\prime}}(w^*)$ into five fractions as follows.
\begin{align*}
&L_{\mu^\prime}(W(T)) - L_{\mu^{\prime}}(w^*) =  L_{\mu^{\prime}}(W(T)) - \hat{L}_{\alpha}(W(T), S, S') \\
& + \hat{L}_{\alpha}(W(T), S, S') - \hat{L}_{\alpha}(W_{\ERM}, S, S') \\
& + \hat{L}_{\alpha}(W_{\ERM}, S, S') - \hat{L}_{\alpha}(w^*, S, S') \\
& + \hat{L}_{\alpha}(w^*, S, S') - L_{\alpha}(w^*) + L_{\alpha}(w^*) - {L}_{\mu^{\prime}}(w^*).
\end{align*}
Following corollary~\ref{prop:strong_convex}, we have
\begin{align}
& \E{L_{\mu^{\prime}}(W(T)) - \hat{L}_{\alpha}(W(T), S, S')} \nonumber \\
&\leq \frac{\alpha \sqrt{2 r^{2}}}{\beta n} \sum_{i=1}^{\beta n} \sqrt{I\left(W(T) ; Z'_{i}\right)} \nonumber\\
& + \frac{(1-\alpha) \sqrt{2 r^{2}}}{(1-\beta) n} \sum_{i=\beta n+1}^{n} \sqrt{\left(I\left(W(T) ; Z_{i}\right)+D\left(\mu \| \mu^{\prime}\right)\right)}.
\label{eq:prof_gen_bound}   
\end{align}
Then use proposition~\ref{prop:strong_convex}, for any $W(T)$, we reach
\begin{align}
&\mathbb{E}[\hat{L}_{\alpha}(W(T), S, S') - \hat{L}_{\alpha}(W_{\ERM}, S, S')] \nonumber \\ 
& \leq  K_{ST}(1-\kappa)^T(\mathbb{E}[\left\| W(0) - W_\ERM \right \|] \nonumber \\
&\quad + \sum_{t=1}^T(1-\kappa)^{-t}\mathbb{E}[\left\| n(t) \right\|]).
\label{eq:prof_gen_risk}    
\end{align}
The remaining term $\hat{L}_{\alpha}(w^*,S,S') - L_{\alpha}(w^*) + L_{\alpha}(w^*) - \hat{L}_{\mu^{\prime}}(w^*)$ can be bounded with Theorem~\ref{thm:excess} for some $w^* \in \mathcal{W}$  that
\begin{align}
&\E{\hat{L}_{\alpha}(w^*) - L_{\alpha}(w^*) + L_{\alpha}(w^*) - \hat{L}_{\mu^{\prime}}(w^*)} \nonumber \\ 
&\leq (1-\alpha)d_{\mathcal W}(\mu, \mu').
\label{eq:prof_rest}
\end{align}
With the property $\hat{L}_{\alpha}(W_{\ERM}) - \hat{L}_{\alpha}(w^*) <0$, we combine the inequality~(\ref{eq:prof_gen_bound}), (\ref{eq:prof_gen_risk}) and (\ref{eq:prof_rest}) and claim the result.
\end{proof}


\section{Proof of Theorem \ref{thm:gibbs-solution}} \label{apd:gibbs}
\begin{proof}\label{proof:gibbs-solution}
We know the information quantity minimisation could be written as (\ref{eq:info_gibbs}), which leads to the well-known Gibbs algorithm as shown in (\ref{eq:gibbs_solution}). Hence we arrive at:
\begin{equation}
P_{W | S' = s', S=s}^{*}(\mathrm{d} w)=\frac{e^{-k \hat{L}_{\alpha}(w,s,s')} Q(\mathrm{d} w)}{\mathbb{E}_{Q}\left[e^{-k \hat{L}_{\alpha}(W,s,s')}\right]} 
\end{equation}
for each  $s \in \mathcal{Z}^{n}$.
\end{proof} 
\begin{figure*}[!htb]
\normalsize
\begin{align}
& \inf _{P_{W | S,S'}}\left(\mathbb{E}\left[\hat{L}_{\alpha}(W,S,S')\right]+\frac{1}{k} D\left(P_{W | S,S'} \| Q | P_{S,S'}\right)\right) \nonumber \\
=& \inf _{P_{W | S,S'}} \int_{\mathcal{Z}^{\beta n}} (\mu')^{\otimes \beta n} (\mathrm{d} s) \int_{\mathcal{Z}^{(1-\beta) n}} \mu^{\otimes (1-\beta) n} (\mathrm{d} s')\left(\mathbb{E}\left[\hat{L}_{\alpha}(W, S, S') | S'=s', S=s\right]+\frac{1}{k} D\left(P_{W | S'=s', S=s} \| Q\right)\right) \nonumber \\
=&  \int_{\mathcal{Z}^{\beta n}} (\mu')^{\otimes \beta n} (\mathrm{d} s') \int_{\mathcal{Z}^{(1-\beta) n}} \mu^{\otimes (1-\beta) n} (\mathrm{d} s) \inf_{P_{W | S,S'}} \left(\mathbb{E}\left[\hat{L}_{\alpha}(W, S, S') | S' = s', S=s\right]+\frac{1}{k} D\left(P_{W | S' = s', S=s} \| Q\right)\right). \label{eq:info_gibbs}
\end{align}
\vspace*{4pt}
\hrulefill
\end{figure*}

\begin{figure*}[!htb]
\normalsize
\begin{align}
    &\inf_{P_{W | S,S'}} \left(\mathbb{E}\left[\hat{L}_{\alpha}(W) | S' = s', S=s\right]+\frac{1}{k} D\left(P_{W | S' = s', S=s} \| Q\right)\right) \nonumber  \\
    =& \inf_{P_{W | S,S'}} \int P_{W|S,S'}(w) \left( k\hat{L}_{\alpha}(w, s, s')  + \log \frac{P_{W|S,S'}(w)}{Q(w)} \right) dw  \nonumber \\
    =&\inf_{P_{W | S,S'}} \int P_{W|S,S'}(w) \log \frac{P_{W|S,S'}(w)}{Q(w)e^{-k\hat{L}_{\alpha}(w, s, s')}/ \mathbb{E}_{Q}[e^{-k \hat{L}_{\alpha}(W,s,s')]}} dw. \label{eq:gibbs_solution}
\end{align}
\vspace*{4pt}
\hrulefill
\end{figure*}

\section{Proof of Corollary \ref{coro:gibbs-finite}}\label{apd:gibbs-finite}

\begin{proof}
Since the bounded loss function $\ell(w,z) \in [0,1]$ also satisfies $r^2$-subgaussian with $r^2 = \frac{1}{4}$, we have:
\begin{align*}
&|\Esub{WSS'}{\gen(W_G, S, S')}| \leq \frac{\alpha\sqrt{2r^2}}{\beta n}\sum_{i=1}^{\beta n}\sqrt{I(W_G;Z'_i)} \\
&+\frac{(1-\alpha)\sqrt{2r^2}}{(1-\beta)n}\sum_{i=\beta n+1}^{n}\sqrt{ (I(W_G;Z_i)+D(\mu||\mu'))} \\
& \leq \frac{\alpha}{\beta n}\sum_{i=1}^{\beta n}\sqrt{\frac{I(W_G;Z'_i|S'_{-i},S)}{2}} \\
&+\frac{(1-\alpha)}{(1-\beta)n}\sum_{i=\beta n+1}^{n}\sqrt{ \frac{I(W_G;Z_i|S',S_{-i})+D(\mu||\mu')}{2}},
\end{align*}
where we denote $S_{-i}$ by deleting the $i$-th element in $S$. The second inequality uses the fact that $Z_i$ and $S^{-i}$ are independent. Then using the Hoeffding inequality with the fact that the loss is bounded by $[0,1]$, we have:
\begin{align*}
    I(W_G;Z'_i|S'_{-i},S) &\leq \frac{\alpha^2k^2}{8\beta^2n^2}, \text{ for } i=1,\cdots, \beta n \\
    I(W_G;Z_i|S_{-i},S') &\leq \frac{(1-\alpha)^2k^2}{8(1-\beta)^2n^2},  \text{ for } i=\beta n + 1,\cdots, n.
\end{align*}
Finally, we arrive at,
\begin{align}
& |\Esub{WSS'}{\gen(W_G, S, S')}| 
\leq \frac{\alpha^2k}{4\beta n} \nonumber \\
& + \frac{(1-\alpha)}{(1-\beta)n} \sum_{i=\beta n +1}^{n}\sqrt{\frac{(1-\alpha)^2k^2}{16(1-\beta)^2n^2} + \frac{D(\mu \| \mu')}{2}}  \nonumber \\
&\leq \frac{\alpha^2k}{4\beta n}+ \frac{(1-\alpha)^2k}{4(1-\beta)n} +  (1-\alpha)\sqrt{\frac{D(\mu \| \mu')}{2}}, 
\end{align}
where we use $\sqrt{x+y} \leq \sqrt{x} + \sqrt{y}$ for both non-negative $x$ and $y$ in the last inequality. With the fact that $D(P^*_{W_G|S,S'} \| Q|S,S')$ is positive for any $Q$,
\begin{align}
&\mathbb{E}_{W} [L_{\mu'}(W_G)] = \mathbb{E}_{WSS'}[L_{\mu'}(W_G) - \hat{L}_{\alpha}(W_G,S,S')] \nonumber \\
&+ \mathbb{E}_{WSS'}[\hat{L}_{\alpha}(W_G,S,S')] \nonumber\\
&\leq \mathbb{E}_{WSS'}[\gen(W_G, S, S')] +  \mathbb{E}_{WSS'}[\hat{L}_{\alpha}(W_G,S,S')] \nonumber \\ 
&+ \frac{1}{k} D(P^*_{W_G|S,S'} \| Q|S,S') \nonumber \\
&\leq \mathbb{E}_{WSS'}[\gen(W_G, S, S')] + \mathbb{E}_{WSS'}[\hat{L}_{\alpha}(w^*_{st}(\alpha),S,S')] \nonumber \\
&+ \frac{1}{k} D(\delta_{w^*_{st}(\alpha)} \| Q|S,S')
\end{align}
as $P^*_{W_G|S,S'}$ is the minimizer of the regularized empirical risk~in~Equation~(\ref{eq:Gibbs}). Since $\mathcal{W}$ is finite, we have that $D(\delta_{w^*_{st}(\alpha)} \| Q|S,S') = -\log \frac{1}{Q(w^*_{st}(\alpha))}$ and $\mathbb{E}_{WSS'}[\hat{L}_{\alpha}(w^*_{st}(\alpha),S,S')] = L_{\alpha}(w^*_{st}(\alpha))$, which completes the proof.
\end{proof}

\section{Proof of Corollary \ref{coro:excess-fdiv-ERM}} \label{proof:coro-fdiv-ERM}
\begin{proof}
Suppose $\ell (W,Z_i)$ is $L_{\infty}$-norm upper bounded by $\sigma$, the $L_{\infty}$-norm of a random variable is defined as
\begin{equation*}
    \| X \|_{\infty} = \inf \{ M: P(X > M) = 0 \},
\end{equation*}
then followed by \cite[Theorem 3]{jiao_dependence_2017}, we have
\begin{equation}
    |\mathbb{E}_{P}[\ell (W, Z_i)] - \mathbb{E}_{Q}[\ell (W, Z_i)] | \leq 2\|\sigma \|_{\infty} D_{\phi}(P \| Q),
\end{equation}
where $D_{\phi}(P \| Q) = \frac{1}{2}\int|dP - dQ|$ is referred to as the $\phi$-divergence with $\phi(x) = \frac{1}{2}|x - 1|$.  If $Z'_i \sim \mu'$, $D_{\phi}(P \| Q) = D_{\phi}(P_{WZ'} || P_{W}\otimes \mu') := I_{\phi}(Z'_i;W)$. If $Z_i \sim \mu$, we have
\begin{align}
    &D_{\phi}(P \| Q) = \frac{1}{2}\int_{\mathcal{W} \times \mathcal{Z}}  \Big{|}dP_{W,Z_i} - dP_{W}d\mu' \Big{|} \nonumber \\
    &= \frac{1}{2}\int_{\mathcal{W} \times \mathcal{Z}}  \Big{|}dP_{W,Z_i} -  dP_{W}d\mu + dP_{W}d\mu - dP_{W}d\mu' \Big{|}\nonumber\\
    &\leq \frac{1}{2}\int_{\mathcal{W} \times \mathcal{Z}} \Big{|}dP_{W,Z_i} -  dP_{W}d\mu \Big{|} \nonumber \\ 
    &\quad + \frac{1}{2}\int_{\mathcal{W} \times \mathcal{Z}} \Big{|}dP_{W}d\mu - dP_{W}d\mu' \Big{|} \nonumber \\
    &= I_{\phi}(W;Z_i) + TV(\mu, \mu'),
\end{align}
where $TV(\mu, \mu') = D_{\phi}(\mu || \mu')$ denotes the total variation distance between the distribution $\mu$ and $\mu'$. By this, we can extend the mutual information measure to $\phi$-divergence.
\end{proof}

\section{Proof of Theorem~\ref{thm:gen-wd}} \label{proof:gen-wd}
\begin{proof}
We firstly look at the generalization error in terms of $P_{WSS'}$ and $P_{W} \otimes P_{S} \otimes P_{S'}$ as shown in (\ref{eq:wass1})-(\ref{eq:wass3}). Then we will be presenting the following theorem for the Wasserstein distance.
\begin{figure*}[!htb]
\normalsize
\allowdisplaybreaks
\begin{align}
& \Esub{WSS'}{\gen(W_{\ERM}, S, S')} =\Esub{WSS'}{L_{\mu'}(W_{\ERM})-\hat L_{\alpha}(W_{\ERM},S, S')} \nonumber \\
& = \int_{\mathcal{W} \times \mathcal Z}    \alpha\ell(W_{\ERM},Z') + (1-\alpha)\ell(W_{\ERM},Z)  dP_{W}d\mu' - \int_{\mathcal{W} \times \mathcal S \times \mathcal S'} \hat L_{\alpha}(W_{\ERM},S, S')  dP_{WSS'} \nonumber \\
& =  \int_{\mathcal{W} \times \mathcal Z} \alpha\ell(W_{\ERM},Z') dP_{W}d\mu' - \int_{\mathcal{W} \times \mathcal S'} \frac{\alpha}{\beta n}\sum_{i= 1}^{\beta n} \ell(W_{\ERM},Z'_i) dP_{WS'}  \nonumber \\
& \quad +  \int_{\mathcal{W} \times \mathcal Z} (1-\alpha)\ell(W_{\ERM},Z) dP_{W}d\mu' - \int_{\mathcal{W} \times \mathcal S'} \frac{1- \alpha}{(1-\beta) n}\sum_{i= \beta n + 1}^{ n} \ell(W_{\ERM},Z_i) dP_{WS}   \nonumber \\
& = \frac{\alpha}{\beta n}\sum_{i=1}^{\beta n} \left(  \int_{\mathcal{W} \times \mathcal Z} \ell(W_{\ERM},Z'_i) dP_Wd\mu' - \int_{\mathcal{W} \times \mathcal Z} \ell(W_{\ERM},Z'_i) dP_{WZ'_i} \right) \nonumber \\
& \quad + \frac{1 - \alpha}{(1-\beta) n}\sum_{i=\beta n + 1}^{n}  \left(  \int_{\mathcal{W} \times \mathcal Z} \ell(W_{\ERM},Z) dP_Wd\mu' - \int_{\mathcal{W} \times \mathcal Z} \ell(W_{\ERM},Z_i) dP_{WZ_i} \right) \nonumber \\
& = \frac{\alpha}{\beta n}\sum_{i=1}^{\beta n} \left(  \int_{\mathcal{W} \times \mathcal Z} \ell(W_{\ERM},Z'_i) dP_Wd\mu' - \int_{\mathcal{W} \times \mathcal Z} \ell(W_{\ERM},Z'_i) dP_{WZ'_i} \right) \label{eq:wass1}\\
& \quad + \frac{1 - \alpha}{(1-\beta) n}\sum_{i=\beta n + 1}^{n}  \left(  \int_{\mathcal{W} \times \mathcal Z} \ell(W_{\ERM},Z) dP_Wd\mu' - \int_{\mathcal{W} \times \mathcal Z} \ell(W_{\ERM},Z_i) dP_{W}d\mu \right) \label{eq:wass2}\\
& \quad + \frac{1 - \alpha}{(1-\beta) n}\sum_{i=\beta n + 1}^{n}  \left(  \int_{\mathcal{W} \times \mathcal Z} \ell(W_{\ERM},Z_i) dP_Wd\mu - \int_{\mathcal{W} \times \mathcal Z} \ell(W_{\ERM},Z_i) dP_{WZ_i} \right) \label{eq:wass3}
\end{align}
\vspace*{4pt}
\hrulefill
\end{figure*}

\begin{theorem}[Kantorovich-Rubinstein Duality Theorem] Let $(\mathcal{X}, {d})$ be a metric space and let $\mu,$ v denote two Radon probability measures contained in $\mathcal{P}_{d}(\mathcal{X})$. Then
\begin{align}
    \mathbb{W}_1(\mu, v)=\sup \left\{\int_{\mathcal X} {f} \mathrm{d} \mu-\int_{\mathcal X} {f} {d} v: f \in \operatorname{Lip}_{1}(\mathcal{X}, d)\right\},
\end{align}
where $\operatorname{Lip}_{1}(\mathcal{X}, d)$ denotes the collection of all 1-Lipschitz continuous functions on $\mathcal{X}$.
\end{theorem}
Using the theorem above, we can bound the generalization error using the Wasserstein distance. By assuming that the loss function is $\mathcal L$-Lipschitz for any $Z \in \mathcal Z$ and $W \in \mathcal W$, then we can bound (\ref{eq:wass1}), (\ref{eq:wass2}), (\ref{eq:wass3}) using the inequalities (\ref{eq:wass_bound1}), (\ref{eq:wass_bound2}) and (\ref{eq:wass_bound3}) respectively, which completes the proof.
\begin{figure*}[!htb]
\normalsize
\begin{align}
    \int_{\mathcal{W} \times \mathcal Z} \ell(W_{\ERM},Z'_i) dP_Wd\mu' - \int_{\mathcal{W} \times \mathcal Z} \ell(W_{\ERM},Z'_i) dP_{WZ'_i} &\leq \mathcal{L}\mathbb{E}_{\mu'}\left[\mathbb{W}_1(P_{W} , P_{W|Z'_i}) \right], \label{eq:wass_bound1}\\
    \int_{\mathcal{W} \times \mathcal Z} \ell(W_{\ERM},Z) dP_Wd\mu' - \int_{\mathcal{W} \times \mathcal Z} \ell(W_{\ERM},Z_i) dP_{W}d\mu &\leq \mathcal{L}\mathbb{W}_1(\mu', \mu),\label{eq:wass_bound2}\\
    \int_{\mathcal{W} \times \mathcal Z} \ell(W_{\ERM},Z_i) dP_Wd\mu- \int_{\mathcal{W} \times \mathcal Z} \ell(W_{\ERM},Z_i) dP_{WZ_i} &\leq \mathcal{L}\mathbb{E}_{\mu}\left[\mathbb{W}_1(P_{W} , P_{W|Z_i}) \right]. \label{eq:wass_bound3}
\end{align}
\vspace*{3pt}
\hrulefill
\end{figure*}
\end{proof}

\section{Proof of Proposition \ref{prop:wd-info-comparison}}
\label{proof:wd-info-comparison}
\begin{proof}
In the following, we denote the random variable $W$ by the ERM solution $W_{\ERM}$ for simplicity. With the fact that $\sqrt{x+y} \geq \sqrt{\frac{x}{2}} + \sqrt{\frac{y}{2}}$ for both $x, y \geq 0$, we can further lower bound the mutual information based quantity $\mathbb{B}_{\textup{Info}}$ by,
\begin{align*}
\mathbb{B}_{\textup{Info}} &= \frac{\sqrt{2r^2}}{n}\sum_{i=1}^{n}  \sqrt{ (I(W_{\ERM};Z_i)+ D(\mu||\mu'))} \\
& \geq \frac{\sqrt{r^2}}{n}\sum_{i=1}^{n} \left(\sqrt{ I(W_{\ERM};Z_i)} +\sqrt{ D(\mu||\mu')}\right).
\end{align*}
With the assumption that $P_W$ and $\mu'$ satisfy the $T_{1}(\frac{r^2}{2\mathcal{L}^2})$ transport cost inequality, we have that for $P_{W|z_i} \ll P_{W}$ with any $z_i$ and $\mu \ll \mu'$,
\begin{align*}
    \mathbb{W}_{1}(P_{W|z_i},P_{W}) &\leq \sqrt{\frac{r^2}{\mathcal{L}^2}D(P_{W|z_i} \| P_W)}, \\
    \mathbb{W}_{1}(\mu, \mu') &\leq \sqrt{\frac{r^2}{\mathcal{L}^2}D(\mu \| \mu')}. \\
\end{align*}
By Jensen's inequality, we can show that for all $i$,
\begin{align*}
\mathbb{E}_{Z_i}\left[\mathcal{L}\mathbb{W}_{1}(P_{W|Z_i},P_{W})\right] &\leq \mathbb{E}_{Z_i}\left[\sqrt{r^2 D(P_{W|Z_i}\|P_{W})} \right] \\
&\leq \sqrt{r^2 \mathbb{E}_{Z_i}\left[ D(P_{W|Z_i}\|P_{W}) \right]} \\
& = \sqrt{r^2 I(W;Z_i)},
\end{align*}
where $I(W;Z_i) = D(P_{WZ_i} \| P_W\otimes P_{Z_i}) = \mathbb{E}_{{Z_i}}[D(P_{W|Z_i}\| P_{W})]$. Similarly, we can also prove that:
\begin{align*}
    \mathcal{L}\mathbb{W}_{1}(\mu,\mu') \leq \sqrt{2r^2D(\mu\|\mu')},
\end{align*}
which completes the proof, and it naturally shows that the Wasserstein distance-based bound is tighter than the mutual information-based bound.
\end{proof}

\section{Setup of Bernoulli adaptation} \label{setup:bernoulli}
\begin{figure*}[!htb]
\normalsize
\begin{align}
\E{L_{\mu'}(W(T)-L_{\mu'}(w^*)} \leq & \alpha \sqrt{\frac{2 r^2(T)}{\beta n} \sum^{T}_{t=1}\frac{1}{2}\log\left( 1+\frac{\eta^2(T)K^2_{ST}(T)}{\sigma^2_t} \right)}  \nonumber \\
&+ (1-\alpha)\sqrt{2 r^{2}(T)\left(\frac{\sum^{T}_{t=1}\frac{1}{2}\log\left( 1+\frac{\eta^2(T)K^2_{ST}(T)}{\sigma^2_t} \right)}{(1-\beta)n}+D\left(\mu \| \mu^{\prime}\right)\right)}  \nonumber \\
&+ (1-\alpha) \sup_{w\in W^T} |\hat L(w, S)-\hat L(w, S')| \nonumber \\
&+ K_{ST}(T)\left\| w(0) - W_\ERM \right \| + K_{ST}(T)\sum_{t=1}^T\left\| n(t) \right\|, 
\end{align}
\vspace*{4pt}
\hrulefill
\end{figure*}
\begin{figure*}[!htb]
\normalsize
\begin{align}
|\Esub{WSS'}{\gen(W(T), S, S')}| \leq& \alpha \sqrt{\frac{2 r^2(T)}{\beta n} \sum^{T}_{t=1}\frac{1}{2}\log\left( 1+\frac{\eta(T)^2K^2_{ST}(T)}{\sigma^2_t} \right)}  \nonumber \\
&+ (1-\alpha)\sqrt{2 r^2(T)\left(\frac{\sum^{T}_{t=1}\frac{1}{2}\log\left( 1+\frac{\eta(T)^2K^2_{ST}(T)}{\sigma^2_t} \right)}{(1-\beta)n}+D\left(\mu \| \mu^{\prime}\right)\right)}.
\end{align}
\vspace*{4pt}
\hrulefill
\end{figure*}

As we use the binary cross-entropy loss function:
\begin{align}
\ell(w, z_i) = -(z_i\log(w) + (1-z_i)\log(1-w)),
\end{align}
the corresponding gradient is given by
\begin{align}
\nabla\ell(w, z_i) &= \frac{1-z_i}{1-w} -\frac{z_i}{w}.
\end{align}
Then the population risk is then defined as
\begin{align}
L_{\mu^{\prime}}(w(T)) &= -p^{\prime}\log(w(T)) \nonumber \\
& \quad - (1-p^{\prime})\log(1-w(T)),
\end{align}
and the corresponding empirical risk is defined as:
\begin{align}
& \hat{L}_{\alpha}\left(w(T), S, S^{\prime}\right) =\frac{\alpha}{\beta n} \sum_{i=1}^{\beta n} \ell\left(w(T), z'_{i}\right)  \nonumber \\ 
& +\frac{1-\alpha}{(1-\beta) n} \sum_{i=\beta n+1}^{n} \ell\left(w(T), z_{i}\right).
\end{align}
We notice that $\left \| \nabla\ell(w, z_i) \right \|$ and $\ell(w, z_i)$ are not bounded in our case if $w$ approaches $0$ or $1$, however, in the simulation, within finite $T$ iterations, we set the relative iterative parameters to be the maximum value among all iterations such that
\begin{align}
K_T(T) &= \max_{\begin{subarray}{l} t=1,2,\cdots,T,\\
i = 1,2,\cdots,\beta n \end{subarray}}\left \| \nabla\ell(w(t);z_i) \right \|, \\
K_S(T) &= \max_{\begin{subarray}{l} t=1,2,\cdots,T,\\
i = \beta n +1,\cdots,n \end{subarray}}\left \| \nabla\ell(w(t);z_i) \right \|, \\
K_{ST}(T) &= \alpha K_T(T) + (1-\alpha)K_S(T), \\
\eta(T) & = \frac{1}{K_{ST}(T)}, \\
r(T) & =  \max_{t=1,2,\cdots,T}\frac{|\log\frac{w(t)}{1-w(t)}|}{\sqrt{2}}.
\end{align}
The generalization error bound after $T$ iterations is given by

By setting $\kappa = 0$ loosely, within the finite iterations such that $\sum_{i=1}^{T}\|n(i)\| < \infty$, the according excess risk bound for this case is then expressed as

where $D(\mu || \mu')$ is calculated by
\begin{align*}
D(\mu || \mu') &= \sum_{i=0}^1 \mu(i)\log\frac{\mu(i)}{\mu'(i)} \\
&= (1-p)\log\frac{1-p}{1-p'} + p\log\frac{p}{p'}.
\end{align*}